\documentclass[11pt,letterpaper]{article}
\usepackage[margin=1in]{geometry}

\usepackage{amsmath,amsthm,amsfonts}
\allowdisplaybreaks[4]
\usepackage{amssymb}
\usepackage{thmtools}
\usepackage{thm-restate}
\usepackage{hyperref}
\usepackage{xspace}
\usepackage{url}
\usepackage{xcolor}
\usepackage{bm}
\usepackage[normalem]{ulem}

\usepackage[ruled,linesnumbered,noend]{algorithm2e}
\usepackage[page]{appendix}
\usepackage{enumitem}

\usepackage{booktabs}
\usepackage{tabularx}
\usepackage{makecell}

\usepackage[capitalize,noabbrev]{cleveref}
\usepackage{booktabs}
\usepackage{graphicx}
\usepackage{subcaption}
\usepackage{comment}

\hypersetup{colorlinks=true,allcolors=blue}

\theoremstyle{plain}
\newtheorem{theorem}{Theorem}[section]
\newtheorem{lemma}[theorem]{Lemma}
\newtheorem{claim}[theorem]{Claim}
\newtheorem*{claim*}{Claim}
\newtheorem{corollary}[theorem]{Corollary}

\newtheorem{fact}[theorem]{Fact}

\crefname{lemma}{lemma}{lemmas}
\Crefname{lemma}{Lemma}{Lemmas}
\Crefname{claim}{Claim}{Claims}
\Crefname{fact}{Fact}{Facts}

\theoremstyle{definition}

\newtheorem{example}[theorem]{Example}\newtheorem*{example*}{Example}

\theoremstyle{remark}

\newcommand{\ProblemName}[1]{\textsc{#1}}
\newcommand{\kMedian}{\ProblemName{$k$-Median}\xspace}
\newcommand{\kMeans}{\ProblemName{$k$-Means}\xspace}
\newcommand{\kzC}{\ProblemName{$(k,z)$-Clustering}\xspace}

\newcommand{\kCenter}{\ProblemName{$k$-Center}\xspace}

\DeclareMathOperator*{\E}{\mathbb{E}}

\renewcommand{\hat}[1]{\widehat{#1}}

\renewcommand{\tilde}[1]{\widetilde{#1}}

\def\fl{{\mathrm{fl}}}

\DeclareMathOperator{\poly}{poly}
\DeclareMathOperator{\polylog}{polylog}
\DeclareMathOperator{\OPT}{OPT}

\DeclareMathOperator{\dist}{dist}

\DeclareMathOperator*{\argmin}{argmin}

\newcommand{\cost}{\mathrm{cost}_z}

\newcommand{\R}{\mathbb{R}}
\newcommand{\Z}{\mathbb{Z}}
\newcommand{\asncost}{\sigma_z}

\newcommand{\LPufl}[1][\lambda]{\mathrm{LP}^{\mathrm{fl}}_{#1}}
\newcommand{\DPufl}[1][\lambda]{\mathrm{DP}^{\mathrm{fl}}_{#1}}

\newcommand{\Cpre}{C^{\mathrm{pre}}}

\newcommand{\Pold}{P^{\mathrm{in}}}
\newcommand{\Fold}{F^{\mathrm{in}}}

\newcommand{\Pnew}{P^{\mathrm{out}}}
\newcommand{\Fnew}{F^{\mathrm{out}}}

\newcommand{\OPTfl}[1][\lambda]{\smash{\widehat{\mathrm{OPT}}}^{\mathrm{fl}}_z({#1})}

\newcommand{\OPTflo}{\smash{\widehat{\mathrm{OPT}}}^{\mathrm{fl}}_z}

\newcommand{\OPTclF}{\smash{\mathrm{OPT}}_z^{\mathrm{cl}}}

\newcommand{\Ring}{R^{(\beta)}}
\newcommand{\Ball}[1][\beta]{B^{(#1)}}

\newcommand{\x}{x}
\newcommand{\Bx}{\bm{\x}}
\newcommand{\By}{\bm{y}}
\newcommand{\Br}{\bm{r}}
\newcommand{\Bv}{\bm{v}}

\newcommand{\xb}[1][\beta]{\x^{(#1)}}
\newcommand{\yb}[1][\beta]{y^{(#1)}}
\newcommand{\rb}[1][\beta]{r^{(#1)}}

\newcommand{\ralg}{\tilde{r}}

\newcommand{\yalg}{\tilde{y}}

\newcommand{\xlb}{\underline{\x}}
\newcommand{\xub}{\overline{\x}}
\newcommand{\xmd}{\tilde{\x}}

\newcommand{\Bxb}[1][\beta]{\Bx^{(#1)}}
\newcommand{\Byb}[1][\beta]{\By^{(#1)}}
\newcommand{\Brb}[1][\beta]{\Br^{(#1)}}
\newcommand{\Bvb}{\Bv^{(\beta)}}
\newcommand{\Bralg}{\tilde{\Br}}
\newcommand{\Byalg}{\tilde{\By}}

\newcommand{\MPbeta}{\mathrm{MP}_{\beta}}

\newcommand{\Bxlb}{\bm{\underline{\x}}}
\newcommand{\Bxub}{\bm{\overline{\x}}}
\newcommand{\Bxmd}{\bm{\tilde{\x}}}

\newcommand{\calF}{\mathcal{F}}

\newcommand{\calC}{\mathcal{C}}
\newcommand{\calM}{\mathcal{M}}

\let\epsilon\varepsilon

\SetKw{KwDownto}{downto}

\newcommand{\mr}[1]{\mbox{\footnotesize $\triangleright\;$#1}\quad} \DeclareSymbolFont{extraup}{U}{zavm}{m}{n}
\DeclareMathSymbol{\varheartsuit}{\mathalpha}{extraup}{86}
\DeclareMathSymbol{\vardiamondsuit}{\mathalpha}{extraup}{87}

\newcommand{\ones}{\mathbf{1}}

\title{Fully Scalable MPC Algorithms for Euclidean Clustering:\\
Constant Approximation in Constant Rounds}

\author{
     Artur Czumaj\thanks{University of Warwick. Email: \texttt{A.Czumaj@warwick.ac.uk}}
     \and
     Shaofeng H.-C. Jiang\thanks{Peking University. Email: \texttt{shaofeng.jiang@pku.edu.cn}
     }
     \and 
     Yaonan Jin\thanks{Hong Kong University of Science and Technology. Email: \texttt{jinyaonan@ust.hk}}
     \and 
     Jianing Lou\thanks{Peking University. Email: \texttt{loujn@pku.edu.cn}}
     \and Weicheng Wang
     \thanks{Peking University. Email: \texttt{2300013053@stu.pku.edu.cn}}
}

\begin{document}
    \begin{titlepage}
        \maketitle
       \thispagestyle{empty}
        \begin{abstract}
    We study Euclidean \kMeans and \kMedian under the
    Massively Parallel Computation (MPC) model,
    focusing on the \emph{fully scalable} setting.
    Previously, fully scalable algorithms for \kMeans and/or \kMedian
    either required $O(\log\log n\cdot\log\log\log n)$ rounds~[Cohen-Addad, Kuhn, and Parsaeian; SODA'26]
    or guaranteed only an $O(\log n)$-approximation~[Goranci, Jiang, Kiss, Kong, Qian, and Szilagyi; SODA'26].
    In this paper, we close this gap by giving the first fully scalable algorithms for both problems, which work even when $k$ is polynomially larger than the local memory, and \emph{simultaneously} achieve an $O(1)$-approximation and $O(1)$ round complexity.
    
    Technically, we devise a new variant of the MP algorithm [Mettu and Plaxton, SICOMP'03] that computes a \emph{fractional solution satisfying the Lagrangian Multiplier Preserving (LMP) property}. 
    Importantly, this LMP guarantee holds even under arbitrary distance distortions, a feature that is crucial for efficient MPC implementations. 
    We then give a novel MPC rounding algorithm that first rounds
    to a bicriteria solution with $k+O(n^{1-\epsilon})$ centers
    and then uses an \emph{iterative kernelization} idea
    to convert it into one with exactly $k$ centers.

\end{abstract}
     \end{titlepage}

    \tableofcontents
    \thispagestyle{empty}
    \newpage
    \setcounter{page}{1}

    \section{Introduction}
\label{sec:intro}

Euclidean \kMeans is a central problem in combinatorial optimization and algorithm design. Given a dataset $P \subseteq \R^d$ of $n$ points and an integer $k \ge 1$, the goal is to find a set $C \subseteq \R^d$ of at most $k$ centers minimizing
\begin{equation*}
    \mathrm{cost}(P, C) ~:=~ \sum_{p\in P} \min_{c\in C} \|p - c\|_2^2.
\end{equation*}
Euclidean \kMedian instead minimizes the sum of (unsquared) distances. Both are special cases of \kzC, which minimizes the sum of $z$-th powers of distances ($z=2$, $z=1$, respectively).

We address the computational challenges of \kMeans and \kMedian on large-scale and high-dimensional datasets, focusing on the \emph{Massively Parallel Computation (MPC)} model. 
The MPC model was introduced in~\cite{KarloffSV10} and serves as a theoretical abstraction of several practical parallel computing frameworks, including MapReduce~\cite{DeanG08}, Hadoop~\cite{0035689}, Spark~\cite{ZahariaCFSS10}, and Dryad~\cite{IsardBYBF07}.
In the MPC model \cite{KarloffSV10,DBLP:conf/isaac/GoodrichSZ11,BeameKS17,AndoniNOY14}, an input of $N$ words is distributed across machines with local memory $s=o(N)$ words each. Computation proceeds in \emph{synchronous rounds}: in each round, every machine performs arbitrary local computation and then exchanges messages with other machines, sending and receiving at most $s$ words. The output may be distributed across the machines. The primary complexity measure is the number of rounds, and total memory is a secondary one.

We aim to devise \emph{fully scalable} MPC algorithms with only constant round complexity, i.e., algorithms that work with local memory $s = N^{\varepsilon}$ for every constant $\varepsilon \in (0,1)$. Full scalability has become the gold standard and a widely pursued goal in MPC algorithm design \cite{AndoniNOY14, AndoniSSWZ18, AndoniSZ19, GhaffariU19, CzumajDP21, EpastoMMZ22, AhanchiAHKZ23, CoyC23, CzumajGJK024, ChangZ24, JayaramMNZ24, AzarmehrBJLMZ25, CzumajG0J25, Cohen-addadKP26, GoranciJKKQS26}, and achieving it has often proved to be particularly difficult. We also aim for low total memory, but this is a secondary concern: $N^{1+\varepsilon}$ total memory is typically considered acceptable. In Euclidean clustering, we count each coordinate of a point as one word, so the input size is $N = nd$, and a fully scalable algorithm must work with local memory $s = (nd)^{\varepsilon}$ for every constant $\varepsilon \in (0,1)$.

\paragraph{The Challenge of $k \gg s$.}
For clustering, full scalability is demanding mainly because of the relation between the local memory $s$ and the number $k$ of clusters. When $s$ is polynomially larger than $k$, \kMeans and \kMedian are well understood in MPC. Each machine summarizes its part of the data by a sample or a (composable) coreset of size roughly $k \cdot \polylog(n)$, these summaries are merged in $O(1)$ rounds, and the resulting small instance is solved on a single machine (see, e.g., \cite{Har-PeledM04,EIM11,BEL13,BBLM14}). Full scalability rules out this approach. Since local memory $s = (nd)^{\varepsilon}$ must be supported for \emph{every} constant $\varepsilon > 0$, any polynomial number of clusters, say $k = n^{\delta}$, eventually exceeds $s$. In the most interesting regime, $k$ is polynomially larger than $s$; in this regime, no single machine can even store a solution, and summary-based approaches break down, since any coreset must contain at least $k$ points \cite{HuangV20,Cohen-AddadLSS22,Huang0024}. As observed in \cite{BhaskaraW18,CzumajGJK024}, the techniques behind most sequential, streaming, and distributed clustering algorithms therefore do not yield fully scalable algorithms. 
Instead, the decision of which centers to open must be made by local computations on different machines,
while still meeting the \emph{global} constraint that at most $k$ centers are opened in total.
This global constraint is the crux of the problem.

This challenge has motivated a surge of interest in designing fully scalable algorithms for clustering problems~\cite{BhaskaraW18, BateniEFM21, Cohen-AddadLNSS21, Cohen-AddadMZ22, CoyCM23, CzumajGJK024, CzumajG0J25, Cohen-addadKP26, GoranciJKKQS26}. 
Yet for \kMeans and \kMedian, previous fully scalable algorithms
still had to compromise on other guarantees.
In particular, no fully scalable algorithm was known to compute
an $O(1)$-approximation with exactly $k$ centers in $O(1)$ rounds.
For algorithms that run in $O(1)$ rounds, existing results either
(i)~compute only $\polylog(n)$-approximations~\cite{Cohen-AddadLNSS21,AhanchiAHKZ23,GoranciJKKQS26};
(ii)~achieve bicriteria approximations that may use more than $k$ centers~\cite{BhaskaraW18,CzumajGJK024};
or (iii)~work only on special inputs~\cite{Cohen-AddadMZ22}.
More recently, Cohen-Addad, Kuhn, and Parsaeian~\cite{Cohen-addadKP26}
obtained an $O(1)$-approximation on general inputs
with exactly $k$ centers,
but their algorithm requires $O(\log\log n \cdot \log\log\log n)$ rounds.

\paragraph{Our Results.} 
We close this gap for \kMeans and \kMedian by giving the first fully scalable MPC algorithms
that compute an $O(1)$-approximation with exactly $k$ centers
in $O(1)$ rounds.
We state our result in \Cref{thm:clustering-solution}, assuming $d = O(\log n)$ and that the aspect ratio of the dataset is bounded by $\poly(n)$ for simplicity; see \Cref{thm:clustering} for the full statement.
Both assumptions are standard in previous work on MPC clustering
(e.g.,~\cite{Cohen-addadKP26,GoranciJKKQS26})
and are without loss of generality,
as they can be ensured by an $O(1)$-round fully scalable
preprocessing algorithm (see \Cref{sec:aspect-ratio}).

\begin{theorem}[Implied by \Cref{thm:clustering}]
    \label{thm:clustering-solution}
    For any constant $\epsilon\in(0,1)$, there exists a fully scalable MPC algorithm for \kMeans and for \kMedian that, given any $n$-point dataset in $\R^{O(\log n)}$ distributed across machines with local memory $n^{\epsilon}$, computes an $\epsilon^{-O(1/\epsilon)}$-approximation with high probability.\footnote{By ``with high probability'', we mean with probability at least $1-1/\poly(n)$, where the degree of $\poly(n)$ can be made arbitrarily large.
    }
    The algorithm runs in $O(\epsilon^{-2})$ rounds and uses $\tilde{O}(n^{1+\epsilon})$ total memory.
\end{theorem}

Our full result (\Cref{thm:clustering}) applies more generally to \kzC,
whose objective uses the power-$z$ distances for any constant $z\ge 1$,
while retaining an $O(1)$-approximation in $O(1)$ rounds
in the fully scalable MPC setting.
In addition, our algorithm works even in the setting of polylogarithmic local memory
$s\ge \poly(d\log n)=\polylog(n)$,
where it achieves a round complexity of $O(\epsilon^{-1}\log_s n)$,
matching the known lower bound~\cite{BhaskaraW18,RoughgardenVW18}
up to a constant $O(\epsilon^{-1})$ factor.

Our result in \Cref{thm:clustering-solution} significantly improves upon the previous
$O(1)$-round $O(\log n)$-approximation for \kMedian~\cite{GoranciJKKQS26}
and the $O(\log\log n\cdot\log\log\log n)$-round
$O(1)$-approximation for \kMeans~\cite{Cohen-addadKP26}.
Both approaches appear to face inherent obstacles to further improvement.
In particular, \cite{GoranciJKKQS26} relies on tree embeddings,
and thus $O(\log n)$ is likely the best possible ratio due to
a distance distortion lower bound~\cite{AlonKPW95,Bartal98}.
For \cite{Cohen-addadKP26}, the main bottleneck in round complexity
is the need to compute ruling sets,
which is a fundamental yet notoriously challenging problem in MPC.
Ruling sets have been studied extensively on general
graphs, but known MPC algorithms either require a
superconstant number of rounds~\cite{abs-1807-08745,GhaffariU19,GhaffariGJ20,0002LMPSSSU023,GilibertiP24,JiKPS25}
or are not fully scalable~\cite{CambusKPU23,HaqiZ23,GilibertiP24}.
Our new techniques bypass both of these barriers,
as we discuss in more detail in \Cref{sec:tech_overview} below.

\subsection{Technical Overview}
\label{sec:tech_overview}

In this paper, we propose a new framework for clustering that, at a high level, consists of two steps: it first computes a \emph{fractional} clustering solution via a new variant of the Mettu-Plaxton (MP) algorithm~\cite{MettuP03} with the \emph{Lagrangian Multiplier Preserving (LMP) property}, and then applies a \emph{rounding} procedure to convert it into an integral solution, i.e., a set of centers.
Importantly, our framework works directly in Euclidean space
and can be implemented mainly using approximate range queries/nearest-neighbor search,
which are standard geometric primitives that admit $O(1)$-round
MPC implementations~\cite{CzumajGJK024}.
This approach is fundamentally different from those of~\cite{Cohen-addadKP26,GoranciJKKQS26}
and bypasses their key technical barriers.
On the one hand, working directly in Euclidean space
allows us to bypass the \emph{tree embeddings} used in~\cite{GoranciJKKQS26}
and their $\Omega(\log n)$ distance distortion.
Moreover, their approach works only for \kMedian and is unlikely to extend to \kMeans,
primarily because the expected distance-distortion guarantees
of tree embeddings do not directly extend to squared distances.
In contrast, our framework is more general: it works for \kzC with any $z\ge 1$.
On the other hand, our framework avoids \emph{ruling-set computation},  the bottleneck behind the
$O(\log\log n\cdot\log\log\log n)$-round complexity
of~\cite{Cohen-addadKP26}.
In particular, this requires us to design a new rounding procedure,
since previous sequential rounding algorithms for \kMeans and/or \kMedian
typically rely on (variants of) ruling-set computation (as we discuss later).

Below, we highlight the key technical novelties of each of the two steps
of our framework.

\paragraph{A New Variant of MP: Distortion-Robust and LMP.}
The key ingredient of our fractional algorithm, which also
constitutes our main technical contribution, is a new variant
of the well-known Mettu-Plaxton (MP) algorithm~\cite{MettuP03}, which has been widely adapted
to various sublinear computational settings for
clustering~\cite{ArcherRS03, BadoiuCIS05, GehweilerLS14, CzumajJK0Y22, BhattacharyaCGL24, BhattacharyaGJQ24, CzumajGJK024}.
As usual, our variant of MP aims to solve the \emph{power-$z$ facility location problem}: a (Lagrangian) relaxation of \kzC,
but importantly, it additionally satisfies the \emph{Lagrangian Multiplier Preserving (LMP) property} for all $z \ge 1$,
even on arbitrarily distorted distances which need not form a metric or even be symmetric; see \Cref{alg:robust-MP} and \Cref{lem:FMP-distance-distortion}.
Indeed,
LMP is a widely considered property that converts approximations for power-$z$ facility location into approximations for \kzC via the \emph{bi-point rounding procedure}~\cite{JainV01, JainMS02, JainMMSV03, LiS16, ByrkaPRST17, Cohen-Addad0LS23, GowdaPST23, BhattacharyaCGL24}.
Moreover, being robust to distance distortion is a natural and useful requirement, since approximating distances is a common way to gain efficiency.
In our Euclidean MPC setting, this allows an $O(1)$-round MPC implementation via an approximate range query~\cite{CzumajGJK024} subroutine (which distorts the distance).
We emphasize that our new distortion-robust LMP variant
of the MP algorithm provides a general algorithmic framework
that may find applications in computational settings beyond MPC.

While variants of MP with the LMP property exist \cite{ArcherRS03,BhattacharyaCGL24}, their LMP guarantees may fail under even arbitrarily small distance distortions.
In fact, handling distance distortion for LMP is a general challenge
that was also identified in recent work~\cite{Cohen-addadKP26},
which designs an MPC clustering algorithm based on the very different
Jain-Vazirani framework~\cite{JainV01}.
Another limitation of existing LMP results for the MP algorithm
is that they all work only for the $z = 1$ case, i.e., \kMedian,
and extending them to any $z \ge 1$ also poses a significant technical challenge.
Our approach addresses both limitations.
Technically, for $z=1$, our algorithm is essentially
the same as that of~\cite{BhattacharyaCGL24}, but we develop
a substantially different analysis that establishes the
robustness of its LMP guarantee to distance distortions.
This new analysis also allows us to further generalize the algorithm
to arbitrary $z \ge 2$ while preserving the distortion-robust LMP property.
We provide a more detailed discussion of the technical
challenges and an overview of our approach in
\Cref{sec:tech_overview:fractional}.

\paragraph{Bypassing Ruling-Set Computation: Iterative Kernelization.}
Our second technical contribution is a new rounding algorithm, which converts any feasible fractional
solution, not necessarily one produced by our fractional MP algorithm,
into a set of exactly $k$ centers with only a constant-factor loss
in the objective value.
More importantly, this new algorithm 
can be implemented in $O(1)$ MPC rounds.
This distinguishes our algorithm from previous sequential rounding algorithms~\cite{CharikarGTS99,CharikarL12,abs-2604-06046},
which typically rely on a \emph{filtering} operation
that seems too global to implement efficiently
in the fully scalable MPC model.
Indeed, filtering is closely related to (and can often be reduced to)
\emph{ruling-set computation},
for which no $O(1)$-round fully scalable MPC algorithm is known, as discussed above.

We bypass this filtering step, or ruling-set computation more generally,
through a \emph{new rounding framework} that consists of two main parts.
First, we give a rounding algorithm that only rounds the given fractional solution
to a bicriteria solution with $k+O(n^{1-\epsilon})$ centers.
Then, building on this bicriteria algorithm, we design an
\emph{iterative kernelization} approach to obtain exactly $k$ centers.
Technically, our bicriteria rounding algorithm is based on a rounding scheme
with an expected $k$-size guarantee from~\cite{AnNS17},
together with some new ideas of our own that eventually yield
a high-probability bound of $k+O(n^{1-\epsilon})$ centers.
For the iterative kernelization, we use the technique from~\cite{abs-2507-11256},
which reduces the task of removing $O(n^{1-\epsilon})$ extra centers
to a clustering problem on a dataset
of size only $O(n^{1-\epsilon})$.
Combining this technique with our bicriteria algorithm,
we obtain a \emph{kernelization} step
that, roughly speaking, converts the clustering input into a new one
with a significantly smaller number of centers $k$.
Our iterative kernelization simply repeats this step
until $k=O(1)$, at which point the resulting problem
can be readily solved in $O(1)$ rounds.
We provide a more detailed overview of the algorithm
in \Cref{sec:tech_overview:rounding}.

We note that this overall plan is similar to that of~\cite{abs-2604-06046}
for sequential rounding, 
which starts by rounding to a bicriteria solution with $k+O(1)$ centers
and then removes the extra centers.
However, their computation of a bicriteria solution
still relies on the global filtering operation,
and their use of the pseudo-approximation reduction
of~\cite{LiS16} to remove the extra centers
is also inefficient in our setting,
since we may have a much larger number of $O(n^{1-\epsilon})$ extra centers.
We also note that similar ideas of iterative kernelization have appeared
in a previous clustering algorithm~\cite{Thorup04},
but that work considers a very different setting of sequential graph clustering,
and thus its concrete implementation differs substantially from ours.
To the best of our knowledge, our algorithm
is the first to employ iterative kernelization for rounding
in the MPC model.

\subsubsection{Challenges and Our Approach to LMP under Distance Distortion}
\label{sec:tech_overview:fractional}

We now discuss in more detail the technical challenges of achieving the LMP property in the presence of distance distortion and how our approach addresses them.

We begin with some more background.
The LMP property is defined for the \emph{power-$z$ facility location problem}.
In this problem, instead of imposing a strict bound of $k$ on the number of centers,
we are given a parameter $\lambda > 0$ as the opening cost per center,
and the goal is to find a set $F \subseteq P$ minimizing
$\lambda |F| + \cost(P, F)$,
where $\lambda |F|$ is called the \emph{opening cost}
and $\cost(P, F)$ (identical to the \kzC cost)
is called the \emph{connection cost}.
We say that a solution $F$ is an LMP $\alpha$-approximation for some $\alpha \geq 1$ if
$\cost(P, F) + \alpha \lambda |F| \le \alpha \cdot \OPT^{\fl}$,
where $\OPT^{\fl}$ denotes the optimal value of the facility location problem. This is clearly a stronger guarantee than a standard $\alpha$-approximation for facility location, which only requires that $\cost(P, F) + \lambda |F| \le \alpha \cdot \OPT^{\fl}$.
In this discussion, since we will consider only a single dataset $P$,
we henceforth omit $P$ from the cost notation,
writing, for example, $\cost(F)$ instead of $\cost(P,F)$.

\paragraph{Challenges of Preserving LMP under Distance Distortion.}
While for non-LMP approximations to power-$z$ facility location, a $\gamma$-factor distortion typically costs a factor of $O(\gamma^z)$ in the final approximation ratio,
the stronger LMP property is more fragile: even a $1+\epsilon$ error can cause the resulting solution to no longer be LMP.
Specifically, consider an LMP $\alpha$-approximate solution $F$ computed using the $(1+\epsilon)$-distorted distances, which means
\begin{equation}
\label{eqn:lmp_prime}
\widetilde{\cost}(F) + \alpha \lambda |F| \le \alpha \smash{\widetilde{\OPT}}^{\fl},
\end{equation}
where both $\widetilde{\cost}$ and $\smash{\widetilde{\OPT}}^{\fl}$ are evaluated on the $(1+\epsilon)$-distorted distances.
However, it is possible that this solution $F$ is not LMP with respect to the original distances.
Indeed, while the opening cost $\lambda |F|$ does not change with respect to distance distortion,
$\smash{\widetilde{\OPT}}^{\fl}$ and $\widetilde{\cost}(F)$ may change by a factor of, say, $1+\epsilon$, and therefore the guarantee \eqref{eqn:lmp_prime} only implies
$(1+\epsilon)\cost(F) + \alpha \lambda |F| \le (1+\epsilon)\alpha \OPT^{\fl}$.
This still leaves a gap from the LMP guarantee we actually want, i.e.,
$\cost(F) + \alpha' \lambda |F| \le \alpha' \OPT^{\fl}$,
for (finite) $\alpha' \ge 1$; in fact, this gap cannot be bridged for a generic solution $F$, as shown by the following counterexample.

\begin{example*}
    Consider $z = 1$ and $\lambda = 1 + \epsilon$ (for any given $\epsilon > 0$).
    Let the dataset consist of only two points $x, y$ with $\dist(x,y) = 1$, and distorted distance $\Tilde{\dist}(x,y) = 1+\epsilon$.
    On the original distances, the optimal value is $\OPT^{\fl} = 2 + \epsilon$.
    On the distorted distances, one optimal solution (which is clearly LMP $1$-approximate) is to open both points as facilities, i.e., $F = \{x,y\}$.
    However, on the original distances, since $\cost(F) = 0$, for any finite $\alpha' \ge 1$, we have
    \begin{equation*}
        \cost(F) + \alpha' \lambda |F| = \alpha' \cdot 2(1+\epsilon) > \alpha' \cdot (2+\epsilon) = \alpha' \cdot \OPT^{\fl}.
    \end{equation*}
    Therefore, $F$ is not an LMP solution with any finite approximation ratio on the original distances.
\end{example*}

Since LMP does not hold under distance distortions in general,
one must utilize the properties of a specific algorithm,
or a specialized way of distorting distances.
Indeed, a recent work~\cite{Cohen-addadKP26} implements the seminal Jain-Vazirani algorithm (JV)~\cite{JainV01} in MPC, but it only works through careful distance approximations (e.g., using Euclidean spanners) and highly sophisticated modifications to JV.

\paragraph{Our New Variant of MP Algorithm.}
We give a new variant of another seminal algorithm, the MP algorithm~\cite{MettuP03},
and show that it achieves our goal of working under arbitrary distance distortion
(without requiring any careful design of the distance approximation).

We first describe our new variant.
In the standard MP algorithm,
a key notion is the value $r_p$ for each point $p \in P$,
proposed by~\cite{MettuP03} (or its generalization to any $z \ge 1$ by~\cite{GehweilerLS14}),
which is the unique $r_p$ with $\sum_q [r_p-\dist(p,q)]^+ = \lambda$.
At a high level, our variant consists of two steps (see also \Cref{alg:robust-MP}):
\begin{enumerate}
    \item we first compute, for each $p \in P$, a value $0 < \rb_p \leq \lambda$,
    which is defined in the same way as $r_p$ 
    but with the power-$z$ distances scaled by a fixed parameter $\beta$;
    \item then, we open each point $p$ as a facility with opening amount $\rb_p/\lambda$.
\end{enumerate}
We refer to this algorithm as $\MPbeta$.
We remark that, for the special case $z = 1$, our variant is exactly the same as
a previous algorithm proposed by~\cite{BhattacharyaCGL24}
which also establishes the LMP property.
However, as mentioned earlier, their analysis does not handle distance distortion.
Hence, even in this special case $z = 1$, our significant contribution is a novel analysis that handles distortion.
Moreover, our algorithm and analysis work for any $z \ge 1$,
and establishing this generalization is also a nontrivial challenge that requires new ideas,
which we discuss later.

\paragraph{Overview of the LMP Analysis for $\MPbeta$ under Distance Distortion.}
We show that if we run $\MPbeta$ on distances distorted by an $O(1)$ factor,
then the resulting fractional solution, denoted by $\Bralg = (\ralg_p)_{p\in P}$,
is an $O(\beta^{-1})$-approximate LMP solution with respect to the \emph{original} distances
(see also \Cref{lem:FMP-distance-distortion}).
Below, we give an overview of the analysis, assuming $\lambda = 1$ for simplicity.

Our key insight is that the $\beta$-scaling operation in $\MPbeta$ already suffices for handling distance distortion.
Specifically, we observe that the resulting solution $\Bralg$ directly translates to the following two-sided approximation guarantee:
\begin{align}
\label{eqn:tilde_r}
    &r_p^{(\Omega(\beta))} \le \tilde{r}_p \le r_p^{(\beta)}, \qquad \forall  p\in P.
\end{align}
Hence, it suffices to prove that any such solution $\Bralg$ satisfying \eqref{eqn:tilde_r} is an LMP $O(\beta^{-1})$-approximation, which requires novel technical ideas.
For example, one might be tempted to argue that there exists $\beta'$ such that $\ralg_p = r_p^{(\beta')}$ for all $p$,
but this is unfortunately not true,
which makes it infeasible to directly apply the result for the non-distorted case (e.g.,~\cite{BhattacharyaCGL24}).
In our argument, we employ a metric reassignment argument to carefully relate the opening and connection costs of $\Bralg$ and $\Brb$, via re-distributing mass of assignments to facilities within a metric ball,
and we also establish new crucial properties of the MP algorithm to justify the existence of such reassignment, making this argument valid.
This new proof strategy also allows us to handle any $z \geq 1$,
which bypasses a fundamental obstacle in generalizing the $z = 1$ proof in~\cite{BhattacharyaCGL24} to the case $z \geq 2$,
even for non-distorted distances 
(as we discuss later).

We remark that the guarantee in~\eqref{eqn:tilde_r} is stronger than a plain $O(1)$-approximation to $\rb_p$, which is what previous work achieves~\cite{BadoiuCIS05,CzumajJK0Y22,CzumajGJK024}.
Crucially, simply having an $O(1)$-approximation to the $r_p$'s does not yield an LMP solution with a finite ratio,
similar to the issue of LMP breaking under distorted distances.

\paragraph{Proof Overview: The Guarantee in~\eqref{eqn:tilde_r} Implies the LMP Property.}
To start, we formally generalize the LMP property to the fractional solutions we consider.
A fractional solution $\Bralg$ opens each point $p$ with amount $\ralg_p$.
The opening cost is $\|\Bralg\|_1 = \sum_{p\in P} \ralg_p$
(recall that we assume $\lambda = 1$).
We define the connection cost $\cost(\Bralg)$ by 
assigning each point to its nearest facilities, in non-decreasing order of distance according to the capacities specified by $\Bralg$,
such that the total assigned amount for each point is $1$.
We say that $\Bralg$ is an LMP $\alpha$-approximation for some $\alpha \ge 1$ if
$\cost(\Bralg) + \alpha \|\Bralg\|_1 \le \alpha \OPTflo$,
where $\OPTflo$ denotes the optimal value of the fractional power-$z$ facility location problem.

We start with the case $z = 1$, for which \cite{BhattacharyaCGL24} shows that the undistorted fractional solution $\Brb$ is LMP $\beta^{-1}$-approximate for $\beta \in (0,1/4]$, i.e.,
\begin{equation}
    \label{eqn:lmp}
    \cost(\Brb) + \beta^{-1}\|\Brb\|_1 \leq \beta^{-1}\OPTflo.
\end{equation}
A natural (and ideal) plan is to show (a) the opening cost $\|\Bralg\|_1 \le \|\Brb\|_1$, which follows immediately from $\Bralg \le \Brb$ in~\eqref{eqn:tilde_r}, and (b) the connection cost $\cost(\Bralg) \le \alpha \cdot \cost(\Brb)$ for some constant $\alpha \ge 1$.
This would readily imply that $\Bralg$ is an LMP solution with ratio $\alpha\beta^{-1} = O(\beta^{-1})$.
However, for (b), we obtain only the following weaker guarantee,
whose proof we outline later using the inequality
$r_p^{(\Omega(\beta))} \le \tilde{r}_p$ in~\eqref{eqn:tilde_r}:
\begin{equation}
    \label{eqn:costty}
    \cost(\Bralg) \leq \alpha \cdot (\cost(\Brb) + \| \Brb \|_1 - \| \Bralg \|_1),
\end{equation}
i.e., with an additive term that roughly corresponds to the gap between $\|\Brb\|_1$ and $\|\Bralg \|_1$.
Luckily, this still suffices to establish the LMP property for $\Bralg$:
\begin{align*}
    &\quad~ \cost(\Bralg) + \alpha\beta^{-1} \cdot \|\Bralg\|_1 \\
    &\leq~ \alpha \cdot \cost(\Brb) + \alpha\cdot (\| \Brb\|_1 - \| \Bralg\|_1) + \alpha \beta^{-1} \cdot  \|\Bralg\|_1 \\
    &\leq~ \alpha \cdot \cost(\Brb) + \alpha\beta^{-1}\cdot (\| \Brb\|_1 - \| \Bralg\|_1) + \alpha \beta^{-1} \cdot  \|\Bralg\|_1\\
    &=~ \alpha \cdot \cost(\Brb) + \alpha\beta^{-1}\cdot  \| \Brb\|_1 \\
    &\leq~ \alpha\beta^{-1}\OPTflo,
\end{align*}
where the second inequality holds because $\| \Brb\|_1 - \| \Bralg\|_1 \ge 0$.

\paragraph{A Metric Reassignment Argument and New Properties of MP.}
To establish \eqref{eqn:costty}, we employ a metric reassignment argument.
In particular, we consider the (fractional) assignment $\Bx \in \R_{\ge 0}^{P\times P}$ that realizes $\cost(\Brb)$ (i.e., it assigns $P$ to $\Brb$ according to the nearest neighbor rule),
where $x_{p,q}$ denotes the amount assigned from $p$ to $q$.
This $\Bx$ may not be a feasible assignment with respect to $\Bralg$, since there may exist $p,q$ such that $x_{p,q} > \ralg_q$, i.e., the assigned amount exceeds the opening at $q$.
The plan is to use a reassignment argument to transform $\Bx$ into $\Bx'$ that is feasible with respect to $\Bralg$, with only a small overhead.
Fix a point $p$. For each facility $q$ such that the assigned amount $x_{p,q}$ exceeds the opening $\ralg_q$, we reassign the excess amount
$\x_{p,q} - \ralg_q \le \rb_q - \ralg_q$ (note that $\x_{p,q} \le \rb_q$ since $\Bx$ is feasible with respect to $\Brb$)
to other facilities located in the local area $B(p, O(\ralg_p))$.
This is roughly where the additive term in \eqref{eqn:costty} arises, although part of the cost also needs to be charged to $\cost(\Brb)$.

To make this plan work, we establish a new property of the MP algorithm, which guarantees that the area $B(p, O(\rb_p))$ contains a sufficiently large total opening, i.e., $\sum_{x \in B(p, O(\rb_p))} \rb_x \ge 1$ for any $\beta \in (0,1)$ and $p\in P$ (see \Cref{lem:properties-rp-localdensity}).
This further translates to a corresponding guarantee for $\Bralg$ via $\Bralg \ge \Brb[\Omega(\beta)]$ in~\eqref{eqn:tilde_r}.

\paragraph{Generalizing to $z \geq 2$: New Steps.}
Finally, we discuss how to generalize the algorithm and analysis to $z \ge 2$.
The first step is to generalize the definition of the $r_p$'s,
and we make use of a natural generalization from~\cite{GehweilerLS14, CzumajGJK024}.
While the above argument for $\Bralg$ generally works, it actually relies on one key premise: exact $\Brb$ is LMP (i.e., whether \eqref{eqn:lmp} holds), which has been proved in~\cite{BhattacharyaCGL24} for $z = 1$.
Unfortunately, this premise is not known to hold for $z \geq 2$.
Therefore, we need an additional step to prove that $\Brb$ is LMP.
To this end, we generalize the analysis of~\cite{BhattacharyaCGL24} from the $z = 1$ case to $z \ge 2$.
The key step in their approach is to define an assignment $\Bx'$ that is feasible with respect to $\Brb$, and to show that its connection cost, which upper-bounds $\cost(\Brb)$, is good enough for LMP.
However, we provide a counterexample showing that this assignment $\Bx'$
may not be feasible for $z \ge 2$
(see \Cref{example:counterexample}),
which highlights a key difference from the $z = 1$ case.
Nonetheless, we show that a slightly weaker inequality $\cost(\Brb) \le O(1)\cdot \cost(\Bx')$ still holds (\Cref{lem:feasibility}), which is sufficient to establish the LMP property of $\Brb$ for the $z \ge 2$ case.
Interestingly, this is achieved using a similar idea as in the above reassignment argument: starting from $\Bx'$, we reassign, for each facility $p$, the excess amount to a local ball $B(p, O(\rb_p))$.
The key difference is that we cannot afford to charge an additive term as in \eqref{eqn:costty}; instead, we exploit the definition of $\Bx'$ together with a new smoothness property of $\rb_p$ (\Cref{lem:properties-rp-relation2}) to remedy this.

\subsubsection{Overview of Our MPC Rounding Algorithm}
\label{sec:tech_overview:rounding}

We now give a more detailed overview of our MPC rounding
algorithm, which bypasses the filtering operation and
ruling-set computation that have been major bottlenecks
to efficient MPC implementations of previous rounding algorithms.
Given a dataset $P$ and a fractional solution $\By \in \R_{\ge 0}^{P}$ with $\|\By\|_1 = k$ (opening each $p$ with amount $y_p$ as a center), our algorithm returns an integral set $C \subseteq P$ with $|C| \le k$ and $\cost(C)\le O(1)\cdot \cost(\By)$.

\paragraph{Starting Point: Expected Size Guarantee.}
We first observe that relaxing the
requirement of opening \emph{exactly} $k$ centers to opening
$k$ centers \emph{in expectation} requires only
\emph{local} computation.
Indeed, there already exists a rounding algorithm with an expected $k$-size guarantee
that determines whether to open each facility based only on ``nearby'' facilities and points~\cite{AnNS17},
and hence can be simulated using geometric range queries,
which are standard $O(1)$-round primitives in MPC~\cite{CzumajGJK024}. 
However, it is still challenging to further strengthen the expected-size guarantee
to a high-probability guarantee.

\paragraph{Converting to a Bicriteria Solution.}
Our idea is to instead aim for a bicriteria solution with $k + O(n^{1-\epsilon})$ centers.
For a random center set $C$ with $\E[|C|] \le k$,
we can use Markov's inequality to deduce that 
$\Pr[|C| > k + n^{1-\epsilon}]
\le \E[|C|] / (k+n^{1-\epsilon})
\le 1 - \Omega(n^{-\epsilon})$.
Hence, by running the algorithm $\widetilde O(n^\epsilon)$ times independently in parallel,
at least one of the resulting center sets satisfies the size guarantee
$|C|\le k + n^{1-\epsilon}$ with high probability.
This parallel repetition increases the total memory usage
by a factor of $\widetilde O(n^\epsilon)$,
which we can afford.

However, there is another issue concerning the cost guarantee.
What the algorithm of~\cite{AnNS17} actually gives is
\begin{equation}\label{eq:size-and-cost}
    \E[|C|] \le k
    \quad \text{and} \quad
    \E[\cost(C)] \le O(1)\cdot \cost(\By),
\end{equation}
where the cost is also bounded only in expectation.
Unfortunately, there may be a correlation between size and cost,
so that solutions with small size may have much larger costs.
In particular, conditioning on the event
$|C| \le k+n^{1-\epsilon}$,
which may have probability as small as $\Omega(n^{-\epsilon})$, the conditional expected cost may be larger by a factor of up to $O(n^\epsilon)$.
Consequently, simply repeating the algorithm does not necessarily guarantee
that both
$|C|\le k+n^{1-\epsilon}$ and
$\cost(C)\le O(1)\cdot\cost(\By)$
hold \emph{simultaneously} with high probability.

To tackle this issue, we introduce an additional algorithmic
step based on \emph{bi-point rounding}~\cite{JainV01}.
Specifically, despite the potential correlation between
the size and cost of each solution,
we show that, with high probability, there exists a \emph{bi-point solution} (i.e., a convex combination of two center sets)
among $t = n^{O(\epsilon)}$ independently
rounded solutions $C_1,\ldots,C_t$.
To see this, we examine the \emph{average} size
$\frac{1}{t}\sum_{i=1}^t |C_i|$ and the \emph{average} cost
$\frac{1}{t}\sum_{i=1}^t \cost(C_i)$.
By standard concentration analysis, the following inequalities hold simultaneously
with constant probability:
\begin{equation}\label{eq:cost-and-size-better}
    \frac{1}{t}\sum_{i=1}^t |C_i|
    \le k+O(n^{1-\epsilon})
    \quad\text{and}\quad
    \frac{1}{t}\sum_{i=1}^t \cost(C_i)
    \le O(1)\cdot\cost(\By)
\end{equation}
The success probability can then be amplified
to high probability by independent repetition.
Now, given that~\eqref{eq:cost-and-size-better} holds, a simple LP
extreme-point argument implies that there exist two
(possibly identical) solutions $C_i,C_j$ among
$C_1,\ldots,C_t$ and a coefficient $\alpha\in[0,1]$ such that
\begin{equation}
    \alpha |C_i|+(1-\alpha)|C_j|
    \le k+O(n^{1-\epsilon})
    \quad\text{and}\quad
    \alpha\cost(C_i)+(1-\alpha)\cost(C_j)
    \le O(1)\cdot\cost(\By).
\end{equation}
That is, we obtain a bi-point solution $(C_i,C_j,\alpha)$.
We then apply a bi-point rounding algorithm to obtain
a center set $C$ satisfying
\[
    |C|\le k+O(n^{1-\epsilon})
    \quad\text{and}\quad
    \cost(C)\le O(1)\cdot\cost(\By),
\]
which is the desired bicriteria solution.
Here, we note that bi-point rounding is already known to admit
an $O(1)$-round implementation in the fully scalable MPC model;
see~\cite{Cohen-addadKP26}.
Therefore, the overall conversion from an algorithm with
expected guarantees to one with high-probability guarantees
only takes $O(1)$ MPC rounds.

\paragraph{Iterative Kernelization.}
So far, by combining our fractional algorithm with the
rounding procedure above, we have obtained a bicriteria
solution for \kzC with $k+O(n^{1-\epsilon})$ centers (see also \Cref{thm:bicriteria}).
The remaining task is to remove the extra
$r=O(n^{1-\epsilon})$ centers.
This task has been extensively studied in the clustering
literature (see, e.g.,~\cite{Thorup04,LiS16,BhattacharyaCGL24,
BhattacharyaCF25,Cohen-Addad0LSS25,abs-2507-11256,abs-2604-06046}),
and we build on the technique of~\cite{abs-2507-11256},
which reduces this task to solving a clustering problem on
a significantly smaller point set of size $O(r)=O(n^{1-\epsilon})$.
By combining our bicriteria algorithm with this technique,
we obtain a \emph{kernelization} step that, roughly speaking,
converts the clustering input into a new one with
the number of centers $k$ reduced by a factor of $n^\epsilon$,\footnote{More precisely, in our algorithm in \Cref{sec:clustering},
we reduce the size of a \emph{candidate center set}
(from which all centers are chosen), which serves as an upper bound on $k$.
}
and a solution to the new input can be transformed into
a solution to the original one with a constant
$\poly(1/\epsilon)$-factor loss in the approximation ratio.
More importantly, this reduction mainly relies on approximate nearest neighbor search
(as well as some basic operations such as sorting),
and thus admits an $O(1)$-round MPC implementation; 
more concretely, $O(\log_s n)=O(\epsilon^{-1})$ rounds
in the fully scalable setting with $s=n^\epsilon$
for any constant $\epsilon\in(0,1)$.

By iteratively applying this kernelization
$O(\epsilon^{-1})$ times,
we eventually only need to solve clustering with a constant $k$,
which can be handled readily in constant rounds.
We lift this solution back through the sequence of
$O(\epsilon^{-1})$ kernelization steps for the original dataset.
Since each kernelization step increases the approximation ratio
by at most a $\poly(1/\epsilon)$ factor,
our final solution is an $\epsilon^{-O(1/\epsilon)}$-approximation.
Moreover, each kernelization takes $O(\epsilon^{-1})$ rounds,
and we have only $O(\epsilon^{-1})$ kernelizations.
This gives an overall round complexity of $O(\epsilon^{-2})$.
Both the approximation ratio and the round complexity
are thus constant for fixed $\epsilon$.

\subsection{Related Work}
\label{sec:related}

Fully scalable MPC algorithms have also been studied for fundamental problems in Euclidean spaces other than \kMeans and \kMedian. 
In the context of clustering, closely related problems include \kCenter~\cite{BateniEFM21,CoyCM23,CzumajG0J25} and \textsc{Facility Location}~\cite{CzumajGJK024}.
To date, the best known fully scalable algorithm for \kCenter achieves an $O(\log n /\log\log n)$-approximation in a constant number of rounds~\cite{CzumajG0J25}, and in low dimensions, this approximation ratio can be further improved to $2+\varepsilon$ (or $1+\varepsilon$ for a bicriteria solution).
For \textsc{Facility Location}, the streaming algorithm of~\cite{CzumajJK0Y22} can be directly translated into an $O(1)$-round MPC algorithm with an $O(1)$-approximation, but this only applies to estimating the optimal value. 
More recently, \cite{CzumajGJK024} resolved the solution version of \textsc{Facility Location} by achieving the same $O(1)$-approximation in $O(1)$ rounds.

Beyond clustering, \textsc{Minimum Spanning Tree (MST)} has attracted considerable interest in Euclidean space.
For this problem, there is a notable gap between the value version and the solution version. 
For the value version, using the streaming algorithm that achieves an $O(1/\varepsilon^2)$-approximation with space $n^{\varepsilon}$~\cite{ChenCJLW23}, one can already obtain an $O(1)$-approximation in $O(1)$ rounds.
In contrast, it remains open to find an $O(1)$-approximate solution for Euclidean MST in $O(1)$ rounds. Indeed, the currently best-known algorithms either find an $O(1)$-approximate MST in $\tilde{O}(\log\log n)$ rounds~\cite{JayaramMNZ24}, or achieve a $\polylog(n)$-approximation in $O(1)$ rounds~\cite{AhanchiAHKZ23}. In low dimensions, e.g., $d = O(1)$, one can find an $O(1)$-approximation in $O(1)$ rounds~\cite{AndoniNOY14}.
Other fundamental problems studied in this context include the \textsc{Earth Mover Distance}~\cite{AndoniNOY14} and \textsc{Max-Cut}~\cite{MenandW26}.

Beyond the fully scalable regime, constant-round MPC algorithms for \kMedian and \kMeans have long been known when the local memory may grow with $k$, e.g., via sampling~\cite{EIM11}, parallel $k$-means++ seeding~\cite{BMVKV12}, or composable coresets~\cite{BEL13,BBLM14}; all of them keep a summary of size $\Omega(k)$ on a single machine. Sequentially, LMP-based primal-dual algorithms~\cite{JainV01,JainMS02,JainMMSV03} underlie many of the best known ratios; notably, Ahmadian, Norouzi-Fard, Svensson, and Ward~\cite{ANSW17} obtained a $6.357$-approximation for Euclidean \kMeans (and $2.633$ for Euclidean \kMedian) via a primal-dual algorithm that opens exactly $k$ centers, but such algorithms rely on global pruning steps that are hard to parallelize. Very recently, Bhore, Filtser, Oh, Shin, and Singh~\cite{BFOSS27} extended consistent hashing, the key primitive behind recent fully scalable geometric MPC algorithms, to doubling metrics, obtaining, among other results, a fully scalable bicriteria approximation for \kMedian and \kMeans that opens $(1+\varepsilon)k$ centers.

\section{Preliminaries}
\label{sec:prelim}

\paragraph{Notation.} 
For an integer $n\ge1$, let $[n]:=\{1,\dots,n\}$, and for integers $m\le n$, let $[m:n]:=\{m,\dots,n\}$. For $a\in\R$, let $[a]^+:=\max\{a,0\}$. Throughout, $z\ge1$ is a fixed integer, $(V,\dist)$ is a metric space (in the Euclidean setting, $V=\R^d$ and $\dist(p,q)=\|p-q\|_2$), and for a point $p$ and a finite set $Q$ we write $\dist(p,Q):=\min_{q\in Q}\dist(p,q)$. Vectors indexed by a finite set $S$ (e.g., $P$, $F$, or $P\times F$) are in boldface, $\bm{a} \in\R^S$, with entries $a_s$; vector inequalities are entrywise, for example, $\bm{a} < \bm{b}$ if and only if $a_q < b_q$ for all $q \in S$.

For a point $p\in \R^d$ and a set $Q\subseteq \R^d$, we define $\mathrm{NN}(p,Q):= \argmin_{q\in Q}\dist(p,q)$ as a \emph{nearest neighbor of $p$ in $Q$}, with ties broken arbitrarily. We say that a point $q\in Q$ is an \emph{$\alpha$-approximate nearest neighbor} (\emph{$\alpha$-ANN} for short) of $p$ in $Q$, for $\alpha\ge 1$, if $\dist(p,q)\le \alpha\cdot \dist(p,Q)$.
For a point set $Q\subseteq \R^d$, we define its \emph{aspect ratio} as $\Delta(Q) := \frac{\max_{p,q\in Q} \dist(p, q)}{\min_{p\neq q\in Q}\dist(p, q)}$.

\begin{fact}[Generalized Triangle Inequality]
\label{lem:triangle-inequality}
Let $(P,\dist)$ be a metric space.
For any $p, q, t \in P$, it holds that
$\dist^z(p, q) \le 2^{z-1} \cdot \left( \dist^z(p, t) + \dist^z(t, q) \right)$.
\end{fact}

\paragraph{Metric Clustering.}
We define a more general formulation of \kzC as follows.
For a general metric space $(V,\dist)$, the input consists of a dataset $P\subseteq V$, a candidate center set $F\subseteq V$, and an integer $1\le k\le |F|$. The goal is to find a center set $C\subseteq F$ of size at most $k$ that minimizes $\cost(P,C):=\sum_{p\in P}\dist^z(p,C)$,
where $\dist(p,C):=\min_{c\in C}\dist(p,c)$. 
We define $\OPTclF(P,F,k):=\min_{C\subseteq F:\,|C|\le k}\cost(P,C)$.
This formulation includes the Euclidean \kzC considered in \Cref{sec:intro} as the special case where $V=F=\R^d$ and $\dist$ is the Euclidean distance.

\paragraph{Fractional Solution to Clustering.}
An important notion in this paper is a \emph{fractional} solution to \kzC, defined as a vector $\By\in \R_{\ge 0}^{F}$ satisfying $1\le \|\By\|_1\le k$, where $F$ is the given candidate center set.
Each entry $y_q$ represents the fractional opening of a candidate center $q\in F$, and the condition $1\le \|\By\|_1\le k$ means that at least one and at most $k$ units of total fractional opening are placed on the candidate centers.

To define the \emph{clustering cost} of $\By$, we start by defining the \emph{fractional assignments} and the \emph{assignment cost}.
A fractional assignment is represented by a vector $\Bx = (x_{p,q})_{p \in P,q\in F} \in \R_{\ge 0}^{P \times F}$, where each entry $x_{p,q}$ represents the amount by which $p\in P$ is assigned to $q\in F$.
The \emph{assignment cost} of $\Bx$ is defined as
\begin{equation*}
    \asncost(P, \Bx) ~:=~ \sum_{p\in P}\sum_{q\in F} \x_{p,q}\cdot \dist^z(p,q). 
\end{equation*}
We say that $\Bx$ is \emph{feasible} with respect to a fractional center set $\By$, denoted by $\Bx \sim \By$, if the following constraints are satisfied.
\begin{enumerate}[label=(C\arabic*)]
    \item\label{eq:LP-cl-sum} $\forall p\in P$, $\sum_{q\in F} x_{p,q} \ge 1$, i.e., each point $p$ is assigned to a total of at least $1$ unit of centers.

    \item\label{eq:LP-cl-cover} $\forall p\in P, q\in F$, $x_{p,q} \le y_q$, i.e., the assignment amount of $p$ to $q$ is at most the opening of $q$.
\end{enumerate}
We can then define the clustering cost of a fractional center set $\By \in \R_{\ge 0}^{F}$ with $\|\By\|_1\ge 1$ as
\begin{equation*}
    \cost(P, \By) := \min_{\Bx \sim \By}   \asncost(P, \Bx).
\end{equation*}
Let $\Bx^*\sim \By$ be the optimal feasible assignment with respect to $\By$, that is, $\cost(\By) = \asncost(\Bx^*)$.
We define the individual cost of each point $p \in P$ with respect to $\By$ as
\begin{equation*}
    \cost(p, \By) ~:=~ \sum_{q\in F} x^*_{p,q}\cdot \dist^z(p, q) ~=~ \min_{\Bx_{p} \in \R_{\ge 0}^{F}: \|\Bx_{p}\|_1 \ge 1, \Bx_{p} \le \By} \sum_{q\in F} x_{p,q}\cdot \dist^z(p, q),
\end{equation*}
where the last equality follows from the observation that the Constraints~\ref{eq:LP-cl-sum} to~\ref{eq:LP-cl-cover} are independent for each data point $p$; thus, $\Bx^*$ also minimizes the assignment cost of each individual point to the fractional center set $\By$.

We remark that the notion of a fractional solution is widely used in the clustering literature, especially in works that involve linear-programming relaxations (e.g.~\cite{CharikarGTS99,CharikarL12,LiS16,Cohen-AddadEMN22,abs-2604-06046}).
Moreover, it is well known that the gap between the optimal cost of a fractional solution and that of an integral solution (i.e., a center set), referred to as the \emph{integrality gap} of \kzC, is $2^{O(z)}$ (see, e.g.,~\cite{abs-2604-06046}), as stated below.
\begin{lemma}[Integrality Gap for Clustering~\cite{abs-2604-06046}]
    \label{lem:integrality-gap}
    $2^{-O(z)}\cdot \OPTclF(P, F, k) \le \smash{\widehat{\mathrm{OPT}}}^{\mathrm{cl}}_z(P, F, k)\le \OPTclF(P, F, k)$, where  $\smash{\widehat{\mathrm{OPT}}}^{\mathrm{cl}}_z(P, F, k) := \min_{\By\in \R_{\ge 0}^{F}: 1\le \|\By\|_1\le k} \cost(P, \By)$.
\end{lemma}

In this paper, when the context is clear, we may omit $P$ from the notation and, for example, write $\cost(P,C)$, $\cost(P,\By)$, and $\asncost(P,\Bx)$ as $\cost(C)$, $\cost(\By)$, and $\asncost(\Bx)$, respectively.

     \section{Fractional MP Algorithm with Robust LMP Approximation}
\label{sec:FL}
In this section, we present a new \emph{fractional Mettu-Plaxton (MP) algorithm}, which admits a crucial \emph{Lagrangian multiplier preserving (LMP)} property~\cite{JainV01} and works even under \emph{distance distortions}.
This algorithm is a variant of the seminal MP algorithm~\cite{MettuP03}. A similar fractional MP algorithm is also considered in~\cite{BhattacharyaCGL24}, but only for the case $z=1$ and without allowing distance distortions, whereas our algorithm handles both general $z\ge 1$ and distance distortions.
This algorithm is designed for the closely related \textsc{Power-$z$ Facility Location} problem.
Below, we start by defining \textsc{Power-$z$ Facility Location} and the LMP property, then present our fractional MP algorithm, and finally discuss how to implement it efficiently in the MPC model.

\paragraph{Power-$z$ Facility Location.}
In the \textsc{Power-$z$ Facility Location} problem, the input additionally includes an opening cost $\lambda>0$, and the hard size constraint $|C|\le k$ in \kzC is replaced by a uniform opening cost $\lambda$ per center in the objective.
In this paper, we consider the special case where $P=F$, i.e., the centers/facilities are chosen from the dataset $P$, and focus on computing a fractional solution. 
Hence, the input consists of a dataset $P$ and an opening cost $\lambda$, and the goal is to find a fractional solution $\By\in \R_{\ge 0}^{P}$ that minimizes $\lambda\cdot \|\By\|_1+\cost(P,\By)$.
Let $\OPTfl[P,\lambda] := \min_{\By\in \R_{\ge 0}^{P}:\,\|\By\|_1\ge 1} \lambda\cdot \|\By\|_1 + \cost(P,\By)$
denote the optimal objective value of \textsc{Power-$z$ Facility Location} over all fractional solutions.

Throughout this section, we consider only a single dataset $P$. Therefore, for simplicity, we omit $P$ from the notation, writing $\cost(P,C)$, $\cost(P,\By)$, $\asncost(P,\Bx)$, and $\OPTfl[P,\lambda]$ as $\cost(C)$, $\cost(\By)$, $\asncost(\Bx)$, and $\OPTfl[\lambda]$, respectively.

\paragraph{Lagrangian Multiplier Preserving (LMP).}
We say that a fractional solution $\By\in \R_{\ge 0}^{P}$ (with $\|\By\|_1 \ge 1$) is an LMP $\alpha$-approximation (for some $\alpha \ge 1$) if
\begin{equation*}
    \cost(\By) ~\le~ \alpha \cdot \big(\OPTfl[\lambda] - \lambda \cdot \|\By\|_{1}\big).
\end{equation*}
An LMP approximation is a stronger notion than a standard approximation and is crucial for transforming a (fractional) solution for \textsc{Power-$z$ Facility Location} into one for \kzC.

\paragraph{Robust Fractional MP Algorithm.} 
We describe our new algorithm for \textsc{Power-$z$ Facility Location} in \Cref{alg:robust-MP}, which is parameterized by a scale factor $\beta\in(0,1)$, and we therefore refer to it as $\MPbeta$.
This algorithm works for general metric spaces, where the distance function need not be Euclidean;
for simplicity, we assume that the underlying set $V=P$ is the dataset itself.

At a high level, the algorithm $\MPbeta$ computes a value $\rb_p \ge 0$ for each point $p \in P$, and then opens $p$ with amount $\rb_p / \lambda$ as a facility; the value $\rb_p$ is computed in a manner similar to~\cite{MettuP03} (specifically, a generalized version as in~\cite{GehweilerLS14,CzumajGJK024} to handle power-$z$), but after scaling the power-$z$ distances by~$\beta$.

\begin{algorithm}[ht]
    \DontPrintSemicolon
    \caption{$\MPbeta$ with $\beta\in (0, 1)$ on input $(P, \dist)$.}
    \label{alg:robust-MP}

    For each point $p\in P$, compute a value $\rb_p > 0$ such that 
    \begin{equation}\label{eq:rp}
        \sum_{q \in P} \big[\rb_{p} - \beta\cdot \dist^z(p, q)\big]^{+} = \lambda
    \end{equation}
    \label{alg:robust-MP:rp}

    \Return $\Byb:=(\rb_p / \lambda)_{p\in P}\in \R_{\ge 0}^{P}$
\end{algorithm}

Our main result of this section is the following \Cref{lem:FMP-distance-distortion}, which establishes the correctness and robustness of $\MPbeta$.
Specifically, we show that if $\MPbeta$ is instead run on input $(P,\Tilde{\dist})$, where $\Tilde{\dist}$ is a distorted distance function, then by setting the scale factor $\beta$ to be a sufficiently small value depending on the distortion factor, the algorithm $\MPbeta$ still yields an LMP approximation with respect to the original distance $\dist$.
Notably, the distorted distance $\Tilde{\dist}$ is not required to be a metric (and need not even be symmetric).

\begin{restatable}{lemma}{RobustMP}
    \label{lem:FMP-distance-distortion}
    Let $\epsilon \in (0, 1)$ be a constant and suppose that a distorted distance function $\Tilde{\dist}$ satisfies
    \begin{align}
        &\dist(p,q) ~\le~ \Tilde{\dist}(p,q) ~\le~ \epsilon^{-2} \cdot \dist(p,q), 
        &&\forall p,q \in P.
        \label{eq:distorted-distanct}
    \end{align}
    Let $\Byalg \in \R_{\ge 0}^{P}$ denote the output of $\MPbeta$ (\Cref{alg:robust-MP}) with $\beta := 2^{-z-1} \cdot \epsilon^{2z}$ on input $(P, \Tilde{\dist})$ (instead of $(P, \dist)$).
    Then $\Byalg$ is a fractional solution to \textsc{Power-$z$ Facility Location} with an LMP $2^{O(z^2)}\epsilon^{-2z}$-approximation guarantee on the original metric space, i.e., 
    \begin{equation*}
        \cost(\Byalg) ~\le~ 2^{O(z^2)}\epsilon^{-2z}\cdot \big(\OPTfl[\lambda] - \lambda\cdot \|\Byalg\|_1 \big),
    \end{equation*}
    where both $\cost(\Byalg)$ and $\OPTfl[\lambda]$ are evaluated with respect to the original distance $\dist$.
\end{restatable}

We prove \Cref{lem:FMP-distance-distortion} using \Cref{sec:properties-rp,sec:revisit-fractionalMP,sec:fractional-MP}.
Specifically, in \Cref{sec:properties-rp}, we establish several structural properties of the vector $\Brb = (\rb_p)_{p \in P}$ computed as in~\eqref{eq:rp} on the original/undistorted metric $(P, \dist)$; in \Cref{sec:fractional-MP}, we prove the correctness of $\MPbeta$ without distance distortions; and finally, in \Cref{sec:revisit-fractionalMP}, we prove the correctness of $\MPbeta$ under distance distortions, i.e., \Cref{lem:FMP-distance-distortion}.

\paragraph{MPC Implementation in Euclidean Space.}
Finally, we consider the Euclidean setting, where $P\subseteq \R^d$ and $\dist$ denotes the Euclidean distance, and design an efficient MPC algorithm based on the robust $\MPbeta$ for computing a fractional LMP approximation for \textsc{Power-$z$ Facility Location}.

\begin{corollary}
\label{thm:lmp-fl}
For any constant $\epsilon \in (0, 1)$ and integer $z \ge 1$, there is a deterministic MPC algorithm for \textsc{Power-$z$ Facility Location} in Euclidean space that,
\begin{itemize}
    \item takes as input an $n$-point dataset $P\subseteq \R^d$, an opening cost $\lambda > 0$, and an upper bound $\Delta\ge \Delta(P)$ on the aspect ratio;
    \item computes an LMP $2^{O(z^2)}\epsilon^{-2z}$-approximation for the dataset $P$.
\end{itemize}
With local memory $s\ge \poly(d\log(n\Delta))$ per machine, the algorithm runs in $O(\log_s n)$ rounds and uses $2^{O(\epsilon d)}\cdot n \cdot \poly(d\log (n\Delta))$ total memory.
\end{corollary}
\begin{proof}
    We present in \Cref{alg:ralg} a more involved (offline) algorithm than $\MPbeta$ (\Cref{alg:robust-MP}), which is more amenable to MPC implementation.
    The input to \Cref{alg:ralg} is $P \subseteq \R^d$, and it considers the Euclidean distance $\dist$.
    We show that its output is equivalent to that of $\MPbeta$ on input $(P,\Tilde{\dist})$, for some distorted distance $\Tilde{\dist}$, thereby establishing its correctness (i.e., an LMP approximation) via \Cref{lem:FMP-distance-distortion}, and finally discuss its MPC implementation.

    Throughout this proof, we assume that for every $p\neq q\in P$, $2\le \dist(p,q)\le O(\Delta)$. This is without loss of generality, since we can apply a standard scaling preprocessing step to the dataset to ensure this property in $O(\log_s n)$ rounds using $O(n\cdot \poly(d\log(n\Delta)))$ total memory (see, e.g.,~\cite{CzumajGJK024}).

    \begin{algorithm}[ht]
    \DontPrintSemicolon
    \caption{MPC-implementable (offline) variant of $\MPbeta$}
    \label{alg:ralg}

    $L \gets \Theta(\log \Delta)$ be sufficiently large\;
    \label{alg:ralg:L}

    $\beta \gets 2^{-z-1}\cdot \epsilon^{2z}$
    \tcp*{the same factor as specified in \Cref{lem:FMP-distance-distortion}}
    \label{alg:ralg:beta}

    For each point $p \in P$, compute sets $A^{(0)}_{p} \subseteq \dots \subseteq A^{(L)}_{p}$ such that
    $\{q \in P: \dist(p, q) \le \epsilon^{-(\ell-1)}\} \subseteq A^{(\ell)}_{p} \subseteq \{q \in P: \dist(p, q) \le \epsilon^{-\ell}\big\}$, $\forall \ell \in [0 : L]$\;
    \label{alg:ralg:s}

    For each point $p \in P$, compute the solution $\ralg_{p} \ge 0$ to the equation
    $\ralg_{p} + \sum_{\ell \in [L]} (|A^{(\ell)}_{p}| - |A^{(\ell - 1)}_{p}|) \cdot [\ralg_{p} - \beta \cdot \epsilon^{-\ell z}]^{+} = \lambda$\;
    \tcp{It is easy to see from \Cref{obs:A} that there exists a unique solution $\ralg_{p}$.}
    \label{alg:ralg:rapx}

    \Return $\Byalg := \Bralg / \lambda$
    \end{algorithm}

    The correctness analysis of \Cref{alg:ralg} begins with the following \Cref{obs:A-dist,obs:A}, concerning the sets $A^{(\ell)}_{p}$ computed in Line~\ref{alg:ralg:s}.

    \begin{fact}
    \label{obs:A-dist}
    $\epsilon^{-(\ell-2)} \le \dist(p, q) \le \epsilon^{-\ell}$, for any point $p \in P$, level $\ell\in [L]$, and point $q \in A^{(\ell)}_{p}\setminus A^{(\ell-1)}_{p}$.
    \end{fact}

    \begin{proof}
    This is a direct consequence of Line~\ref{alg:ralg:s}.
    \end{proof}

    \begin{fact}
    \label{obs:A}
    $\{p\} = A^{(0)}_{p} \subseteq \dots \subseteq A^{(L)}_{p} = P$ for every $p \in P$, when $L = \Theta(\log \Delta)$ is sufficiently large.
    \end{fact}

    \begin{proof}
    The inclusion $A^{(0)}_{p} \subseteq \dots \subseteq A^{(L)}_{p}$ is a direct consequence of Line~\ref{alg:ralg:s}.
    The first equality $A^{(0)}_{p} = \{p\}$ holds since $\epsilon^{-0} = 1 < 2 \le \min_{p \ne q \in P} \dist(p, q)$, and the second equality $A^{(L)}_{p} = P$ holds since a sufficiently large $L = \Theta(\log \Delta)$ leads to $\epsilon^{-(L - 1)} \ge \max_{p \ne q \in P} \dist(p, q) = O(\Delta)$.
    This finishes the proof of \Cref{obs:A}.
    \end{proof}
    By \Cref{obs:A}, the LHS of the equation in Line~\ref{alg:ralg:rapx} is continuous and increasing in $\ralg_{p}$, which guarantees the existence and uniqueness of $\ralg_p$.

    We define a distorted distance function $\Tilde{\dist}$ as follows:
    \begin{itemize}
        \item for every $p\in P$, let $\Tilde{\dist}(p,p) := 0 = \dist(p,p)$;
        \item for every $p\neq q \in P$, let $\ell \in [L]$ be such that $q \in A_p^{(\ell)}\setminus A_p^{(\ell - 1)}$, and then let $\Tilde{\dist}(p,q) := \epsilon^{-\ell}$; the existence and uniqueness of $\ell$ is guaranteed by \Cref{obs:A}.
    \end{itemize} 
    By \Cref{obs:A-dist}, we have $\dist(p,q) \le \Tilde{\dist}(p,q) \le \epsilon^{-2}\cdot \dist(p,q)$ for every $p\neq q\in P$, and therefore the distorted distance function $\Tilde{\dist}$ satisfies \eqref{eq:distorted-distanct}.
    Moreover, by \Cref{obs:A}, the definition of $\Tilde{\dist}$, and the computation of $\ralg_p$ (Line~\ref{alg:ralg:rapx}), we can deduce that 
    \begin{equation*}
        \sum_{q \in P} \big[\ralg_p - \beta \cdot {\big(\Tilde{\dist}(p,q)\big)^z}\big]^{+} = \lambda.
    \end{equation*}
    Hence, the output $\Byalg = \Bralg / \lambda$ of \Cref{alg:ralg} is identical to the output of $\MPbeta$ (\Cref{alg:robust-MP}) with $\beta = 2^{-z-1} \cdot \epsilon^{2z}$ on input $(P,\Tilde{\dist})$, and therefore \Cref{lem:FMP-distance-distortion} implies that $\Byalg$ is an LMP $2^{O(z^2)} \epsilon^{-2z}$-approximation on the original Euclidean $(P, \dist)$.

    Finally, we discuss the MPC implementation of \Cref{alg:ralg}, and it suffices to focus on the nontrivial Lines~\ref{alg:ralg:s} and~\ref{alg:ralg:rapx}.
    \begin{itemize}
        \item Line~\ref{alg:ralg:s} aims to compute the sets $A^{(\ell)}_{p}$, whose faithful implementation in MPC is inefficient.
    Nonetheless, in fact, the latter Line~\ref{alg:ralg:rapx} only requires the sizes $|A^{(\ell)}_{p}|$.
    For a single level $\ell\in [0:L]$, the sizes $\{|A^{(\ell)}_{p}|\}_{p \in P}$ can be computed deterministically using an MPC primitive from~\cite[Theorem 3.1]{CzumajGJK024}.
    By running $L + 1$ instances of this primitive---one per level $\ell \in [0:L]$---in parallel, we obtain all sizes $\{|A^{(\ell)}_{p}|\}_{p \in P, \ell \in [0:L]}$.
    This parallel invocation requires local memory $s \ge \poly(d \log n)$ and uses $O(\log_s n)$ rounds, exactly as in \cite[Theorem 3.1]{CzumajGJK024}; the total memory is $2^{O(\epsilon d)} \cdot n \cdot \poly(d \log n) \cdot \log \Delta$, which is larger by a factor of $L = \Theta(\log \Delta)$ since we need to replicate the data $L$ times for parallelization.

    \item To implement Line~\ref{alg:ralg:rapx}, we can redistribute the data so that for each point $p \in P$, the machine storing $p$ also stores all sizes $\{|A^{(\ell)}_p|\}_{\ell \in [0:L]}$; then, the value $\ralg_p$ can be computed locally, with no further communication.
    Such redistribution can be MPC-implemented via, e.g., sorting~\cite{DBLP:conf/isaac/GoodrichSZ11}, requiring local memory $s \ge L + 1 = O(\log \Delta)$ (to store all sizes $\{|A^{(\ell)}_p|\}_{\ell \in [0: L]}$ in a single machine), using total memory $O(n L) = O(n \log \Delta)$, and in $O(\log_s n)$ rounds.
    \end{itemize}

    In total, the entire MPC implementation
    requires local memory $s \ge \poly(d \log (n\Delta))$, runs in $O(\log_s n)$ rounds, and uses total memory $2^{O(\epsilon d)} \cdot n \cdot \poly(d \log (n\Delta))$.
\end{proof}

 \subsection{\texorpdfstring{Structural Properties of $\Brb$}{Structural Properties of rbeta}}
\label{sec:properties-rp}

Throughout this and the following subsections, we fix a metric space $(P, \dist)$. For every $\beta > 0$, let $\Brb = (\rb_p)_{p \in P}$ denote the vector computed by $\MPbeta$ (in Line~\ref{alg:robust-MP:rp} of \Cref{alg:robust-MP}) on $(P, \dist)$.

In this subsection, we present several structural properties 
of the vector $\Brb = (\rb_p)_{p\in P}$:
\begin{itemize}
    \item \Cref{lem:properties-rp-size} bounds the number of points in every small-radius ball centered at each $p \in P$.
    \item \Cref{lem:properties-rp-monotone,lem:properties-rp-relation1,lem:properties-rp-relation2} derive a set of smoothness properties for the $\rb_p$ values.
    \item \Cref{lem:properties-rp-localdensity} demonstrates a local density property of the $\rb_p$ values, providing lower bounds on the sum of values over the neighborhood of each point (with a sufficiently large radius).
\end{itemize}
These properties play a crucial role in the analyses later in \Cref{sec:fractional-MP,sec:revisit-fractionalMP}.

\paragraph{Additional Notations: Balls and Rings.} 
Before proceeding, we introduce some notations to better characterize the properties of $\rb_p$.
For every point $p \in P$ and any radius $r \ge 0$, we define the ($\beta$-scaled power-$z$) ball $\Ball_p(r)$ centered at $p$ with radius $r$, and the ($\beta$-scaled power-$z$) ring $\Ring_p(r)$ centered at $p$ with inner radius $r$ and outer radius $2r$, as follows.
\begin{align}
    \Ball_{p}(r) ~&:=~ \big\{q \in P:\ \beta \cdot \dist^{z}(p, q) \le r \big\}.
    \label{eq:ball-beta}\\
    \Ring_{p}(r) ~&:=~ \Ball_{p}(2r)\setminus \Ball_{p}(r) ~=~ \big\{q \in P:\ r < \beta \cdot \dist^{z}(p, q) \le 2r \big\}
    \label{eq:ring}
\end{align}
It is easy to verify that the defining equation for $\rb_{p}$ in \eqref{eq:rp} is equivalent to
\begin{align}
\label{eq:r-vector-equivalent}
    \sum_{q \in \Ball_{p}(\rb_{p})} \big(\rb_{p} - \beta \cdot \dist^{z}(p, q)\big)
    ~=~ \lambda.
\end{align}

\paragraph{Size Bounds for Balls.}
The following \Cref{lem:properties-rp-size} provides an upper bound on the size of the small-radius ball $\Ball_{p}(r)$. 
This lemma is adapted from \cite[Claim~4.4]{CzumajGJK024}, and we include its proof for completeness.

\begin{lemma}
\label{lem:properties-rp-size}
$|\Ball_{p}(r)| \le \frac{\lambda}{\rb_p - r}$, for any scale factor $\beta > 0$, point $p\in P$ and radius $0 < r < \rb_p$.
\end{lemma}

\begin{proof}
By definition, each point $q \in \Ball_{p}(r)$ satisfies $\rb_{p} - \beta \cdot \dist^{z}(p, q) \ge \rb_{p} - r>0$. This together with \eqref{eq:rp} gives
$\lambda
= \sum_{q \in P} [\rb_{p} - \beta \cdot \dist^{z}(p, q)]^{+}
\ge |\Ball_{p}(r)| \cdot (\rb_{p} - r)$. The lemma then follows directly.
\end{proof}

\paragraph{Smoothness of $\rb_{p}$.}
We establish a set of \emph{smoothness} properties for $\rb_{p}$ (\Cref{lem:properties-rp-monotone,lem:properties-rp-relation1,lem:properties-rp-relation2}).
First, \Cref{lem:properties-rp-monotone} shows that $\rb_p$ is smooth with respect to $\beta$, in the sense that when $\beta$ is scaled by a factor $c \ge 1$, the value of $\rb_p$ changes by a factor of at most $c$.

\begin{lemma}
\label{lem:properties-rp-monotone}
$(\beta' / \beta) \cdot \Brb \le \Brb[\beta'] \le \Brb$, for any scale factors $\beta \ge \beta' > 0$.
\end{lemma}

\begin{proof}
It suffices to show $\frac{\beta'}{\beta} \cdot \rb_{p} \le \rb[\beta']_{p} \le \rb_{p}$ for each point $p \in P$.
Consider the function $\Phi^{(\beta')}_{p}(r) := \sum_{t \in P} [r - \beta' \cdot \dist^{z}(p, t)]^{+}$, which is increasing in $r$ and satisfies $\Phi_{p}^{(\beta')}(\rb[\beta']_{p}) = \lambda$ by \eqref{eq:rp}; similarly for the other function $\Phi^{(\beta)}_{p}(r)$. (Note that $\lambda > 0$ is nonzero.) Then, we can verify the lemma as follows:
\begin{itemize}
    \item $\Phi^{(\beta')}_{p}(\frac{\beta'}{\beta} \cdot \rb_{p})
    = \frac{\beta'}{\beta} \cdot \Phi^{(\beta)}_{p}(\rb_{p})
    = \frac{\beta'}{\beta} \cdot \lambda
    \le \lambda
    = \Phi^{(\beta')}_{p}(\rb[\beta']_{p})
    \implies
    \frac{\beta'}{\beta} \cdot \rb_{p} \le \rb[\beta']_{p}$.
    \item $\Phi^{(\beta')}_{p}(\rb[\beta]_{p})
    \ge \Phi^{(\beta)}_{p}(\rb[\beta]_{p})
    = \lambda
    = \Phi^{(\beta')}_{p}(\rb[\beta']_{p})
    \implies
    \rb[\beta']_{p} \le \rb_{p}$.
\end{itemize}
This finishes the proof.
\end{proof}

Second, \Cref{lem:properties-rp-relation1} shows that $\rb_{p}$ is smooth with respect to the location of $p \in P$, i.e., nearby points have comparable values of $\rb_{p}$. A similar result can be found in~\cite[Claim 4.4]{CzumajGJK024}.

\begin{lemma}
\label{lem:properties-rp-relation1} 
For any scale factor $\beta > 0$ and any two points $p, q \in P$,
\begin{align*}    
    \rb_{p}
    ~\le~ r^{(2^{z-1}\beta)}_{q} + 2^{z-1} \cdot \beta \cdot \dist^{z}(p, q)
    ~\le~ 2^{z-1} \cdot \big(r^{(\beta)}_{q} + \beta \cdot \dist^{z}(p, q)\big).
\end{align*}
\end{lemma}

\begin{proof}
The second inequality, or equivalently $r^{(2^{z-1}\beta)}_{q} \le 2^{z-1} \cdot r^{(\beta)}_{q}$, follows directly from \Cref{lem:properties-rp-monotone}. Below, we prove the first inequality.
By the generalized triangle inequality (\Cref{lem:triangle-inequality}), we have
\begin{align*}
    \Phi_{p}^{(\beta)}\big(\rb[2^{z-1}\beta]_{q} + 2^{z-1} \cdot \beta \cdot \dist^{z}(p, q)\big)
    & ~=~ \sum_{t \in P} \big[\rb[2^{z-1}\beta]_{q} + 2^{z-1} \cdot \beta \cdot \dist^{z}(p, q) - \beta \cdot \dist^{z}(p, t)\big]^{+}\\
    \mr{\Cref{lem:triangle-inequality}}
    & ~\ge~ \sum_{t \in P} \big[\rb[2^{z-1}\beta]_{q} - 2^{z-1} \cdot \beta \cdot \dist^{z}(q, t)\big]^{+}\\
    & ~=~ \Phi_{q}^{(2^{z-1}\beta)}(\rb[2^{z-1}\beta]_{q}) ~=~ \lambda ~=~ \Phi^{(\beta)}_{p}(\rb_{p}).
\end{align*}
This implies $\rb_{p} \le \rb[2^{z-1}\beta]_{q} + 2^{z-1} \cdot \beta \cdot \dist^{z}(p, q)$, since the function $\Phi^{(\beta)}_{p}$ is increasing (and $\lambda > 0$).
\end{proof}

Third, we establish \Cref{lem:properties-rp-relation2},
which is similar to \Cref{lem:properties-rp-relation1} but, via a more careful argument, shows that $\rb_{p}$ and $\rb_{q}$ only differ by an additive term.
This difference depends on the (power-$z$) distance $\dist^{z}(p, q)$ as well as the neighborhood structure of $q$ (or $p$, by symmetry).

\begin{lemma}
\label{lem:properties-rp-relation2} 
For any scale factor $\beta > 0$ and any two points $p, q \in P$,
\begin{align*}
    \rb_{p}
     ~\le~& \rb_{q} + 2^{z-1} \cdot \beta \cdot \dist^{z}(p, q) + 
     \frac{1}{|\Ball_{q}(\rb_q/2)|} \cdot \Big(\sum_{t \in \Ball_{q}(\rb_q/2)} (2^{z-1} - 1) \cdot \beta \cdot \dist^{z}(t, q)\\
     &\phantom{\rb_{q} + 2^{z-1} \cdot \beta \cdot \dist^{z}(p, q) + 
     \frac{1}{|\Ball_{q}(\rb_q/2)|} \cdot \Big(}
    + \sum_{t \in \Ring_{q}(\rb_q/2)} \big(\rb_{q} - \beta \cdot \dist^{z}(t, q)\big)\Big).
\end{align*}
\end{lemma}

\begin{proof}
Let
\begin{equation*}
    \gamma := \frac{1}{|\Ball_{q}(\rb_q/2)|} \cdot \Big(\sum_{t \in \Ball_{q}(\rb_q/2)} (2^{z-1} - 1) \cdot \beta \cdot \dist^{z}(t, q) + \sum_{t \in \Ring_{q}(\rb_q/2)} \big(\rb_{q} - \beta \cdot \dist^{z}(t, q)\big)\Big).
\end{equation*}
By the generalized triangle inequality (\Cref{lem:triangle-inequality}), we have
\begin{align*}
    \Phi^{(\beta)}_{p}(\rb_{q} + 2^{z-1} \cdot \beta \cdot \dist^{z}(p, q) + \gamma)
    & ~=~ \sum_{t \in P} \big[\rb_{q} + 2^{z-1} \cdot \beta \cdot \dist^{z}(p, q) + \gamma - \beta \cdot \dist^{z}(t, p) \big]^{+} \\
    \mr{\Cref{lem:triangle-inequality}}
    & ~\ge~ \sum_{t \in P} \big[\rb_{q} + \gamma - 2^{z-1} \cdot \beta \cdot \dist^{z}(t, q)\big]^{+} \\
    & ~\ge~ \sum_{t \in \Ball_{q}(\rb_q/2)} \big(\rb_{q} + \gamma - 2^{z-1} \cdot \beta \cdot \dist^{z}(t, q)\big) \\
    & ~=~ \sum_{t \in \Ball_{q}(\rb_q/2)} \big(\rb_{q}
    - \beta \cdot \dist^{z}(t, q)\big)\\
    & \phantom{~\le~}\quad 
    + \sum_{t \in \Ring_{q}(\rb_q/2)} \big(\rb_{q} - \beta \cdot \dist^{z}(t, q)\big) \\
    \mr{$\Ball_{q}(\rb_{q} ) = \Ball_{q}(\rb_{q} / 2) \sqcup \Ring_{q}(\rb_{q}/2)$}
    & ~=~ \sum_{t \in \Ball_{q}(\rb_{q} )} \big(\rb_{q} - \beta \cdot \dist^{z}(t, q)\big) \\
    & ~=~ \lambda ~=~ \Phi^{(\beta)}_{p}(\rb_{p}).
\end{align*}
This implies $\rb_{p}\le \rb_{q} + 2^{z-1} \cdot \beta \cdot \dist^{z}(p, q) + \gamma$, since the function $\Phi^{(\beta)}_{p}$ is increasing (and $\lambda > 0$).
\end{proof}

\paragraph{Local Density.} 
Finally, we establish \Cref{lem:properties-rp-localdensity}, which shows that for each point $p \in P$, the sum of $\rb_{q}$ over all points $q$ within distance $O(\rb_p)$ (sufficiently large) is at least $\lambda$.
This property is straightforward in the $z = 1$ case by combining the definition of $\rb_{p}$ (in \eqref{eq:rp}) with \Cref{lem:properties-rp-relation1}. However, this simple argument cannot extend to all $z \ge 1$ cases, because of the $2^{z-1}$ factor in the statement of \Cref{lem:properties-rp-relation1}, and a more involved argument is needed for the extension.

\begin{lemma}
\label{lem:properties-rp-localdensity}
Let $\gamma^{*} := 2^{2z^{2} + z} = 2^{O(z^{2})}$.
For any scale factor $\beta > 0$ and every point $p \in P$,
\begin{align*}
    \sum_{q \in \Ball_{p}(\gamma^{*} \cdot \rb_{p})} \rb_{q} ~\ge~ \lambda.
\end{align*}
\end{lemma}

\begin{proof}
We show \Cref{lem:properties-rp-localdensity} by induction; specifically, \Cref{claim:local-density-base-case} verifies the base case, and \Cref{claim:local-density-induction} verifies the induction step.

\begin{claim}[Base Case for \Cref{lem:properties-rp-localdensity}]
\label{claim:local-density-base-case}
For any scale factor $\beta > 0$ and every point $p \in P$,
\begin{align*}
    \sum_{q \in \Ball_{p}(2^{z} \cdot \rb_{p})} \rb_{q}
    ~\ge~ (1 / 2^{z}) \cdot \lambda.
\end{align*}
\end{claim}

\begin{proof}
To ease the presentation, we reindex all points $P = \{p_{1}, \dots, p_n\}$ such that $\rb_{p_{1}} \le \dots \le \rb_{p_n}$; we will prove \Cref{claim:local-density-base-case} by induction on this ordering of points.

\textit{Base Case:}
The first point $p_{1}$ satisfies $\rb_{p_{1}} \le \rb_{q},\ \forall q \in P$, which together with~\eqref{eq:r-vector-equivalent} gives
\begin{align*}
    \sum_{q \in \Ball_{p_{1}}(2^{z} \cdot \rb_{p_{1}})} \rb_{q}
    ~\ge~ \sum_{q \in \Ball_{p_{1}}(\rb_{p_{1}})} \big(\rb_{p_{1}} - \beta \cdot \dist^{z}(p_{1}, q)\big)
    ~=~ \lambda
    ~\ge~ (1 / 2^{z}) \cdot \lambda.
\end{align*}

\textit{Induction Step:}
Fix an index $i \in [2 : n]$. We assume, as an induction hypothesis, that \Cref{claim:local-density-base-case} holds for points $p_{1}, \dots, p_{i - 1}$.
We next establish \Cref{claim:local-density-base-case} for $p_{i}$. Below, we rewrite $p = p_{i}$ for notational brevity and conduct a case analysis.

\textit{Case~1:}
for every point $q \in \Ball_{p}(\rb_{p})$, it holds that
$\rb_{q} \ge \frac{1}{2^{z}} \cdot (\rb_{p} - \beta \cdot \dist^{z}(p, q))$.
In this case, it follows directly that 
\begin{align*}
    \sum_{q \in \Ball_{p}(\rb_{p})} \rb_{q}
    ~\ge~ \sum_{q \in \Ball_{p}(\rb_{p})} (1 / 2^{z}) \cdot \big(\rb_{p} - \beta \cdot \dist^{z}(p, q)\big)
    ~=~ (1 / 2^{z}) \cdot \lambda,
\end{align*}
which completes our analysis of Case~1.

\textit{Case~2:}
there exists some point $t \in \Ball_{p}(\rb_{p})$ such that
$\rb_{t} < \frac{1}{2^{z}} \cdot (\rb_{p} - \beta \cdot \dist^{z}(p, t))$.
In this case, we make the following observations:
\begin{itemize}
    \item $t = p_{j}$ for some index $j \in [i - 1]$, since $\rb_{t} < \frac{1}{2^{z}} \cdot (\rb_{p} - \beta \cdot \dist^{z}(p, t)) < \rb_{p}$.
    
    \item $\Ball_{t}(2^{z} \cdot \rb_{t}) \subseteq \Ball_{p}(2^{z} \cdot \rb_{p})$. To see this, consider a specific point $q \in \Ball_{t}(2^{z} \cdot \rb_{t})$; we have $\beta \cdot \dist^{z}(q, t) \le 2^{z} \cdot \rb_{t}$.
    Moreover, by the premise of Case~2, the point $t$ satisfies $\beta \cdot \dist^{z}(p, t) \le \rb_{p}$ and $\rb_{t} < \frac{1}{2^{z}} \cdot (\rb_{p} - \beta \cdot \dist^{z}(p, t)) \le \frac{1}{2^{z}} \cdot \rb_{p}$.
    By the generalized triangle inequality (\Cref{lem:triangle-inequality}), we have
    \begin{align*}
        \beta \cdot \dist^{z}(q, p)
        & ~\le~ \beta \cdot 2^{z-1} \cdot \big(\dist^{z}(q, t) + \dist^{z}(p, t)\big) \\
        & ~\le~ 2^{z-1} \cdot \big(2^{z} \cdot \rb_{t} + \rb_{p}\big) \\
        & ~\le~ 2^{z-1} \cdot \big(2^{z} \cdot \tfrac{1}{2^{z}} + 1\big) \cdot \rb_{p} \\
        & ~=~ 2^{z} \cdot \rb_{p}.
    \end{align*}
    This implies $q \in \Ball_{p}(2^{z} \cdot \rb_{p})$ for every $q \in \Ball_{t}(2^{z} \cdot \rb_{t})$, i.e., $\Ball_{t}(2^{z} \cdot \rb_{t}) \subseteq \Ball_{p}(2^{z} \cdot \rb_{p})$.
\end{itemize}
Combining the above observations and the induction hypothesis completes our analysis in Case~2, as follows:
\begin{align*}
    \sum_{q \in \Ball_{p}(2^{z} \cdot \rb_{p})} \rb_{q}
    ~\ge~ \sum_{q \in \Ball_{t}(2^{z} \cdot \rb_{t})} \rb_{q}
    ~\ge~ (1 / 2^{z}) \cdot \lambda.
\end{align*}

Combining both cases completes the induction step, which immediately implies \Cref{claim:local-density-base-case}.
\end{proof} 

\begin{claim}[Induction Step for \Cref{lem:properties-rp-localdensity}]
\label{claim:local-density-induction}
Given parameters $\alpha \in (0,1)$ and $\gamma \ge 1$, if it holds that
\begin{align}
    \label{eq:local-density-induction-condition}
    & \sum_{q \in \Ball_{p}(\gamma \cdot \rb_{p})} \rb_{q} ~\ge~ \alpha\cdot \lambda,
    && \forall \beta > 0,\ \forall p \in P,
\end{align}
then it also holds that
\begin{align}
    \label{eq:local-density-induction-conclusion}
    & \sum_{q \in \Ball_{p}(2^{2z} \cdot \gamma \cdot \rb_{p})} \rb_{q} ~\ge~ \min\{2\alpha,\, 1\} \cdot \lambda,
    && \forall \beta > 0,\ \forall p \in P.
\end{align}

\end{claim}

\begin{proof}
We fix a scale factor $\beta >0$ and a point $p\in P$.
Let $q^{*} := \argmin_{q \in \Ball_{p}(2^{z} \cdot \gamma \cdot \rb_{p})} \rb_{q}$; we have $\beta \cdot \dist^{z}(q^{*}, p) \le 2^{z} \cdot \gamma \cdot \rb_{p}$ and $\rb_{q^{*}} \le \rb_{p}$.
We conduct a case analysis.

\textit{Case~1:}
for every point $t \in \Ball_{q^{*}}(\rb_{q^{*}})$, it holds that
$\rb_{t} \ge \rb_{q^{*}}$.
In this case, we observe that $\Ball_{q^{*}}(\rb_{q^{*}}) \subseteq \Ball_{p}(2^{2z} \cdot \gamma \cdot \rb_{p})$.
To see this, consider a specific point $t \in \Ball_{q^{*}}(\rb_{q^{*}})$; we have $\beta \cdot \dist^{z}(t, q^{*}) \le \rb_{q^{*}}$.
By the generalized triangle inequality (\Cref{lem:triangle-inequality}), we have
\begin{align*}
    \beta \cdot \dist^{z}(t, p)
    & ~\le~ \beta \cdot 2^{z-1} \cdot \big(\dist^{z}(t, q^{*}) + \dist^{z}(p, q^{*})\big)\\
    & ~\le~ 2^{z-1} \cdot \big(\rb_{q^{*}} + 2^{z} \cdot \gamma \cdot \rb_{p}\big)\\
    \mr{$\rb_{q^{*}} \le \rb_{p}$}
    & ~\le~ 2^{2z} \cdot \gamma \cdot \rb_{p}.
\end{align*}
This implies $t \in \Ball_{p}(2^{2z} \cdot \gamma \cdot \rb_{p})$ for every $t \in \Ball_{q^{*}}(\rb_{q^{*}})$, i.e., $\Ball_{q^{*}}(\rb_{q^{*}}) \subseteq \Ball_{p}(2^{2z} \cdot \gamma \cdot \rb_{p})$.
Combining this observation with the premise of Case~1 and \eqref{eq:r-vector-equivalent} completes our analysis in Case~1, as follows:
\begin{align*}
    \sum_{t \in \Ball_{p}(2^{2z} \cdot \gamma \cdot \rb_{p})} \rb_{t}
~\ge~ \sum_{t \in \Ball_{q^{*}}(\rb_{q^{*}})} \big(\rb_{q^{*}} - \beta \cdot \dist^{z}(q^{*}, t)\big)
    ~=~ \lambda.
\end{align*}

\textit{Case~2:}
there exists some point $t \in \Ball_{q^{*}}(\rb_{q^{*}})$ such that
$\rb_{t} < \rb_{q^{*}} \le \rb_{p}$.
In this case, we make the following observations:
\begin{itemize}
    \item $\Ball_{t}(\gamma \cdot \rb_{t}) \cap \Ball_{p}(\gamma \cdot \rb_{p}) = \emptyset$.
    To see this, assume for contradiction that there exists some point $q \in \Ball_{t}(\gamma \cdot \rb_{t}) \cap \Ball_{p}(\gamma \cdot \rb_{p})$.
    It follows that
    $\beta \cdot \dist^{z}(p, t)
    \le \beta \cdot 2^{z-1} \cdot (\dist^{z}(p, q) + \dist^{z}(t, q))
    \le 2^{z-1} \cdot \gamma \cdot (\rb_{t} + \rb_{p})
    \le 2^{z} \cdot \gamma \cdot \rb_{p}
    \implies
    t \in \Ball_{p}(2^{z} \cdot \gamma \cdot \rb_{p})$; given the premise of Case~2 that $\rb_{t} < \rb_{q^{*}}$, this contradicts the choice of $q^{*} = \argmin_{q \in \Ball_{p}(2^{z} \cdot \gamma \cdot \rb_{p})} \rb_{q}$.

    \item $\Ball_{t}(\gamma \cdot \rb_{t}) \subseteq \Ball_{p}(2^{2z} \cdot \gamma \cdot \rb_{p})$.
    To see this, 
    consider a specific point $q \in \Ball_{t}(\gamma \cdot \rb_{t})$; we have $\beta \cdot \dist^{z}(t, q) \le \gamma \cdot \rb_{t}$. By applying the generalized triangle inequality (\Cref{lem:triangle-inequality}) twice, we have
    \begin{align*}
        \beta \cdot \dist^{z}(p, q)
& ~\le~ \beta \cdot 2^{z-1} \cdot \big(\dist^{z}(p, q^{*}) + 2^{z-1} \cdot \big(\dist^{z}(q^{*}, t) + \dist^{z}(t, q)\big)\big)\\
        & ~\le~ 2^{z-1} \cdot \big(2^{z} \cdot \gamma \cdot \rb_{p} + 2^{z-1} \cdot \big(\rb_{q^{*}} + \gamma \cdot \rb_{t}\big)\big)\\
        \mr{$\rb_{t} < \rb_{q^{*}} \le \rb_{p}$}
        & ~\le~ 2^{2z} \cdot \gamma \cdot \rb_{p}.
    \end{align*}
    This implies $q \in \Ball_{p}(2^{2z} \cdot \gamma \cdot \rb_{p})$ for all $q \in \Ball_{t}(\gamma \cdot \rb_{t})$, i.e., $\Ball_{t}(\gamma \cdot \rb_{t}) \subseteq \Ball_{p}(2^{2z} \cdot \gamma \cdot \rb_{p})$.
\end{itemize}
Combining the above observations and the premise of \Cref{claim:local-density-induction} completes our analysis in Case~2, as follows:
\begin{align*}
    \sum_{q \in \Ball_{p}(2^{2z} \cdot \gamma \cdot \rb_{p})} \rb_{q} ~\ge~ \sum_{q \in \Ball_{t}(\gamma \cdot \rb_{t})} \rb_{q}
    + \sum_{q \in \Ball_{p}(\gamma \cdot \rb_{p})} \rb_{q}
    ~\ge~ 2\alpha\cdot \lambda.
\end{align*}

Combining both cases completes the proof of \Cref{claim:local-density-induction}.
\end{proof}

Let us return to the proof of \Cref{lem:properties-rp-localdensity}; applying \Cref{claim:local-density-base-case,claim:local-density-induction} gives, for every integer $i \ge 0$,
\begin{align*}
    \sum_{q \in \Ball_{p}(2^{(2i + 1) z} \cdot \rb_{p})} \rb_{q} ~\ge~ \min\{2^{i - z},\, 1 \} \cdot \lambda.
\end{align*}
By setting $i = z$, and thus $\gamma^{*} = 2^{2z^{2} + z}$, \Cref{lem:properties-rp-localdensity} follows.
\end{proof}

\subsection{\texorpdfstring{Analysis of $\MPbeta$ without Distance Distortions}{Analysis of MPbeta without Distance Distortions}}
\label{sec:fractional-MP}

In this subsection, we analyze $\MPbeta$ (\Cref{alg:robust-MP}) on input a metric space $(P, \dist)$ without any distance distortion, and show that it yields a fractional LMP approximation for \textsc{Power-$z$ Facility Location}.
The result is formally stated in the following lemma, which can be viewed as a generalization of \cite[Theorem~13.1]{BhattacharyaCGL24} from the $z = 1$ case to any $z \ge 1$.

\begin{restatable}{lemma}{FractionalMP}
\label{lem:fractional-MP}
Let $\beta^{*} := 2^{-z-1}$ and let $(P, \dist)$ be a metric space.
For any scale factor $\beta \in (0, \beta^{*}]$, the algorithm $\MPbeta$ (\Cref{alg:robust-MP}) on input $(P, \dist)$ returns a fractional solution $\Byb$ to \textsc{Power-$z$ Facility Location} with an LMP $2^{O(z^{2})}\beta^{-1}$-approximation guarantee, i.e., 
\begin{align*}
    \cost(\Byb)
    ~\le~ 2^{O(z^{2})} \cdot \beta^{-1} \cdot (\OPTfl[\lambda] - \lambda \cdot \| \Byb \|_{1}).
\end{align*}
\end{restatable}

Before proving \Cref{lem:fractional-MP}, we present the following technical claim, which relates the clustering cost $\cost(\By)$ of $\By$ to the assignment cost $\asncost(\Bx)$ for some assignment $\Bx$ that may be infeasible with respect to $\By$.

\begin{restatable}{claim}{AdditiveError}
\label{claim:additive-error}
Suppose that a fractional solution $\By \in \R_{\ge 0}^{P}$, a fractional assignment $\Bx \in \R_{\ge 0}^{P \times P}$, and a vector $\bm{\mu}\in \R_{\ge 0}^{P}$ satisfy the following conditions:
\begin{enumerate}
    \item Each point $p \in P$ satisfies $\sum_{q \in P: \dist^{z}(p, q) \le \mu_{p}} y_{q} \ge 1$.
    
    \item Each point $p \in P$ satisfies $\sum_{q \in P} \x_{p, q} = 1$.
    
    \item Each pair $p, q \in P$ with $x_{p, q} > 0$ satisfies $\dist^{z}(p, q) \le \mu_{p}$.
\end{enumerate}
Then, it holds that 
\begin{align*}
    \cost(\By) ~\le~ \asncost(\Bx) + \sum_{p, q \in P} \left[\x_{p, q} - y_{q} \right]^{+} \cdot \mu_{p}.
\end{align*}
\end{restatable}

\begin{proof}
We introduce two auxiliary assignments $\Bxlb = (\xlb_{p, q})_{p, q \in P}$ and $\Bxub = (\xub_{p, q})_{p, q \in P}$:
\begin{itemize}
    \item $\Bxlb$ is defined as $\xlb_{p, q} := \min\{\x_{p, q},\ y_{q}\}$, $\forall p, q \in P$.
    
    \item $\Bxub$ is defined as $\xub_{p, q} := y_{q} \cdot \mathbb{I}[\dist^{z}(p, q) \le \mu_{p}]$, $\forall p, q \in P$, where $\mathbb{I}[\cdot]$ is the indicator function.
\end{itemize}
The premises of \Cref{claim:additive-error} guarantee $\Bxlb \le \Bxub$ and for every $p\in P$,
$\sum_{q \in P} \xlb_{p, q }
\le \sum_{q \in P} \x_{p, q}
= 1
\le \sum_{q \in P: \dist^{z}(p, q)
\le \mu_{p}} y_{q}
= \sum_{q \in P} \xub_{p, q}$.
This further implies the existence of an assignment $\Bxmd = (\xmd_{p, q})_{p, q \in P}$ with $\Bxlb \le \Bxmd \le \Bxub$ and $\sum_{q \in P} \xmd_{p, q} = 1$, $\forall p \in P$. Note that $\xmd_{p, q} \le \xub_{p, q} \le y_{q}$, $\forall p, q \in P$.
Hence, this is a feasible assignment $\Bxmd \sim \By$, and we can deduce that
\begin{align*}
    \cost(\By)
    ~\le~ \asncost(\Bxmd)
    & ~=~ \asncost(\Bx) + \sum_{p, q \in P} (\xmd_{p, q} - \x_{p, q}) \cdot \dist^{z}(p, q)\\
    & ~\le~ \asncost(\Bx) + \sum_{p, q \in P} [\xmd_{p, q} - \x_{p, q}]^{+} \cdot \dist^{z}(p, q)\\
    & ~\le~ \asncost(\Bx) + \sum_{p, q \in P} [\xmd_{p, q} - \x_{p, q}]^{+} \cdot \mu_{p}\\
    & ~=~ \asncost(\Bx) + \sum_{p, q \in P} [\x_{p, q} - \xmd_{p, q}]^{+} \cdot \mu_{p}\\
    \mr{$\xmd_{p, q} \ge \xlb_{p, q} = \min\{\x_{p, q},\, y_{q}\}$}
    & ~\le~ \asncost(\Bx) + \sum_{p, q \in P} [\x_{p, q} - \min\{\x_{p, q},\, y_{q}\}]^{+} \cdot \mu_{p}\\
    & ~=~ \asncost(\Bx) + \sum_{p, q \in P} [\x_{p, q} - y_{q}]^{+} \cdot \mu_{p}.
\end{align*}
Here, the third inequality holds since $[\xmd_{p, q} - \x_{p, q}]^{+} > 0$ only if $0 < \xmd_{p, q} \le \xub_{p, q}$, and then only if $\dist^{z}(p, q) \le \mu_{p}$.
The second equality holds since $\sum_{q \in P} \xmd_{p, q} 
= \sum_{q \in P} \x_{p, q} = 1 \implies \sum_{q \in P} [\xmd_{p, q} - \x_{p, q}]^{+} = \sum_{q \in P} [\x_{p, q} - \xmd_{p, q}]^{+}$, $\forall p \in P$.
This completes the proof of \Cref{claim:additive-error}.
\end{proof}

We are now ready to prove \Cref{lem:fractional-MP}.

\begin{proof}[Proof of \Cref{lem:fractional-MP}]

Let $\Brb = (\rb_p)_{p\in P}$ denote the vector computed by $\MPbeta$ (in Line~\ref{alg:robust-MP:rp} of \Cref{alg:robust-MP}); then, the output of $\MPbeta$ is $\Byb = \Brb / \lambda$.
The proof follows a primal-dual approach. 
To begin, we present the primal and dual LPs for \textsc{Power-$z$ Facility Location}, denoted by $\LPufl$ and $\DPufl$, respectively.
\begin{alignat}{2}
\LPufl:\qquad
\min \quad & \lambda \cdot \sum_{p \in P} y_{p} + \sum_{p, q \in P} x_{p, q} \cdot \dist^{z}(p, q) \tag{L}\label{eq:LP-ufl-objective} \\
\text{s.t.}\quad
& \sum_{q \in P} x_{p, q} \ge 1 && \forall p \in P \tag{L1} \label{eq:LP-ufl-sum} \\
& x_{p, q} \le y_{q}            && \forall p, q \in P \tag{L2} \label{eq:LP-ufl-cover}\\
& x_{p, q},y_{p}\ge 0   && \forall p, q \in P \tag{L3} \label{eq:LP-ufl-nonneg}\\
\DPufl: \qquad 
\max \quad & \sum_{p \in P} v_{p} \tag{D} \label{eq:DP-ufl-objective}\\
\text{s.t.} \quad 
& \sum_{p \in P} \big[v_{p} - \dist^{z}(p, q)\big]^{+} \le \lambda,\quad  && \forall q \in P \tag{D1} \label{eq:DP-ufl-constraint}\\
& v_{p} \ge 0, &&\forall p \in P \tag{D2} \label{eq:DP-ufl-nonneg}
\end{alignat}
The optimal primal value is exactly $\OPTfl[\lambda] = \min_{\By \in \R_{\ge 0}^{P}: \|\By\|_{1} \ge 1} \lambda \cdot \|\By\|_{1} +  \cost(\By)$, since Constraints~\eqref{eq:LP-ufl-sum} to \eqref{eq:LP-ufl-cover} are identical to Constraints~\ref{eq:LP-cl-sum} to \ref{eq:LP-cl-cover} (see \Cref{sec:prelim}).
Further, by weak duality, every feasible dual solution $\Bv$ satisfies $\|\Bv \|_{1} = \sum_{p \in P}v_{p} \le \OPTfl[\lambda]$.

We define a concrete fractional assignment $\Bxb = (\xb_{p, q})_{p,q \in P}$ as follows:
\begin{align}
    \label{eq:infeasible-assignment}
    & \xb_{p, q} ~:=~ (1 / \lambda) \cdot \big[\rb_{p} - \beta \cdot \dist^{z}(p, q)\big]^{+},
    && \forall p, q \in P.
\end{align}
To prove \Cref{lem:fractional-MP}, we use the assignment cost $\asncost(\Bxb)$ as an intermediate bridge. 
Specifically, we show in \Cref{lem:fractional-MP:LMP} that the cost $\asncost(\Bxb)$ satisfies the desired LMP property, and then show in \Cref{lem:feasibility} that $\cost(\Byb)$ is not much larger than $\asncost(\Bxb)$.
We note that the proof of \Cref{lem:fractional-MP:LMP} follows almost the same approach as in~\cite{BhattacharyaCGL24}, except that we use our new smoothness property in \Cref{lem:properties-rp-relation1}, rather than the property they use which only holds for the $z = 1$ case.

\begin{lemma}
\label{lem:fractional-MP:LMP}
$\asncost(\Bxb) \le \beta^{-1} \cdot (\OPTfl[\lambda] - \lambda \cdot \|\Byb\|_{1})$.
\end{lemma}
\begin{proof}
We establish this lemma by analyzing a concrete \emph{dual solution} $\Bvb := \Brb[2\beta]$.
We first verify that this dual solution is feasible; by weak duality, this means $\|\Bvb\|_{1} \le \OPTfl[\lambda]$.
Constraint~\eqref{eq:DP-ufl-nonneg} trivially holds since $\Brb[2\beta] \in \R_{\ge 0}^{P}$, and we can verify Constraint~\eqref{eq:DP-ufl-constraint} via \Cref{lem:properties-rp-relation1}:
\begin{align*}
    \sum_{p \in P} \big[\rb[2\beta]_{p} - \dist^{z}(p, q)\big]^{+}
    & ~\le~ \sum_{p \in P} \big[\rb[2^{z} \cdot \beta]_{q} + 2^{z} \cdot \beta \cdot \dist^{z}(p, q) - \dist^{z}(p, q)\big]^{+}\\
    \mr{$\beta \le \beta^{*} = \frac{1}{2^{z+1}}$}
    & ~\le~ \sum_{p \in P} \big[\rb[2^{z} \cdot \beta]_{q} - 2^{z} \cdot \beta \cdot \dist^{z}(p, q)\big]^{+}\\
    \mr{\eqref{eq:rp}}
    & ~=~ \lambda.
\end{align*}
Moreover, the following \Cref{claim:lem13.6-BCG24} essentially restates \cite[Lemma~13.6]{BhattacharyaCGL24} under our notation.

\begin{restatable}[{\cite[Lemma~13.6]{BhattacharyaCGL24}}]{claim}{BCGLemma}
\label{claim:lem13.6-BCG24}
For each point $p \in P$,
\begin{align*}
    \rb_{p} + \sum_{q \in P} (1 / \lambda) \cdot \big[\rb_{p} - \beta \cdot \dist^{z}(p, q) \big]^{+} \cdot \beta \cdot \dist^{z}(p, q)
    ~\le~ \rb[2\beta]_{p}.
\end{align*}
\end{restatable}

\begin{proof}
The proof is almost identical to that of~\cite[Lemma 13.6]{BhattacharyaCGL24}, although they only consider the special case with $z=1$ and $\beta = 1/4$.
Specifically, rescaling all distances in their original proof directly gives the extension to $\beta \in (0, 1)$.
Also, replacing their distance function $d$ with $\dist^{z}$ directly gives the extension to $z \ge 1$; in this regard, we emphasize that their original proof does not require the triangle inequality, so all the arguments carry over.
\end{proof}

By combining everything, we can deduce \Cref{lem:fractional-MP:LMP} as follows: 
\begin{align*}
    \asncost(\Bxb)
    & ~=~ \sum_{p \in P} \sum_{q \in P} \xb_{p, q} \cdot \dist^{z}(p, q)\\
    \mr{\eqref{eq:infeasible-assignment}}
    & ~=~ \sum_{p \in P} \sum_{q \in P} (1 / \lambda) \cdot \big[\rb_{p} - \beta \cdot \dist^{z}(p, q)\big]^{+} \cdot \dist^{z}(p, q)\\
    \mr{\Cref{claim:lem13.6-BCG24}}
    & ~\le~ \beta^{-1} \cdot \sum_{p \in P} \big(\rb[2\beta]_{p} - \rb[\beta]_{p}\big)\\
    \mr{$\Byb = (1 / \lambda) \cdot \Brb[\beta]$}
    & ~=~ \beta^{-1} \cdot \big(\|\Brb[2\beta]\|_{1} - \lambda \cdot \|\Byb\|_{1}\big)\\
    \mr{weak duality}
    & ~\le~ \beta^{-1} \cdot \big(\OPTfl[\lambda] - \lambda \cdot \|\Byb\|_{1}\big).
    \qedhere
\end{align*}
\end{proof}

In the $z = 1$ case considered in~\cite{BhattacharyaCGL24}, the assignment $\Bxb$ is feasible with respect to $\Byb$ (simply by \Cref{lem:properties-rp-relation1}). 
However, this feasibility does not extend to $z \ge 2$, and we provide a counterexample in \Cref{example:counterexample} at the end of this subsection.
To address this issue, we additionally establish the following \Cref{lem:feasibility}.

\begin{lemma}
\label{lem:feasibility}
$\cost(\Byb) \le 2^{O(z^{2})} \cdot  \asncost(\Bxb)$.
\end{lemma}

\begin{proof}
We first upper-bound the difference between $\cost(\Byb)$ and $\asncost(\Bxb)$ via \Cref{claim:additive-error}. To this end,
we verify the three premises of \Cref{claim:additive-error} for the choices $\By := \Byb = (1 / \lambda) \cdot \Brb$, $\Bx := \Bxb$, and $\bm{\mu} := \gamma^{*} \cdot \beta^{-1} \cdot \Brb$ (for the constant $\gamma^{*} = 2^{O(z^2)}$ specified in \Cref{lem:properties-rp-localdensity}) as follows:
\begin{enumerate}
    \item For every $p \in P$, by \Cref{lem:properties-rp-localdensity}, we have $\sum_{q \in P: \dist^{z}(p, q) \le \mu_{p}} y_{q}
    = \sum_{q \in \Ball_{p}(\gamma^{*}\cdot \rb_{p})} \rb_{q}/\lambda
    \ge 1$.
    
    \item For every $p \in P$, 
    $\sum_{q \in P} \x_{p, q}
    = \sum_{q \in P} \xb_{p, q}
    = \sum_{q \in P} (1 / \lambda) \cdot [\rb_{p} - \beta \cdot \dist^{z}(p, q)]^{+} = 1$, where the second step applies \eqref{eq:infeasible-assignment}, and the last step follows from \eqref{eq:rp}.
    
    \item For every $p, q \in P$ with $0 < \x_{p, q} = \xb_{p, q} = (1 / \lambda) \cdot [\rb_{p} - \beta \cdot \dist^{z}(p, q)]^{+}$, it follows directly that
    $\dist^{z}(p, q)
    \le \beta^{-1} \cdot \rb_{p}
    \le \gamma^{*} \cdot \beta^{-1} \cdot \rb_p
    = \mu_{p}$.
\end{enumerate}
Then, by applying \Cref{claim:additive-error}, we can derive
\begin{align}
    \label{eq:additive-error-yb}
    \cost(\Byb) ~\le~ \asncost(\Bxb) + \gamma^{*} \cdot \beta^{-1} \cdot \sum_{p, q \in P} [\xb_{p, q} - \yb_{q}]^{+} \cdot \rb_{p}.
\end{align}
We divide all pairs $P \times P = \calF \sqcup \calC$ into the \textit{``far'' pairs} $\calF$ and the \textit{``close'' pairs} $\calC$, as follows:
\begin{align*}
    \calF & ~:=~ \big\{(p, q) \in P \times P:\ \beta \cdot \dist^{z}(p, q) > \tfrac{2}{3} \cdot 2^{-z} \cdot \rb_{p}\big\},\\
\calC & ~:=~ \big\{(p, q) \in P \times P:\ \beta \cdot \dist^{z}(p, q) \le \tfrac{2}{3} \cdot 2^{-z} \cdot \rb_{p}\big\}.
\end{align*}
Below, \Cref{claim:boundF,claim:boundC} investigate these two kinds of pairs separately.

\begin{restatable}{claim}{BoundF}
\label{claim:boundF}
$\sum_{(p, q) \in \calF} [\xb_{p, q} - \yb_{q}]^{+} \cdot \rb_{p} \le 2^{O(z)} \cdot \beta \cdot \asncost(\Bxb)$.
\end{restatable}

\begin{proof}
For every far pair $(p, q) \in \calF$, by definition it satisfies that $\rb_{p} \le \frac{3}{2} \cdot 2^{z} \cdot \beta \cdot \dist^{z}(p, q)$. Hence, we have $\sum_{(p, q) \in \calF} [\xb_{p, q} - \yb_{q}]^{+} \cdot \rb_{p}
\le \sum_{(p, q) \in \calF} \xb_{p, q} \cdot (\frac{3}{2} \cdot 2^{z} \cdot \beta \cdot \dist^{z}(p, q))
\le 2^{O(z)} \cdot \beta \cdot \asncost(\Bxb)$.
\end{proof}

\begin{restatable}{claim}{BoundC}
\label{claim:boundC}
$\sum_{(p, q) \in \calC} [\xb_{p, q} - \yb_{q}]^{+} \cdot \rb_{p} \le 2^{O(z)} \cdot \beta \cdot \asncost(\Bxb)$.
\end{restatable}

\begin{proof}
For every close pair $(p, q) \in \calC$, by definition we have
$\beta \cdot \dist^{z}(p, q)
\le \frac{2}{3} \cdot 2^{-z} \cdot \rb_{p}
\le \frac{1}{3} \cdot \rb_{p}$.
This in combination with \eqref{eq:infeasible-assignment} implies
$\rb_{p} \le \frac{3}{2} \cdot \lambda \cdot \xb_{p, q}$, and in combination with (\Cref{lem:properties-rp-relation1}) $\rb_{p} \le 2^{z-1}\cdot (\rb_q + \beta\cdot \dist^z(p,q))$ implies  $\rb_p \le  3\cdot 2^{z-2} \cdot \rb_{q}$ and then $\beta \cdot \dist^{z}(p, q) \le \rb_{q} / 2$. Therefore, we have $\calC \subseteq \{(p, q) \in P \times P:\ p \in \Ball_{q}(\rb_{q} / 2)\}$.

Below, we simply write $\Ball_{q} = \Ball_{q}(\rb_{q} / 2)$ and $\Ring_{q} = \Ring_{q}(\rb_{q} / 2)$ for notational brevity.
Following the definitions of $\Byb$ and $\Bxb$ (see \Cref{lem:fractional-MP} and \eqref{eq:infeasible-assignment}), we can deduce that
\begin{align*}
    \sum_{(p, q) \in \calC} \big[\xb_{p, q} - \yb_{q}\big]^{+} \cdot \rb_{p}
    & ~=~ (1 / \lambda) \cdot \sum_{(p, q) \in \calC} \rb_{p} \cdot \big[\rb_{p} - \beta \cdot \dist^{z}(p, q) - \rb_{q}\big]^{+}\\
    \mr{\Cref{lem:properties-rp-relation2}}
    & ~\le~ (1 / \lambda) \cdot \sum_{(p, q) \in \calC} \rb_{p} \cdot (2^{z-1} - 1) \cdot \beta \cdot \dist^{z}(p, q)\\
    & \phantom{~=~}\quad + (1 / \lambda) \cdot \sum_{(p, q) \in \calC} \rb_{p} \cdot \frac{1}{|\Ball_{q}|} \cdot \Big(\sum_{t \in \Ball_{q}} (2^{z-1} - 1) \cdot \beta \cdot \dist^{z}(t, q)\\
    & \phantom{~=~\quad + (1 / \lambda) \cdot \sum_{(p, q) \in \calC} \rb_{p} \cdot \frac{1}{|\Ball_{q}|} \cdot \Big(} + \sum_{t \in \Ring_{q}} \big(\rb_{q} - \beta \cdot \dist^{z}(t, q)\big)\Big) \\
    \mr{$\rb_{p} \le \frac{3}{2} \cdot \lambda \cdot \xb_{p, q}$}
    & ~\le~ (1 / \lambda) \cdot \sum_{(p, q) \in \calC} \big(\tfrac{3}{2} \cdot \lambda \cdot \xb_{p, q}\big) \cdot (2^{z-1} - 1) \cdot \beta \cdot \dist^{z}(p, q)\\
    \mr{$\rb_{p} \le 3\cdot 2^{z-2}\cdot \rb_{q}$}
    & \phantom{~=~}\quad + (1 / \lambda) \cdot \sum_{(p,q) \in \calC} \frac{3\cdot 2^{z-2}\cdot \rb_{q}}{|\Ball_{q}|}\cdot \Big(\sum_{t \in \Ball_{q}} (2^{z-1} - 1) \cdot \beta \cdot \dist^{z}(t, q)\\
    & \phantom{~=~\quad + (1 / \lambda) \cdot \sum_{(p, q) \in \calC} \big(3\cdot 2^{z-2}\cdot \rb_{q}\big) \cdot \Big(} + \sum_{t \in \Ring_{q}} \big(\rb_{q} - \beta \cdot \dist^{z}(t, q)\big)\Big)\\
    \mr{$\calC \subseteq \{(p, q):\ p \in \Ball_{q}\}$}
    & ~\le~ (1 / \lambda) \cdot \sum_{q\in P}\sum_{p\in \Ball_q} \big(\tfrac{3}{2} \cdot \lambda \cdot \xb_{p, q}\big) \cdot (2^{z-1} - 1) \cdot \beta \cdot \dist^{z}(p, q)\\
    & \phantom{~=~}\quad + (1 / \lambda) \cdot \sum_{q \in P} \big(3\cdot 2^{z-2}\cdot \rb_{q}\big) \cdot \Big(\sum_{t \in \Ball_{q}} (2^{z-1} - 1) \cdot \beta \cdot \dist^{z}(t, q)\\
    & \phantom{~=~\quad + (1 / \lambda) \cdot \sum_{q \in P} \big(3\cdot 2^{z-2}\cdot \rb_{q}\big) \cdot \Big(} + \sum_{t \in \Ring_{q}} \big(\rb_{q} - \beta \cdot \dist^{z}(t, q)\big)\Big)\\
    & ~\le~ 2^{O(z)} \cdot \beta \cdot \sum_{q\in P}\sum_{p\in \Ball_q} \xb_{p, q} \cdot \dist^{z}(p, q)\\
    & \phantom{~=~}\quad + 2^{O(z)} \cdot (1 / \lambda) \cdot \sum_{q \in P} \Big(\beta \cdot \sum_{t \in \Ball_{q}} \rb_{q} \cdot \dist^{z}(t, q)\\
    & \phantom{~=~\quad + 2^{O(z)} \cdot (1 / \lambda) \cdot \sum_{q \in P} \Big(} + \sum_{t \in \Ring_{q}} \rb_{q} \cdot \big(\rb_{q} - \beta \cdot \dist^{z}(t, q)\big)\Big).
\end{align*}
For every  $t \in \Ball_{q}$, $\beta \cdot \dist^{z}(q, t) \le \rb_{q} / 2$ (see \eqref{eq:ball-beta}) $\implies \rb_{q} \le 2 \cdot [\rb_{q} - \beta \cdot \dist^{z}(q, t)]^{+} = 2\lambda \cdot \xb_{q, t}$.
For every $t \in \Ring_{q}$, $\beta \cdot \dist^{z}(t, q) \le \rb_{q} \le 2\beta \cdot \dist^{z}(t, q)$ (see \eqref{eq:ring}) $\implies \rb_{q} - \beta \cdot \dist^{z}(q, t) = [\rb_{q} - \beta \cdot \dist^{z}(q, t)]^{+} = \lambda \cdot \xb_{q, t}$. Plugging these into the above equation completes the proof:
\begin{align*}
    \sum_{(p, q) \in \calC} \big[\xb_{p, q} - \yb_{q}\big]^{+} \cdot \rb_{p}
    & ~\le~ 2^{O(z)} \cdot \beta \cdot \sum_{q\in P}\sum_{p\in \Ball_q} \xb_{p, q} \cdot \dist^{z}(p, q)\\
    & \phantom{~=~}\quad + 2^{O(z)} \cdot (1 / \lambda) \cdot \sum_{q \in P} \Big(\beta \cdot \sum_{t \in \Ball_{q}} \big(2\lambda \cdot \xb_{q, t}\big) \cdot \dist^{z}(t, q)\\
    & \phantom{~=~\quad + 2^{O(z)} \cdot (1 / \lambda) \cdot \sum_{q \in P} \Big(} + \sum_{t \in \Ring_{q}} \big(2\beta \cdot \dist^{z}(t, q)\big) \cdot \big(\lambda \cdot \xb_{q, t}\big)\Big)\\
& ~\le~ 2^{O(z)} \cdot \beta \cdot \sum_{p, q \in P} \xb_{p, q} \cdot \dist^{z}(p, q)\\
    & ~=~ 2^{O(z)} \cdot \beta \cdot \asncost(\Bxb).
    \qedhere
\end{align*}
\end{proof}

Applying \Cref{claim:boundF,claim:boundC} to \eqref{eq:additive-error-yb} completes the proof of \Cref{lem:feasibility}, as follows:
\begin{align*}
    \cost(\Byb)
    & ~\le~ \asncost(\Bxb) + \gamma^{*} \cdot \beta^{-1} \cdot \sum_{(p, q) \in \calF \sqcup \calC} \big[\xb_{p, q} - \yb_{q}\big]^{+} \cdot \rb_{p}\\
    & ~\le~ \asncost(\Bxb) + \gamma^{*} \cdot 2^{O(z)} \cdot \asncost(\Bxb)\\
    \mr{$\gamma^{*} = 2^{O(z^2)}$}
    &~\le~ 2^{O(z^{2})} \cdot \asncost(\Bxb).
    \qedhere
\end{align*}
\end{proof}

\Cref{lem:fractional-MP} follows directly by combining \Cref{lem:fractional-MP:LMP,lem:feasibility}.
\end{proof}

\begin{example}[Counterexample to $\Bxb$'s Feasibility]
\label{example:counterexample}
Let $\lambda=1$, $\beta=\beta^{*} = 2^{-z-1} = 1/8$, and $z=2$.
Consider points on the $1$-dimensional Euclidean line: place $p$ at $0$, $q$ at $\sqrt{2}$, and $n\ge 1$ points at $\sqrt{5}$.
Then $\rb_{p}=5/8$ and therefore $\xb_{p,q}=3/8$.
On the other hand, it is easy to see that
$\rb_{q}\le \frac{1}{8}(\sqrt{5}-\sqrt{2})^2 + O(1/n)\approx 0.0844 + O(1/n)$.
For sufficiently large $n$, we have $\yb_{q} = \rb_{q} < \xb_{p,q}$,
which violates Constraint~\eqref{eq:LP-ufl-cover}.
As a result, for this instance, the assignment $\Bxb$ is not feasible with respect to $\Byb$.
\end{example}

\subsection{\texorpdfstring{Analysis of $\MPbeta$ under Distance Distortions (Proof of \Cref{lem:FMP-distance-distortion})}{}}
\label{sec:revisit-fractionalMP}

In this subsection, we prove \Cref{lem:FMP-distance-distortion}, which analyzes $\MPbeta$ (\Cref{alg:robust-MP}) under distance distortions, and we restate it below.

\RobustMP*

\begin{proof}Let $\beta^{*} := 2^{-z-1}$; then $\beta = \beta^{*}\epsilon^{2z}$.
Let $\Byb := (1 / \lambda) \cdot \Brb$ and $\Byb[\beta^{*}] := (1 / \lambda) \cdot \Brb[\beta^{*}]$, which are the outputs of $\MPbeta$ and $\mathrm{MP}_{\beta^{*}}$ on the original metric space $(P, \dist)$, respectively.
Let $\Bralg = (\ralg_p)_{p\in P}$ where $\ralg_p$ is computed by $\MPbeta$ when run on the distorted $(P, \Tilde{\dist})$, i.e., 
\begin{align*}
    &\sum_{q\in P}\big[\ralg_p - \beta\cdot \Tilde{\dist}^z(p,q)\big]^{+} ~=~ \lambda, && \forall p\in P.
\end{align*}
We first show that the vector $\Bralg$ satisfies $\Brb\le \Bralg \le \Brb[\beta^{*}]$.
To see this, observe from \eqref{eq:distorted-distanct} that, for every $p\in P$,
\begin{align*}
    \sum_{q\in P} \big[\ralg_p - \beta^{*}\cdot \dist^z(p,q)\big]^{+}
    &~\le~ 
    \sum_{q\in P} \big[\ralg_p - \beta\cdot \Tilde{\dist}^z(p,q)\big]^{+} 
    ~=~
    \lambda
    ~=~
    \sum_{q\in P} \big[\rb[\beta^{*}]_p - \beta^{*}\cdot \dist^z(p,q)\big]^{+}\\
    \sum_{q\in P} \big[\ralg_p - {\beta}\cdot {\dist}^z(p,q)\big]^{+}
    &~\ge~ \sum_{q\in P} \big[\ralg_p - {\beta}\cdot \Tilde{\dist}^z(p,q)\big]^{+} 
    ~=~
    \lambda
    ~=~
    \sum_{q\in P} \big[\rb_p - {\beta}\cdot \dist^z(p,q)\big]^{+}
\end{align*}
Since the expression $\sum_{q \in P} [r - \beta \cdot \dist^{z}(p, q)]^{+}$ is increasing in $r$ (and $\lambda > 0$ is nonzero), we infer from the above two equations that $\rb[\beta]_{p} \le \ralg_{p} \le \rb[\beta^{*}]_{p}$; therefore, 
$\Brb[{\beta}] \le \Bralg \le \Brb[\beta^{*}]$ holds.

Since $\Byalg = \Bralg / \lambda$, we have $\Byb \le \Byalg \le \Byb[\beta^{*}]$.
Finally, we use this fact to prove the LMP property of $\Byalg$. 
Recall that $\beta^{*} = 2^{-z-1}$ and $\Byb[\beta^{*}]$ is the output of $\mathrm{MP}_{\beta^{*}}$ on the original space $(P,\dist)$; by \Cref{lem:fractional-MP}, we have 
\begin{equation}\label{eq:LMP-Bybbetastar}
    \cost(\Byb[\beta^*])
    ~\le~ 2^{O(z^{2})}\cdot (\OPTfl[\lambda] - \lambda \cdot \| \Byb[\beta^{*}] \|_{1}).
\end{equation}
The basic idea is to relate the clustering cost $\cost(\Byalg)$ to the assignment cost $\asncost(\Bxb[\beta^{*}]) = \cost(\Byb[\beta^{*}])$, where $\Bxb[\beta^{*}] \sim \Byb[\beta^{*}]$ is the optimal feasible assignment to $\Byb[\beta^{*}]$.
To this end, we apply \Cref{claim:additive-error}.
We verify the premises of \Cref{claim:additive-error} for the choices $\By := \Byalg = (1 / \lambda) \cdot \Bralg$, $\Bx:= \Bxb[\beta^{*}]$, and $\bm{\mu} := \gamma^{*} \cdot \beta^{-1} \cdot \Brb$ (for the constant $\gamma^{*} = 2^{O(z^2)}$ specified in \Cref{lem:properties-rp-localdensity}):
\begin{enumerate}
    \item For every $p \in P$, by definition we have 
    $\ralg_{p} \ge \rb_p$.
It then follows that
    $\sum_{q \in P: \dist^{z}(p, q) \le \mu_{p}} \ralg_{q}
    \ge \sum_{q \in P: \dist^{z}(p, q) \le \gamma^{*}\cdot \beta^{-1}\cdot \rb_{p}} \rb_{q}
    = \sum_{q \in \Ball_{p}(\gamma^{*} \cdot \rb_{p})} \rb_{q}
    \ge \lambda$,
    where the second step follows from the definition of $\Ball_{p}(\gamma^{*} \cdot \rb_{p})$ in \eqref{eq:ball-beta}, and the third step applies \Cref{lem:properties-rp-localdensity}.
    This further implies that $\sum_{q \in P: \dist^{z}(p, q) \le \mu_{p}} y_{q} \ge 1$.

    \item For every $p \in P$, $\sum_{q \in P} x_{p, q} = \sum_{q \in P} \xb[\beta^{*}]_{p, q} = 1$, given the optimality of $\Bxb[\beta^{*}] \sim \Byb[\beta^{*}]$.

    \item 
    Since $\By = \Byalg \le \Byb[\beta^{*}]$, we have 
    $\sum_{q \in P: \dist^{z}(p, q) \le \mu_{p}} \yb[\beta^{*}]_{q} \ge \sum_{q \in P: \dist^{z}(p, q) \le \mu_{p}} y_{q}\ge 1$ for every $p\in P$, i.e., the nearby points $\{q \in P: \dist^{z}(p, q) \le \mu_{p}\}$ already have sufficient opening for $p$.
So, in the optimal fractional assignment $\Bxb[\beta^{*}] \sim \Byb[\beta^{*}]$ under Constraints~\ref{eq:LP-cl-sum} and~\ref{eq:LP-cl-cover} (see \Cref{sec:prelim}), every point $q \in P$ that $p$ is assigned to (i.e., $x_{p, q} = \xb[\beta^{*}]_{p, q} > 0$) satisfies $\dist^{z}(p, q) \le \mu_{p}$.
\end{enumerate}
By applying \Cref{claim:additive-error} and $\asncost(\Bxb[\beta^{*}]) = \cost(\Byb[\beta^{*}])$, we can derive 
\begin{align}
    \cost(\Byalg)
    & ~\le~ \cost(\Byb[\beta^{*}]) + \sum_{p, q \in P} [\xb[\beta^{*}]_{p, q} - \yalg_{q}]^{+} \cdot \gamma^{*} \cdot \beta^{-1} \cdot \rb_{p}\\
    \mr{$\Brb\le \Brb[\beta^*]$}
    & ~\le~ \cost(\Byb[\beta^{*}]) + \sum_{p, q \in P} [\xb[\beta^{*}]_{p, q} - \yalg_{q}]^{+} \cdot \gamma^{*} \cdot \beta^{-1} \cdot \rb[\beta^{*}]_{p}.
    \label{eq:correctness-lmp-fl-additive-error}
\end{align}
We divide all pairs $P \times P = \calF \sqcup \calC$ into the ``far'' pairs $\calF$ and the ``close'' pairs $\calC$, as follows:
\begin{align*}
    \calF & ~:=~ \big\{(p, q) \in P \times P:\ \beta^{*} \cdot \dist^{z}(p, q) > \tfrac{1}{3\cdot 2^{z-1}} \cdot \rb[\beta^{*}]_{p}\big\},\\
    \calC & ~:=~ \big\{(p, q) \in P \times P:\ \beta^{*} \cdot \dist^{z}(p, q) \le \tfrac{1}{3\cdot 2^{z-1}} \cdot \rb[\beta^{*}]_{p}\big\}.
\end{align*}
Below, \Cref{claim:boundF1,claim:boundC1} reason about these two kinds of pairs separately.

\begin{claim}
\label{claim:boundF1}
$\sum_{(p, q) \in \calF} [\xb[\beta^{*}]_{p, q} - \yalg_{q}]^{+} \cdot \rb[\beta^{*}]_{p} \le \frac{3}{4} \cdot \cost(\Byb[\beta^{*}])$.
\end{claim}

\begin{proof}
Each far pair $(p, q) \in \calF$ satisfies
$\rb[\beta^{*}]_{p} \le 3 \cdot 2^{z-1} \cdot \beta^{*} \cdot \dist^{z}(p, q)
= \frac{3}{4} \cdot \dist^{z}(p, q)$. 
Hence, we have 
$\sum_{(p, q) \in \calF} [\xb[\beta^{*}]_{p, q} - \yalg_{q}]^{+} \cdot \rb[\beta^{*}]_{p}
\le \sum_{(p, q) \in \calF} \xb[\beta^{*}]_{p, q} \cdot \frac{3}{4} \cdot \dist^{z}(p, q)
\le \frac{3}{4} \cdot \asncost(\Bxb[\beta^{*}])
= \frac{3}{4} \cdot \cost(\Byb[\beta^{*}])$.
\end{proof}

\begin{claim}
\label{claim:boundC1}
$\sum_{(p, q) \in \calC} [\xb[\beta^{*}]_{p, q} - \yalg_{q}]^{+} \cdot  \rb[\beta^{*}]_{p}
\le 3 \cdot 2^{z-1} \cdot \lambda \cdot (\|\Byb[\beta^{*}]\|_{1} - \|\Byalg\|_{1})$.
\end{claim}

\begin{proof}
Each close pair $(p, q) \in \calC$ satisfies
$\beta^{*} \cdot \dist^{z}(p, q) \le \frac{1}{3 \cdot 2^{z-1}} \cdot \rb[\beta^{*}]_{p}$. In combination with (\Cref{lem:properties-rp-relation1}) $\rb[\beta^{*}]_{p} \le 2^{z-1} \cdot (\rb[\beta^{*}]_{q} + \beta^{*} \cdot \dist^{z}(p, q))$, we can deduce that
$\rb[\beta^{*}]_{p} \le 3 \cdot 2^{z-2} \cdot \rb[\beta^{*}]_{q}$
and then $\beta^{*} \cdot \dist^{z}(p, q) \le \frac{1}{2} \cdot \rb[\beta^{*}]_{q}$.
Therefore, we have $\calC \subseteq \{(p, q) \in P \times P:\ p \in \Ball[\beta^{*}]_{q}(\rb[\beta^{*}]_q/2)\}$.
It follows that
\begin{align*}
    \sum_{(p, q) \in \calC} [\xb[\beta^{*}]_{p, q} - \yalg_{q}]^{+} \cdot \rb[\beta^{*}]_{p}
    & ~\le~ \sum_{(p, q) \in P \times P: p \in \Ball[\beta^{*}]_{q}(\rb[\beta^{*}]_q/2)} [\xb[\beta^{*}]_{p, q} - \yalg_{q}]^{+} \cdot \rb[\beta^{*}]_{p}\\
    \mr{$\xb[\beta^{*}]_{p, q} \le \yb[\beta^{*}]_{q}$ (feasibility)}
    & ~\le~ \sum_{(p, q) \in P \times P: p \in \Ball[\beta^{*}]_{q}(\rb[\beta^{*}]_q/2)} [\yb[\beta^{*}]_{q} - \yalg_{q}]^{+} \cdot \rb[\beta^{*}]_{p}\\
    & ~\le~ \sum_{(p, q) \in P \times P: p \in \Ball[\beta^{*}]_{q}(\rb[\beta^{*}]_q/2)} [\yb[\beta^{*}]_{q} - \yalg_{q}]^{+} \cdot 3 \cdot 2^{z-2} \cdot \rb[\beta^{*}]_{q}\\
    & ~\le~ \sum_{q \in P} |\Ball[\beta^{*}]_{q}(\rb[\beta^{*}]_q/2)| \cdot [\yb[\beta^{*}]_{q} - \yalg_{q}]^{+} \cdot 3 \cdot 2^{z-2} \cdot \rb[\beta^{*}]_{q}\\
    \mr{\Cref{lem:properties-rp-size}}
    & ~\le~ \sum_{q \in P} (2\lambda / \rb[\beta^{*}]_{q}) \cdot [\yb[\beta^{*}]_{q} - \yalg_{q}]^{+} \cdot 3 \cdot 2^{z-2} \cdot \rb[\beta^{*}]_{q}\\
    \mr{$\Byalg \le \Byb[\beta^{*}]$}
    & ~=~ 3 \cdot 2^{z-1} \cdot \lambda \cdot (\|\Byb[\beta^{*}]\|_{1} - \|\Byalg\|_{1}).
\end{align*}
This completes the proof of \Cref{claim:boundC1}.
\end{proof}

Applying \Cref{lem:fractional-MP,claim:boundF1,claim:boundC1} to \eqref{eq:correctness-lmp-fl-additive-error} finishes the proof of \Cref{lem:FMP-distance-distortion}, as follows:
\begin{align*}
    \cost(\Byalg)
& ~\le~ \cost(\Byb[\beta^{*}])
    + \gamma^{*} \cdot \beta^{-1} \cdot \sum_{(p, q) \in \calF \sqcup \calC} [\xb[\beta^{*}]_{p, q} - \yalg_{q}]^{+} \cdot \rb[\beta^{*}]_{p}\\
    \mr{\Cref{claim:boundF1,claim:boundC1}}
    & ~\le~ \cost(\Byb[\beta^{*}])
    + \gamma^{*} \cdot \beta^{-1} \cdot
    \left(\tfrac{3}{4} \cdot \cost(\Byb[\beta^{*}])
    + 3 \cdot 2^{z-1} \cdot \lambda \cdot (\|\Byb[\beta^{*}]\|_{1} - \|\Byalg\|_{1})\right)\\
\mr{$\gamma^{*} = 2^{O(z^2)}$}
    & ~\le~ 2^{O(z^2)}\cdot \beta^{-1} \cdot
    \left(\cost(\Byb[\beta^{*}])
    + \lambda \cdot \big(\|\Byb[\beta^{*}]\|_{1} - \|\Byalg\|_{1}\big)\right)\\
    \mr{\eqref{eq:LMP-Bybbetastar}}
    & ~\le~ 2^{O(z^2)}\cdot \beta^{-1}\cdot  \left(2^{O(z^2)}\cdot \big(\OPTfl[\lambda] - \lambda\cdot \|\Byb[\beta^{*}]\|_{1}\big)+\lambda \cdot \big(\|\Byb[\beta^{*}]\|_{1} - \|\Byalg\|_{1}\big)\right)\\
& ~\le~ 2^{O(z^2)} \cdot \beta^{-1} \cdot \left(\OPTfl[\lambda] - \lambda \cdot \|\Byalg\|_{1}\right),
\end{align*}
where the last step holds since $\OPTfl[\lambda] - \lambda\cdot \|\Byb[\beta^{*}]\|_{1} \ge 0$ (\Cref{lem:fractional-MP}) and $\|\Byb[\beta^{*}]\|_{1} - \|\Byalg\|_{1} \ge 0$.
\end{proof}

 \section{\texorpdfstring{$O(1)$-Approximation with $k+O(n^{1-\epsilon})$ Centers}{O(1)-Approximation with k+O(n^(1-epsilon)) Centers}}
\label{sec:bicriteria}

In this section, we design an MPC algorithm that computes a \emph{bicriteria} $O(1)$-approximation for \textsc{Euclidean} \kzC (\Cref{thm:bicriteria}) using  $O(n^{1-\epsilon})$ additional centers.

\begin{theorem}[{Bicriteria Approximation}]\label{thm:bicriteria}
    For any constant $\epsilon\in(0,1)$ and integer $z\ge 1$, there is a randomized MPC algorithm for \textsc{Euclidean} \kzC that
    \begin{itemize}
        \item 
        takes as input a dataset $P\subseteq \R^d$, a candidate center set $F\subseteq \R^d$, an integer $1\le k\le |F|$, and upper bounds $n\ge |P|+|F|$ and $\Delta\ge \Delta(P\cup F)$ on the total number of points and aspect ratio, respectively;
        \item computes, with probability $1-\frac{1}{\poly(n)}$, a center set $C\subseteq F$ satisfying 
        \begin{equation*}
            |C|\le k + 1 + |F|/n^{\epsilon} \le k + O(n^{1-\epsilon}),\quad \cost(P, C) \le 2^{O(z^2)} \cdot \epsilon^{-O(z)}\cdot \OPTclF(P, F, k).
        \end{equation*}
\end{itemize}
    With local memory $s\ge \poly(d\log(n\Delta))$ per machine, the algorithm runs in $O(\log_s n)$ rounds and uses $2^{O(\epsilon d)}\cdot n^{1+O(\epsilon)}\cdot \poly(d\log(n\Delta))$ total memory.
\end{theorem}
Here, \Cref{thm:bicriteria} is stated in a slightly more general form that allows $n$ to be any upper bound on $|P|+|F|$, for technical reasons that arise in \Cref{sec:clustering}; for example, we may run \Cref{thm:bicriteria} on a much smaller input $(P',F')$, where $P'\subseteq P$ and $F'\subseteq F$, while still requiring the size guarantee $|C|\le k+1+|F'|/n^{\epsilon}$ with respect to the original size $n = |P| + |F|$.

Throughout this section, we only consider a single dataset $P$. Hence, for simplicity, we henceforth write $\cost(P,C)$ as $\cost(C)$ for any center set $C$, $\cost(P,\By)$ as $\cost(\By)$ for any fractional solution $\By$, and $\asncost(P,\Bx)$ as $\asncost(\Bx)$ for any fractional assignment $\Bx$.

\begin{proof}
    At a high level, the algorithm in \Cref{thm:bicriteria} consists of two steps. First, we compute a fractional solution that opens exactly $k$ total units of centers, using the fractional LMP algorithm presented in \Cref{sec:FL}. We then round this fractional solution to a center set while introducing at most $1+|F|/n^{\epsilon}$ additional centers. 
    We present formal statements of the algorithms for these two steps below.

    \begin{restatable}[Computing a Fractional Solution]{lemma}{Fractional}
        \label{lem:fractional}
        For any constant $\epsilon\in(0,1)$ and integer $z\ge 1$,
        there is a randomized MPC algorithm for \textsc{Euclidean} \kzC that
        \begin{itemize}
            \item
            takes as input a dataset $P\subseteq \R^d$, a candidate center set $F\subseteq \R^d$, an integer $1\le k\le |F|$, and upper bounds $n\ge |P|+|F|$ and $\Delta\ge \Delta(P\cup F)$ on the total number of points and aspect ratio, respectively;

            \item computes, with probability $1-\frac{1}{\poly(n)}$, a fractional solution $\By \in \R_{\ge 0}^{F}$ satisfying 
            \begin{equation*}
                \|\By\|_1 ~=~ k, \quad \cost(\By) \le 2^{O(z^2)} \cdot \epsilon^{-O(z)}\cdot \OPTclF(P, F, k)
\end{equation*}
\end{itemize}
        With local memory $s\ge \poly(d\log(n\Delta))$ per machine, the algorithm runs in $O(\log_s n)$ rounds and uses $2^{O(\epsilon d)} \cdot n\cdot \poly(d\log(n\Delta))$ total memory.
    \end{restatable}
    \begin{proof}
        The proof can be found in \Cref{sec:fractional}.
    \end{proof}

    \begin{restatable}[Rounding with Additional Centers]{lemma}{Rounding}\label{lem:rounding}
        For any constant $\epsilon\in(0,1)$ and integer $z\ge 1$,
        there is a randomized MPC algorithm for \textsc{Euclidean} \kzC that
        \begin{itemize}
            \item 
            takes as input a dataset $P\subseteq \R^d$, a candidate center set $F\subseteq \R^d$, an integer $1\le k\le |F|$, a fractional solution $\By\in \R_{\ge 0}^{F}$ with $\|\By\|_1 = k$, and upper bounds $n\ge |P|+|F|$ and $\Delta\ge \Delta(P\cup F)$ on the total number of points and aspect ratio, respectively;

            \item computes, with probability $1-\frac{1}{\poly(n)}$, a center set $C\subseteq F$ satisfying
            \begin{equation*}
                |C| ~\le~ k + 1 + |F|/n^{\epsilon}, \quad
                \cost(C) \le (z/\epsilon)^{O(z)}\cdot \cost(\By). 
            \end{equation*}
        \end{itemize}
        With local memory $s\ge \poly(d\log(n\Delta))$ per machine,
        the algorithm runs in $O(\log_s n)$ rounds and uses $2^{O(\epsilon d)}\cdot n^{1+O(\epsilon)}\cdot\poly(d\log(n\Delta))$ total memory.
    \end{restatable}
    \begin{proof}
        The proof can be found in \Cref{sec:rounding}.
    \end{proof}

    The algorithm for \Cref{thm:bicriteria} is then obtained straightforwardly by combining the above two algorithms:
    we first run the algorithm from \Cref{lem:fractional} to compute a fractional solution $\By\in \R_{\ge 0}^{F}$, and then apply the algorithm from \Cref{lem:rounding} to obtain a solution $C\subseteq F$. 
    The size bound $|C|\le k + 1 + |F|/n^{\epsilon}$ follows directly from \Cref{lem:rounding}, and
    the cost satisfies
    \begin{equation*}
    \cost(C)
    ~\le~
    (z/\epsilon)^{O(z)}\cdot \cost(\By)
    ~\le~
    2^{O(z^2)}\cdot \epsilon^{-O(z)}\cdot
    \OPTclF(P, F, k).
    \end{equation*}
    The round complexity and total memory are at most the sums of the corresponding bounds in \Cref{lem:fractional} and \Cref{lem:rounding}, and are thus as stated in \Cref{thm:bicriteria}.
    Finally, both the algorithm in \Cref{lem:fractional} and that in \Cref{lem:rounding} succeed with probability $1-\frac{1}{\poly(n)}$. Therefore, by a union bound, the overall algorithm succeeds with probability at least $1-\frac{1}{\poly(n)}$.
    This finishes the proof of \Cref{thm:bicriteria}.
\end{proof}

\subsection{Computing a Fractional Solution (Proof of \Cref{lem:fractional})}
\label{sec:fractional}

In this subsection, we prove \Cref{lem:fractional}, which is restated below for convenience.
\Fractional*

\begin{proof}

We divide the proof into two parts. First, we consider the special case where $F=P$, which is exactly the setting considered by our algorithm for \textsc{Power-$z$ Facility Location} in \Cref{thm:lmp-fl}. Then, we apply a postprocessing step to transform the fractional solution $\By\in\R_{\ge 0}^{P}$ into a fractional solution in $\R_{\ge 0}^{F}$, as desired.

\paragraph{Algorithm for the Case $F=P$.}

Our MPC algorithm for the case $F = P$, shown in \Cref{alg:kzC}, is adapted from the algorithm of~\cite[Section~14]{BhattacharyaCGL24}, which invokes an LMP algorithm for \textsc{Power-$z$ Facility Location} as a subroutine. 
We replace that subroutine with our MPC algorithm from \Cref{thm:lmp-fl}.
Below, we describe and analyze \Cref{alg:kzC} assuming $2\le k\le |P|$, as the algorithm cannot handle the case $k = 1$ (specifically, the index $\ell^{*}$ computed in Line~\ref{alg:kzC:istar} may not exist).
For the case $k = 1$, we instead use a much simpler uniform sampling-based algorithm, which we discuss later.

At a high level, this algorithm invokes \Cref{thm:lmp-fl} with different values of $\lambda$ to yield a series of fractional center sets $\{\By^{(\ell)}\}_{\ell \in [0:L]}$. 
It then identifies two particular fractional center sets $\By^{(\ell^{*}-1)}$ and $\By^{(\ell^{*})}$, such that $\|\By^{(\ell^{*}-1)}\|_{1} \ge k$ and $\|\By^{(\ell^{*})}\|_{1} \le k$, and outputs their linear combination $\By$ satisfying $\|\By\|_1 = k$.
The LMP approximation guarantee established in \Cref{thm:lmp-fl} is key to showing that $\By$ is a good fractional solution to \kzC.
In what follows, we assume without loss of generality that for every $p\neq q\in (P\cup F)$, $2\le \dist(p,q)\le O(\Delta)$ (as in the proof of \Cref{thm:lmp-fl}).

\begin{algorithm}[h]
\DontPrintSemicolon
\caption{Computing a fractional solution to \textsc{Euclidean} \kzC}
\label{alg:kzC}

$L \gets O(\log(n\Delta))$ (with a sufficiently large hidden constant), $\lambda_{\ell} \gets 2^{\ell z}$ for every $\ell\in [0:L]$.
\label{alg:kzC:L}\;

$\By^{(0)} \gets \ones^{P}$ (i.e., $y^{(0)}_p\gets 1$ for every point $p\in P$)
\label{alg:kzC:y0}\;

for each $\ell\in [L]$, $\By^{(\ell)} \gets$ the fractional center set computed by \Cref{thm:lmp-fl} with $\lambda = \lambda_\ell$
\label{alg:kzC:yi}\;

find an arbitrary $\ell^{*} \in [L]$ such that {$\|\By^{(\ell^{*}-1)}\|_{1} \ge k$ and $\|\By^{(\ell^{*})}\|_{1} \le k$}
\label{alg:kzC:istar}\;
\tcp*{The existence of such an $\ell^{*}$ is guaranteed by \Cref{obs:lambda0,obs:lambdaL}.}

compute $\alpha\in [0,1]$ such that $\alpha\cdot \|\By^{(\ell^{*}-1)}\|_{1} + (1-\alpha)\cdot \|\By^{(\ell^{*})}\|_{1} = k$
\label{alg:kzC:alpha}\;

\Return $\By\gets \alpha\cdot \By^{(\ell^{*} - 1)} + (1-\alpha)\cdot \By^{(\ell^{*})}$
\label{alg:kzC:return}\;
\end{algorithm}

The correctness analysis of \Cref{alg:kzC} is (almost) identical to that in~\cite[Section~14]{BhattacharyaCGL24}, which is summarized in \Cref{lem:correctness-kzC} in our notation.\footnote{While the analysis in~\cite[Section~14]{BhattacharyaCGL24} only considers the case $z = 1$, it can be easily extended to general $z \ge 1$, as stated in \Cref{lem:correctness-kzC}.}
For completeness, we provide a full proof in Appendix~\ref{sec:missing-proof-value}.

\begin{lemma}[Correctness of \Cref{alg:kzC}]
    \label{lem:correctness-kzC}
    Assume that $2\le k\le |P|$.
    \Cref{alg:kzC} returns a fractional center set $\By\in \R_{\ge 0}^{P}$ such that $\|\By\|_{1} = k$ and $\cost(\By) \le 2^{O(z^2)} \cdot \epsilon^{-2z} \cdot \OPTclF(P,P, k)$.
\end{lemma}
\begin{proof}
    The proof can be found in Appendix~\ref{sec:missing-proof-value}.
\end{proof}

\paragraph{MPC Implementation of \Cref{alg:kzC}.}
Line~\ref{alg:kzC:L} is trivial.
Lines~\ref{alg:kzC:y0} and~\ref{alg:kzC:yi} for computing the fractional center sets $\{\By^{(\ell)}\}_{\ell \in [0:L]}$ can be easily parallelized; the round complexity and local memory requirement are the same as those of \Cref{thm:lmp-fl}, while the total memory is that of \Cref{thm:lmp-fl} multiplied by $L = \Theta(\log(n\Delta))$ (due to parallelization).
Line~\ref{alg:kzC:istar} requires the level $\ell^{*}$, which can be found by first computing each norm $\|\By^{(\ell)}\|_1$ by aggregation and then gathering all pairs $(\ell, \|\By^{(\ell)}\|_{1})$ at a single (leader) machine;
this requires local memory $s \ge \Theta(L + 1) = O(\log(n\Delta))$.
Line~\ref{alg:kzC:alpha} requires the value $\alpha$, which can be found locally on the leader machine. 
After broadcasting $\alpha$ and $\ell^{*}$ and bringing the entries $y_p^{(\ell^{*}-1)}$ and $y_p^{(\ell^{*})}$ to the same machine for every point $p\in P$, the computation of $\By = \alpha \cdot \By^{(\ell^{*} - 1)} + (1 - \alpha) \cdot \By^{(\ell^{*})}$ is straightforward; this computation takes $O(\log_s n)$ rounds.

In sum, the above implementation of \Cref{alg:kzC} requires local memory $s\ge \poly(d \log (n\Delta))$, runs in $O(\log_s n)$ rounds, and uses total memory $2^{O(\epsilon d)} \cdot n\cdot \poly(d\log (n\Delta))$. 
We also note that the above implementation of \Cref{alg:kzC} (for $k\ge 2$) is deterministic.

\paragraph{Special Case where $k = 1$.}
When $k = 1$, it is well known that selecting a center uniformly at random from $P$ provides a $2^{O(z)}$-approximation in expectation (see, e.g.,~\cite{DBLP:conf/soda/ArthurV07}).
We can boost the success probability to $1 - \frac{1}{\poly(n)}$ by independently sampling $O(\log n)$ such centers in parallel and selecting the best one as the final solution. 
Since both uniform sampling and computing the cost of a singleton center set can be efficiently implemented in the MPC model (via broadcast and converge-cast), this case can be addressed within the claimed number of rounds and space.

\paragraph{Extending to the Case $F\neq P$.}
We first handle the case $k\ge |P|$. For every $p\in P$, compute an $O(\epsilon^{-1})$-ANN $\pi(p)\in F$.
Let $S:=\{\pi(p):p\in P\}$. We place one unit of opening at every center in $S$ and at any $k-|S|$ additional centers in $F\setminus S$. Such additional centers exist because $|S|\le |P|\le k\le |F|$. This gives total opening $k$ and cost at most $O(\epsilon^{-z})\cdot\cost(F)\le O(\epsilon^{-z})\cdot\OPTclF(P,F,k)$. Below, we focus on the case $1\le k<|P|$.

We can extend the preceding algorithm to an arbitrary candidate center set $F\subseteq\R^d$, as follows. We first compute a fractional solution $\By\in\R_{\ge0}^{P}$ using the preceding algorithm. Then, for each $p\in P$, we compute an $O(\epsilon^{-1})$-ANN $\pi(p)\in F$ of $p$ in $F$, and define a new fractional solution $\By'\in\R_{\ge0}^{F}$ by
\begin{equation*}
    \forall c\in F,\qquad
    y'_c:=\sum_{p\in P:\pi(p)=c}y_p.
\end{equation*}
We return the resulting fractional solution $\By'$ as the final output.

The MPC implementation is straightforward. To compute $\{\pi(p)\}_{p\in P}$, we use the deterministic approximate nearest-neighbor primitive from~\cite[Section~3.1]{CzumajGJK024}.
Given $\{\pi(p)\}_{p\in P}$, computing the new fractional solution $\By'$ reduces to an aggregation task, which can be implemented using standard MPC sorting and converge-cast (see, e.g.,~\cite[Lemma~5.10]{Cohen-addadKP26}).
The bottleneck is therefore the approximate nearest-neighbor computation. Hence, the implementation requires local memory $s\ge \poly(d\log(n\Delta))$, runs in $O(\log_s n)$ rounds, and uses total memory $2^{O(\epsilon d)}\cdot n\cdot \poly(d\log(n\Delta))$.

We now establish the correctness of this extension.
Clearly, the construction of $\By'$ preserves the total fractional mass, and thus $\|\By'\|_1 = \|\By\|_1 = k$.
We then bound the cost of $\By'$.
Let $\Bx\in \R_{\ge 0}^{P\times P}$ be an optimal fractional assignment with respect to $\By$, i.e., $\asncost(\Bx) = \cost(\By)$, and define
\begin{equation*}
    \forall p\in P,\ c\in F,\qquad
    x'_{p,c}:=\sum_{q\in P:\pi(q)=c}x_{p,q}.
\end{equation*}
Clearly, 
this assignment is feasible with respect to $\By'$.
For every $p,q\in P$, by the triangle inequality,
\begin{align*}
    \dist(p,\pi(q))
    ~&\le~ \dist(p,q)+\dist(q,\pi(q)) ~\le~ \dist(p, q) + O(\epsilon^{-1})\cdot \dist(q, F)\\
    &\le~ O(\epsilon^{-1})\cdot (\dist(p, q) + \dist(p, F)).
\end{align*}
Therefore, by the generalized triangle inequality (\Cref{lem:triangle-inequality}), we obtain
\begin{align}
    \cost(\By') ~\le~ \asncost(\Bx') ~&=~\sum_{p\in P}\sum_{c\in F} x'_{p, c}\cdot \dist^z(p, c)\nonumber\\
    &=~ \sum_{p,q\in P}x_{p,q}\cdot\dist^z(p,\pi(q))\nonumber\\
    &\le~ 2^{O(z)}\epsilon^{-z}\cdot
        \Big(\sum_{p,q\in P}x_{p,q}\cdot\dist^z(p,q)
        +
        \sum_{p,q\in P}x_{p,q}\cdot \dist^z(p,F)\Big)\nonumber\\
    &=~ 2^{O(z)}\epsilon^{-z}\cdot \Big(\cost(\By) +\cost(F)\Big) \nonumber\\
    &\le~ 2^{O(z^2)} \cdot \epsilon^{-3z} \cdot \OPTclF(P,P, k) + 2^{O(z)}\epsilon^{-z}\cdot \OPTclF(P, F, k), \label{eq:costnewy}
\end{align}
where the last step follows from $\cost(\By)\le 2^{O(z^2)}\cdot\epsilon^{-2z}\cdot\OPTclF(P,P,k)$ by \Cref{lem:correctness-kzC} (or the sampling guarantee above when $k=1$), and from $\cost(F)\le \OPTclF(P,F,k)$ by the definition of $\OPTclF(P,F,k)$.

Finally, we compare $\OPTclF(P,P,k)$ with $\OPTclF(P,F,k)$, as follows.
Let $C^\star\subseteq F$ be an optimal size-$k$ center set, i.e., $\cost(C^{\star}) = \OPTclF(P, F, k)$. Let $C' :=\{\mathrm{NN}(c, P): c\in C^{\star}\} \subseteq P$.
Then, $|C'|\le |C^{\star}| \le k$, and thus $\OPTclF(P, P, k) \le \cost(C')$.
For every point $p\in P$, letting $c:= \mathrm{NN}(p, C^{\star})$, we have
\begin{equation*}
    \dist(p, C') ~\le~ \dist(p, \mathrm{NN}(c, P)) ~\le~ \dist(p, c) + \dist(c, P) ~\le~ 2\dist(p, c) ~=~ 2\dist(p, C^{\star}).
\end{equation*}
Therefore, 
\begin{equation*}
    \OPTclF(P, P, k) ~\le~ \cost(C') ~\le~ 2^{z}\cdot \cost(C^{\star}) ~=~ 2^{z}\cdot \OPTclF(P, F, k). 
\end{equation*}
Plugging this into~\eqref{eq:costnewy} gives
\begin{equation*}
    \cost(\By')~\le~2^{O(z^2)}\epsilon^{-3z}\cdot\OPTclF(P,F,k),
\end{equation*}
which verifies the correctness of $\By'$.
\end{proof}

\subsection{Rounding to $k + O(n^{1-\epsilon})$ Centers (Proof of \Cref{lem:rounding})}
\label{sec:rounding}

In this subsection, we prove \Cref{lem:rounding}, which is restated below for convenience.

\Rounding*

Throughout this subsection, we assume without loss of generality that $n^{\epsilon} \ge 1$ is an integer.
The proof of \Cref{lem:rounding} spans \Cref{sec:rounding:algorithm,sec:rounding:MPC}. Specifically, in \Cref{sec:rounding:algorithm}, we present the algorithm in the offline setting (without considering its MPC implementation) and analyze the correctness of the returned center set, showing that it satisfies the desired size and cost guarantees in \Cref{lem:rounding}. Then, in \Cref{sec:rounding:MPC}, we show how to implement the algorithm in the MPC setting with the stated round complexity and total memory, thereby completing the proof of \Cref{lem:rounding}.

\subsubsection{Algorithm Description and Analysis}
\label{sec:rounding:algorithm}

We design the offline algorithm for \Cref{lem:rounding} in three steps.
In the first step, we design a weaker algorithm that returns a center set $C\subseteq F$ for which the size and cost guarantees hold only \emph{in expectation}; this step applies the exponential-clock rounding scheme of~\cite{AnNS17}.
In the second step, we strengthen these expected guarantees to hold with constant probability via \emph{bi-point rounding}. In the third step, we further boost the success probability to high probability via simple repetition.

\paragraph{Step 1: Expected Guarantees via Exponential Clocks.}
An offline version of the algorithm for the first step is presented in \Cref{alg:coarse-rounding}.
The algorithm starts by computing an auxiliary fractional solution $\By'$ whose cost and total opening are comparable to those of $\By$, together with a feasible assignment $\Bx'\sim\By'$, and then applies the exponential-clock rounding scheme of~\cite{AnNS17} on $(\Bx', \By')$.
Crucially, both $\By'$ and $\Bx'$ take values only in $\{0,1/n^{\epsilon}\}$, which is important for both (i)~satisfying the requirement of the rounding scheme of~\cite{AnNS17} (see the proof of \Cref{lem:correctness-coarse-rounding}) and (ii)~facilitating an efficient MPC implementation.

\begin{algorithm}[h]
\DontPrintSemicolon
\caption{Rounding with expected size and cost guarantees}
\label{alg:coarse-rounding}

construct a multiset $Q$ of points from $F$ and a fractional solution $\By' := \frac{1}{n^{\epsilon}}\cdot \mathbf{1}^{Q}$ such that $\|\By'\|_1\le k+|F|/n^{\epsilon}$ and $\cost(\By')\le\cost(\By)$
\label{alg:coarse-rounding:yprime}\label{alg:corase-rounding:yprime}\;

compute an assignment $\Bx'\in \{0, 1/n^{\epsilon}\}^{P\times Q}$ such that $\Bx' \sim \By'$, $\sum_{q\in Q}x'_{p,q}=1$ for every $p\in P$, and $\asncost(\Bx') \le O(\epsilon^{-z})\cdot \cost(\By')$
\label{alg:coarse-rounding:xprime}\label{alg:corase-rounding:xprime}\;

\tcc{Lines~\ref{alg:coarse-rounding:alpha} to \ref{alg:coarse-rounding:return}: exponential-clock rounding~\cite{AnNS17} on $(\Bx', \By')$}

independently sample $\alpha_{q} \sim \mathrm{Exp}(1/n^{\epsilon})$ 
for every $q\in Q$
\label{alg:coarse-rounding:alpha}\;

independently sample $\beta_p \sim \mathrm{Exp}(1)$ for every  $p\in P$
\label{alg:coarse-rounding:beta}\;

for every $p\in P$, $f(p) \gets \argmin_{q\in Q: x'_{p, q} > 0} \alpha_{q}$
\label{alg:coarse-rounding:f}\;

\Return $C\gets$ the set $\big\{f(p): p\in P, p = \argmin_{p'\in P: x'_{p', f(p)} > 0} \beta_{p'} \big\} \subseteq F$ 
\label{alg:coarse-rounding:return}\;

\end{algorithm}

The correctness of \Cref{alg:coarse-rounding} is established in the following claim.

\begin{claim}[Correctness of \Cref{alg:coarse-rounding}]
    \label{lem:correctness-coarse-rounding}
    The center set $C\subseteq F$ returned by \Cref{alg:coarse-rounding} satisfies $\E[|C|] \le k + |F|/n^{\epsilon}$ and $\E[\cost(C)] \le (z/\epsilon)^{O(z)}\cdot \cost(\By)$.
\end{claim}
\begin{proof}
    We first have the following \Cref{claim:generalize-ANS17}, which establishes the correctness of the exponential-clock rounding scheme on $(\Bx',\By')$ in Lines~\ref{alg:coarse-rounding:alpha} to \ref{alg:coarse-rounding:return}.
    This claim follows by essentially the same argument as in~\cite[Section 4]{AnNS17}, with their analysis generalized from the case $z=1$ to arbitrary $z\ge 1$; for completeness, we provide a proof in Appendix~\ref{sec:proof-correctness-coarse-rounding}.
    We also note that the analysis of~\cite[Section 4]{AnNS17} requires the input $(\Bx',\By')$ to the exponential-clock rounding scheme to satisfy (i) $\Bx'\sim\By'$ is feasible, and (ii) $x'_{p,q}\in \{0,y'_q\}$ for every $p\in P$ and $q\in Q$; both requirements are guaranteed by the constructions in Lines~\ref{alg:coarse-rounding:yprime} and~\ref{alg:coarse-rounding:xprime}.\footnote{
        The algorithm in~\cite[Section~3]{AnNS17} is designed for a dynamic setting over $T$ time steps and additionally involves a vector $\bm{z}$ that accounts for changes in $\Bx'$. In our setting, there is only a single time step, i.e., $T=1$, and hence the vector $\bm{z}$ is not needed.
    }

    \begin{claim}\label{claim:generalize-ANS17}
        $\E[|C|] \le \|\By'\|_1$ and $\E[\cost(C)] \le (z/\epsilon)^{O(z)}\cdot \cost(\By')$.
    \end{claim}
    \begin{proof}
        The proof can be found in Appendix~\ref{sec:proof-correctness-coarse-rounding}.
    \end{proof}
    
    Combining \Cref{claim:generalize-ANS17} with the guarantees that $\|\By'\|_1\le k+|F|/n^{\epsilon}$ and $\cost(\By')\le\cost(\By)$ from Line~\ref{alg:coarse-rounding:yprime} completes the proof of \Cref{lem:correctness-coarse-rounding}.
\end{proof}

\paragraph{Step 2: Constant-probability Guarantee via Bi-point Rounding.} 
The solution returned by \Cref{alg:coarse-rounding} satisfies both the size and cost bounds only in expectation, as stated in \Cref{lem:correctness-coarse-rounding}. 
Using \Cref{alg:coarse-rounding} as a subroutine, we now give an algorithm that, with constant probability, returns a solution $C$ satisfying $|C|\le k+O(n^{1-\epsilon})$ and $\cost(C)\le (z/\epsilon)^{O(z)}\cdot\cost(\By)$.

The algorithm is presented in \Cref{alg:strengthen-size}. The key idea is that, by repeating \Cref{alg:coarse-rounding} sufficiently many times, we can obtain a \emph{bi-point solution} with \emph{worst-case} size and cost guarantees with constant probability. 
Here, a bi-point solution is represented by a tuple $(C_1,C_2,a)$, where $C_1, C_2$ are two center sets and $a\in[0,1]$; its size and cost are defined as $a|C_1|+(1-a)|C_2|$ and $a\cost(C_1)+(1-a)\cost(C_2)$, respectively.
Provided that we obtain a bi-point solution, we then apply the classical bi-point rounding algorithm of~\cite{JainV01} to obtain a center set whose size is at most (the ceiling of) the size of the bi-point solution, while its expected cost is within an $\epsilon^{-O(z)}$ factor of the cost of the bi-point solution.
In this algorithm description, we do not give the details of the bi-point rounding procedure, and we note that an efficient MPC implementation has already been given in~\cite{Cohen-addadKP26} (see also \Cref{lem:bi-point-rounding}).

\begin{algorithm}[h]
\DontPrintSemicolon
\caption{Computing a solution with constant-probability guarantees}
\label{alg:strengthen-size}

independently run \Cref{alg:coarse-rounding} $h:=n^{2\epsilon}$ times to obtain solutions $C_1,C_2,\dots,C_h \subseteq F$
\label{alg:strengthen-size:run}\;

for every $i\in [h]$, compute $\Hat{\cost}(C_i)$ such that $\cost(C_i) \le \Hat{\cost}(C_i) \le O(\epsilon^{-z})\cdot \cost(C_i)$
\label{alg:strengthen-size:cost}\;

find a bi-point solution $(C_i,C_j,a)$ for some $i,j\in[h]$ (possibly equal) and $a\in[0,1]$ minimizing $a\Hat{\cost}(C_i)+(1-a)\Hat{\cost}(C_j)$ subject to $a|C_i|+(1-a)|C_j|\le k+2\cdot |F|/n^{\epsilon}$
\label{alg:strengthen-size:bi-point-solution}\;

\If{no such bi-point solution is found}{ \label{alg:strengthen-size:if}
    pick an arbitrary bi-point solution $(C_i, C_j, a)$
    \label{alg:strengthen-size:if1}\;
}

run the bi-point rounding~\cite{JainV01} on $(C_i,C_j,a)$ to obtain a solution $C\subseteq F$ such that 
(i)~$|C|\le  a|C_i|+(1-a)|C_j| + 1$, and 
(ii)~$\E[\cost(C)]\le \epsilon^{-O(z)}\cdot\bigl(a\cost(C_i)+(1-a)\cost(C_j)\bigr)$
\label{alg:strengthen-size:bi-point-rounding}\;

\Return $C$\;
\end{algorithm}

\begin{claim}[Correctness of \Cref{alg:strengthen-size}]
    \label{lem:correctness-strengthen-size}
    The center set $C\subseteq F$ returned by \Cref{alg:strengthen-size} satisfies the following guarantees with probability at least $1/4$:
    \begin{equation*}
        |C| ~\le~ k + 1 + 2\cdot |F| / n^{\epsilon}, \quad \text{and} \quad \cost(C) ~\le~ (z/\epsilon)^{O(z)}\cdot \cost(\By).
    \end{equation*}
\end{claim}
\begin{proof}
    By \Cref{lem:correctness-coarse-rounding}, for every $i\in[h]$, we have
    \begin{equation*}
        \E[|C_i|] ~\le~ k + |F|/n^{\epsilon},
        \quad
        \E[\cost(C_i)] ~\le~ (z/\epsilon)^{O(z)}\cdot \cost(\By).
    \end{equation*}
    We claim that the following event occurs with probability at least $1/2$; we prove this claim at the end of the proof:
    \begin{equation}\label{eq:event-bi-point}
        \begin{aligned}    
        \exists \text{ a bi-point solution $(C_i,C_j,a)$ s.t. } &a|C_i| + (1-a)|C_j| ~\le~ k + 2\cdot |F|/n^{\epsilon},\\
        &  a\Hat{\cost}(C_i)+(1-a)\Hat{\cost}(C_j) ~\le~ (z/\epsilon)^{O(z)}\cdot \cost(\By).
        \end{aligned}
    \end{equation}
    Conditioned on this event, Line~\ref{alg:strengthen-size:bi-point-solution} finds such a bi-point solution, after which Line~\ref{alg:strengthen-size:bi-point-rounding} computes a solution $C$ satisfying $|C|\le k + 1 + 2\cdot |F|/n^{\epsilon}$ and
    \begin{align*}
        \E[\cost(C)] ~&\le~  \epsilon^{-O(z)}\cdot\bigl(a\cost(C_i)+(1-a)\cost(C_j)\bigr)\\
        \mr{Line~\ref{alg:strengthen-size:cost}}
        &\le~ \epsilon^{-O(z)}\cdot\bigl(a\Hat{\cost}(C_i)+(1-a) \Hat{\cost}(C_j)\bigr) ~\le~ (z/\epsilon)^{O(z)}\cdot \cost(\By).
    \end{align*}
    Thus, by Markov's inequality, we conclude that $C$ satisfies the size and cost guarantees stated in \Cref{lem:correctness-strengthen-size} with conditional probability at least $1/2$ given Event~\eqref{eq:event-bi-point}. Overall, the algorithm returns a desired solution with probability at least $1/4$.

    It remains to prove that Event~\eqref{eq:event-bi-point} occurs with probability at least $1/2$.
    Define the random variables $v_s := \frac{1}{h}\sum_{i=1}^h |C_i|$ and $v_c := \frac{1}{h}\sum_{i=1}^h {\cost}(C_i)$. By \Cref{lem:correctness-coarse-rounding}, we have
    \begin{equation*}
        \E[v_s] ~\le~ k + |F| / n^{\epsilon}, \quad \E[v_c] ~\le~ (z/\epsilon)^{O(z)}\cdot \cost(\By).
    \end{equation*}
    We then examine the concentration of both random variables. For $v_s$, note that $0\le |C_i|\le |F|$ for every $i\in[h]$. By Hoeffding's inequality, we have
    \begin{align*}
        \Pr\big[v_s \ge k+2\cdot |F|/n^{\epsilon} \big] ~&\le~ \Pr\big[v_s - \E[v_s] \ge |F|/n^{\epsilon}\big]
        \\
        &\le~ \exp\Big(\frac{-2h\cdot (|F|/n^{\epsilon})^2}{|F|^2}\Big) ~\le~ e^{-2}.
    \end{align*}
    For $v_c$, we simply apply Markov's inequality to obtain $\Pr[v_c>4\E[v_c]]\le 1/4$. 
    By a union bound, the event
    \begin{equation*}
        \left\{v_s ~=~ \frac{1}{h}\sum_{i=1}^h |C_i| ~\le~ k + 2\cdot |F|/n^{\epsilon},\quad v_c ~=~  \frac{1}{h}\sum_{i=1}^h \cost(C_i) ~\le~ (z/\epsilon)^{O(z)}\cdot \cost(\By)\right\}
    \end{equation*}
    holds with probability at least $1-e^{-2} - 1/4 > 1/2$.
    Then, conditioning on this event, we show the existence of a bi-point
solution $(C_i,C_j,a)$ satisfying the bounds in~\eqref{eq:event-bi-point}.
Let $q_i := (|C_i|,\cost(C_i))\in\R^2$ for every $i\in[h]$. Since $(v_s,v_c)=\frac1h\sum_{i=1}^h q_i$, the point $(v_s,v_c)$ lies in the convex hull of $\{q_i\}_{i\in[h]}$. 
    Let $(v_s',v_c')$ be the point obtained by moving from $(v_s,v_c)$ as far as possible in the direction $(-1,-1)$ while remaining in the convex hull. Then $(v_s',v_c')$ lies on its boundary and satisfies $v_s'\le v_s$ and $v_c'\le v_c$.
    Since this convex hull lies in $\R^2$ and $(v_s',v_c')$ is on its boundary, there exist $i,j\in[h]$ and $a\in[0,1]$ such that $(v_s',v_c')=a q_i+(1-a)q_j$. In other words, $a|C_i|+(1-a)|C_j|\le k+2\cdot |F|/n^{\epsilon}$ and $a\cost(C_i)+(1-a)\cost(C_j)\le (z/\epsilon)^{O(z)}\cost(\By)$. 
    Finally, by Line~\ref{alg:strengthen-size:cost}, we obtain
    $a\Hat{\cost}(C_i)+(1-a)\Hat{\cost}(C_j)\le (z/\epsilon)^{O(z)}\cost(\By)$,
    with a larger hidden constant in the factor $(z/\epsilon)^{O(z)}$.
    Thus, $(C_i, C_j, a)$ satisfies~\eqref{eq:event-bi-point}.
    This finishes the proof.
\end{proof}

\paragraph{Step 3: High-probability Guarantee via Repetition (The Final Algorithm).}
Finally, we give an algorithm for \Cref{lem:rounding} that computes a solution $C$ such that both the size and cost bounds hold \emph{with high probability}.
The algorithm simply repeats \Cref{alg:strengthen-size} independently $O(\log n)$ times and returns a solution with the minimum estimated cost among those satisfying the size bound.

\begin{algorithm}[h]
\DontPrintSemicolon
\caption{Computing a solution with high-probability guarantees}
\label{alg:strengthen-cost}

independently run \Cref{alg:strengthen-size} $O(\log n)$ times (with a sufficiently large hidden constant)\;

for every resulting solution $C$, compute a value $\Hat{\cost}(C)$ as in Line~\ref{alg:strengthen-size:cost} of \Cref{alg:strengthen-size}
\label{alg:strengthen-cost:cost}\;

\If{there exists a resulting solution satisfying $|C|\le k + 1 + 2\cdot |F|/n^{\epsilon}$}{
    \Return the solution with the minimum estimated cost $\Hat{\cost}(C)$ among all such solutions
    \label{alg:strengthen-cost:return}\;
}
\Else{
    \Return an arbitrary subset of $F$
}

\end{algorithm}

\begin{claim}[Correctness of \Cref{alg:strengthen-cost}]
    \label{claim:correctness-strengthen-cost}
    With probability $1-\frac{1}{\poly(n)}$, \Cref{alg:strengthen-cost} returns a solution $C\subseteq F$ such that $|C|\le k+1 + 2\cdot |F|/n^{\epsilon}$ and $\cost(C)\le (z/\epsilon)^{O(z)}\cdot\cost(\By)$.
\end{claim}
\begin{proof}
By \Cref{lem:correctness-strengthen-size}, with probability at least $1/4$, a run of \Cref{alg:strengthen-size} returns a solution $C\subseteq F$ such that
\begin{equation*}
|C|\le k+1+2\cdot |F|/n^{\epsilon}
\quad\text{and}\quad
\cost(C)\le (z/\epsilon)^{O(z)}\cdot\cost(\By).
\end{equation*}
With a sufficiently large hidden constant in the $O(\log n)$ independent repetitions, at least one run returns such a solution with probability $1-\frac{1}{\poly(n)}$; denote one such solution by $C^\star$.
Let $C$ be the solution returned by \Cref{alg:strengthen-cost}.
The size bound $|C|\le k + 1 + 2\cdot |F|/n^{\epsilon}$ follows directly from Line~\ref{alg:strengthen-cost:return}. As for the cost bound, by the selection of $C$ in Line~\ref{alg:strengthen-cost:return} and the guarantee of $\Hat{\cost}(C)$ in Line~\ref{alg:strengthen-cost:cost}, it follows that
\begin{equation*}
    \cost(C) ~\le~\Hat{\cost}(C)
    ~\le~ \Hat{\cost}(C^\star)
    ~\le~ O(\epsilon^{-z})\cdot \cost(C^\star)
    ~\le~ (z/\epsilon)^{O(z)}\cdot\cost(\By).
\end{equation*}
This finishes the proof of \Cref{claim:correctness-strengthen-cost}.
\end{proof}

Finally, we note that \Cref{claim:correctness-strengthen-cost} guarantees a size bound of $|C|\le k+1+2\cdot |F|/n^{\epsilon}$, which can be easily strengthened to $|C|\le k+1+|F|/n^{\epsilon}$, as desired in \Cref{lem:rounding}, by scaling $n$ by an appropriate constant factor. This scaling does not affect the asymptotic approximation ratio, round complexity, or total memory bound.

\subsubsection{MPC Implementation}
\label{sec:rounding:MPC}

In this subsection, we discuss in turn how to implement the algorithms from the three steps in \Cref{sec:rounding:algorithm} in the MPC model, and show that all these implementations achieve the asymptotic round complexity and total memory bounds stated in \Cref{lem:rounding}, thereby completing its proof.

Throughout, we say that an MPC algorithm successfully implements an offline algorithm with probability $1-\frac{1}{\poly(n)}$ if there exists a success event of probability at least $1-\frac{1}{\poly(n)}$ such that, conditioned on this event, all guarantees established for the offline algorithm
remain valid for the MPC implementation; we make no guarantee on its output when the success event does not occur.

\paragraph{MPC Implementation of \Cref{alg:coarse-rounding}.}
We first discuss the MPC implementation of \Cref{alg:coarse-rounding} from the first step in \Cref{sec:rounding:algorithm}.

\begin{claim}[MPC Implementation of \Cref{alg:coarse-rounding}]
    \label{lem:MPC-coarse-rounding}
    There exists an MPC algorithm that successfully implements \Cref{alg:coarse-rounding} with probability at least $1-\frac{1}{\poly(n)}$; specifically, conditioned on successful implementation,  the guarantees in \Cref{lem:correctness-coarse-rounding} hold.
    
    With local memory $s\ge \poly(d\log(n\Delta))$ per machine, the algorithm runs in $O(\log_s n)$ rounds and uses $2^{O(\epsilon d)}\cdot n^{1+O(\epsilon)}\cdot\poly(d\log(n\Delta))$ total memory.
\end{claim}
\begin{proof}
To implement Line~\ref{alg:coarse-rounding:yprime}, we define the multiset $Q$ as follows: for every $q\in F$, $Q$ contains exactly $\lceil n^{\epsilon}y_q\rceil$ copies of $q$; thus, $n^{\epsilon}\le |Q|\le n^{\epsilon}k+|F|=O(n^{1+\epsilon})$.
Such a multiset $Q$ can be constructed in the MPC model by replicating each $q\in F$ into $\lceil n^{\epsilon}y_q\rceil$ copies.
The fractional solution $\By'=\frac{1}{n^{\epsilon}}\cdot \mathbf{1}^{Q}$ can then be constructed trivially, and its total opening satisfies
\begin{equation*}
\|\By'\|_1
~=~
\frac{1}{n^{\epsilon}}\cdot \sum_{q\in F} \lceil n^{\epsilon}\cdot y_q\rceil
~\le~
\frac{1}{n^{\epsilon}}\cdot \sum_{q\in F} ( n^{\epsilon}\cdot y_q + 1)
~=~
\|\By\|_1+\frac{|F|}{n^{\epsilon}}.
\end{equation*}
It is also easy to verify that $\cost(\By')\le \cost(\By)$, since, for every $q\in F$, the total opening at location $q$ under $\By'$ (aggregated over all copies of $q$ in $Q$) is at least the opening $y_q$ of $q$ under $\By$.
The main bottleneck in implementing Line~\ref{alg:coarse-rounding:yprime} is data replication, which is a standard procedure that can be implemented in $O(\log_s n)$ rounds using $O(dn^{1+\epsilon})$ total memory (for storing the multiset $Q$ of size $O(n^{1+\epsilon})$); see, e.g.,~\cite[Lemma~1]{NanongkaiS22}.

The MPC implementation of Line~\ref{alg:coarse-rounding:xprime} is somewhat involved, and we summarize its guarantee in the following claim.

\begin{claim}[MPC Implementation of Line~\ref{alg:coarse-rounding:xprime} of \Cref{alg:coarse-rounding}]
    \label{claim:MPC-implementation-xprime}
    There exists an MPC algorithm that, given the fractional solution $\By'=\frac{1}{n^{\epsilon}}\cdot \mathbf{1}^{Q}$ computed in Line~\ref{alg:coarse-rounding:yprime} of \Cref{alg:coarse-rounding},  
    with probability $1-\frac{1}{\poly(n)}$ computes for every $p\in P$ a size-$n^{\epsilon}$ submultiset $Q(p)$ of $Q$ such that the assignment $\Bx'\in \{0,1/n^{\epsilon}\}^{P\times Q}$ defined below satisfies $\Bx'\sim \By'$ and $\asncost(\Bx')\le O(\epsilon^{-z})\cdot \cost(\By')$.
    \begin{equation}\label{eq:xprime}
        \forall p\in P, q\in Q, \quad x'_{p, q} = 
        \begin{cases}
            1/n^{\epsilon} & q\in Q(p)\\
            0 & q\notin Q(p)
        \end{cases}.
    \end{equation}
    With local memory $s\ge \poly(d\log(n\Delta))$ per machine, the algorithm runs in $O(\log_s n)$ rounds and uses $2^{O(\epsilon d)}\cdot n^{1+O(\epsilon)}\cdot\poly(d\log(n\Delta))$ total memory.
\end{claim}

We defer the proof of \Cref{claim:MPC-implementation-xprime} to later and now proceed to implement the remaining lines of \Cref{alg:coarse-rounding} in the MPC model, whose implementation is relatively straightforward.

Lines~\ref{alg:coarse-rounding:alpha} and~\ref{alg:coarse-rounding:beta} can be trivially implemented in one MPC round.
We finally describe how to implement Lines~\ref{alg:coarse-rounding:f} and~\ref{alg:coarse-rounding:return} using the multisets $\{Q(p)\}_{p\in P}$ computed in \Cref{claim:MPC-implementation-xprime}, as follows:
\begin{itemize}
    \item 
    Line~\ref{alg:coarse-rounding:f} is an aggregation task that computes, for every $p\in P$, the minimum $\alpha_q$ over all $q\in Q(p)$. Hence, we can use standard sorting and converge-cast to aggregate these values for all $p\in P$ simultaneously (see, e.g.,~\cite[Lemma~5.10]{Cohen-addadKP26} for a similar aggregation task), in $O(\log_s n)$ rounds using $O(n^{1+\epsilon}\cdot \poly(d\log(n\Delta)))$ total memory.

    \item 
    Line~\ref{alg:coarse-rounding:return} can be implemented by an aggregation task followed by a simple lookup. First, for every $q\in \bigcup_{p\in P}Q(p)$, we compute the minimum $\beta_p$ over all $p\in P$ with $q\in Q(p)$, which determines $g(q):=\argmin_{p\in P:x'_{p,q}>0}\beta_p$. Then, for every $p\in P$, we check whether $g(f(p))=p$ via sorting and, if so, add $f(p)$ to $C$. To ensure that $C$ is a set, we perform a final sorting step to remove duplicate points.
    Hence, Line~\ref{alg:coarse-rounding:return} can be implemented in $O(\log_s n)$ rounds using $O(n^{1+\epsilon}\cdot \poly(d\log(n\Delta)))$ total memory.
\end{itemize}
This completes our discussion of the MPC implementation of \Cref{alg:coarse-rounding}. Overall, the implementation runs in $O(\log_s n)$ rounds and uses $2^{O(\epsilon d)}\cdot n^{1+O(\epsilon)}\cdot\poly(d\log(n\Delta))$ total memory. 
The operations used above are either deterministic or satisfy the high-probability guarantees of their corresponding primitives.
By a union bound, the entire implementation successfully implements \Cref{alg:coarse-rounding} with probability at least $1-1/\poly(n)$.
This finishes the proof of \Cref{lem:MPC-coarse-rounding}.
\end{proof}

It remains only to prove \Cref{claim:MPC-implementation-xprime}, which we do below.

\begin{proof}[Proof of \Cref{claim:MPC-implementation-xprime}]
    For simplicity, we only consider the case where $Q$ is a set rather than a multiset; the same argument can be easily generalized to the multiset case, for example, by distinguishing colocated points using their IDs.

\begin{algorithm}[h]
\DontPrintSemicolon
\caption{A more detailed implementation of
Line~\ref{alg:coarse-rounding:xprime} of \Cref{alg:coarse-rounding}}
\label{alg:xprime}

$B\gets 2n^{2\epsilon}$ and  $L\gets O(\log n)$ (with a sufficiently large hidden constant)\;

\For{every $\ell\in[L]$}{ \label{alg:xprime:first-for}
    independently sample $v_q^{(\ell)}$ uniformly from $[B]$
    for every $q\in Q$
    \label{alg:xprime:v}\;

    for every $i\in[B]$, let
    $Q_i^{(\ell)}\gets
    \{q\in Q:v_q^{(\ell)}=i\}$
    \label{alg:xprime:partition}\;
}

\For{every $(\ell,i)\in[L]\times[B]$
such that $Q_i^{(\ell)}\neq\emptyset$}{\label{alg:xprime:second-for}
    compute, for every $p\in P$, an $O(\epsilon^{-1})$-ANN $\hat{q}_i^{(\ell)}(p)\in Q_i^{(\ell)}$ in $Q_i^{(\ell)}$
    \label{alg:xprime:ann}\;
}

\tcp*{the algorithm fails if fewer than $n^{\epsilon}$ points exist}

for every $p\in P$, let $Q(p)\gets$ the $n^{\epsilon}$ points in $\{\hat{q}_i^{(\ell)}(p): \ell\in[L], i\in[B], Q_i^{(\ell)} \neq \emptyset \}$ closest to $p$
\label{alg:xprime:Qp}\;

\Return $\{Q(p)\}_{p\in P}$\;
\end{algorithm}

We present in \Cref{alg:xprime} a more detailed implementation of
Line~\ref{alg:coarse-rounding:xprime} of \Cref{alg:coarse-rounding}.
    We start by describing the MPC implementation of \Cref{alg:xprime}:
\begin{itemize}
    \item We first implement the first \textbf{for} loop (Lines~\ref{alg:xprime:first-for} to~\ref{alg:xprime:partition}), which aims to compute $\{Q_i^{(\ell)}\}_{\ell\in[L],i\in[B]}$. 
    To this end, for every $(\ell,i)\in[L]\times[B]$, we allocate a set $\calM_i^{(\ell)}$ of machines responsible for storing $Q_i^{(\ell)}$.
    We replicate $Q$ into $L$ copies.
    For every point $q\in Q$ in the $\ell$-th copy, we independently sample $v_q^{(\ell)}\in[B]$ locally and route $q$ to $\calM_i^{(\ell)}$, where $i=v_q^{(\ell)}$.
    After this step, each $\calM_i^{(\ell)}$ stores exactly $Q_i^{(\ell)}$. This step runs in $O(\log_s n)$ rounds, where the main bottleneck is replicating $Q$ into $L=O(\log n)$ copies. Moreover, since each $\calM_i^{(\ell)}$ stores at most $|Q|=O(n^{1+\epsilon})$ points and there are $LB=O(n^{2\epsilon}\log n)$ such sets of machines, this implementation of the first \textbf{for} loop uses $n^{1+3\epsilon}\cdot\poly(d\log (n\Delta))$ total memory.

    \item To implement the second \textbf{for} loop (Lines~\ref{alg:xprime:second-for} to~\ref{alg:xprime:ann}), we first replicate the query set $P$ to every nonempty machine set $\calM_i^{(\ell)}$. We then run the MPC approximate nearest-neighbor algorithm from~\cite[Section~3.1]{CzumajGJK024} on each $\calM_i^{(\ell)}$ in parallel, taking $Q_i^{(\ell)}$ as the terminal set and $P$ as the query set. This runs in $O(\log_s n)$ rounds.
    Since there are at most $LB=O(n^{2\epsilon}\log n)$ machine sets, the total memory is that required by \cite[Section~3.1]{CzumajGJK024} multiplied by $O(n^{2\epsilon}\log n)$, giving a total memory of $2^{O(\epsilon d)}\cdot n^{1+O(\epsilon)}\cdot\poly(d\log(n\Delta))$.
    
    \item Finally, we construct $Q(p)$ for every $p\in P$. For every $p\in P$, $\ell\in[L]$, and $i\in[B]$ such that $Q_i^{(\ell)}\neq\emptyset$, we create a tuple
    $(p, \dist(p, \hat{q}_i^{(\ell)}(p)), \hat{q}_i^{(\ell)}(p))$;
    there are at most $nLB=O(n^{1+2\epsilon}\log n)$ such tuples.
    We first remove duplicate tuples and then sort the remaining tuples lexicographically. For every $p\in P$, we retain the first $n^{\epsilon}$ tuples associated with $p$ and let $Q(p)$ consist of the corresponding points $\hat{q}_i^{(\ell)}(p)$.
    These operations can be implemented in $O(\log_s n)$ rounds using $n^{1+2\epsilon}\cdot\poly(d\log(n\Delta))$ total memory.
\end{itemize}
Combining the above three steps, \Cref{alg:xprime} can be implemented in $O(\log_s n)$ rounds using $2^{O(\epsilon d)}\cdot n^{1+O(\epsilon)}\cdot\poly(d\log(n\Delta))$ total memory.

It remains to prove the correctness of \Cref{alg:xprime}.
To this end, for every $p\in P$, let $Q^{\star}(p)\subseteq Q$ be the $n^{\epsilon}$ closest points in $Q$ to $p$. 
Whenever all $Q(p)$ are well-defined, the definition of $\Bx'$ in \eqref{eq:xprime} and the definition of $\By'=\frac{1}{n^{\epsilon}}\cdot \mathbf{1}^{Q}$ imply that $\Bx'\sim \By'$ is feasible, and
\begin{equation}\label{eq:costxprime}
\asncost(\Bx')
=
\frac{1}{n^{\epsilon}}
\sum_{p\in P}\sum_{q\in Q(p)}\dist^z(p,q),
\qquad
\cost(\By')
=
\frac{1}{n^{\epsilon}}
\sum_{p\in P}\sum_{q\in Q^{\star}(p)}\dist^z(p,q).
\end{equation}
Hence, it suffices to show that, for every fixed point $p\in P$, with probability $1-\frac{1}{\poly(n)}$, Line~\ref{alg:xprime:Qp} finds $n^{\epsilon}$ points, so that $Q(p)$ is well-defined, and
\begin{equation}\label{eq:Q-p}
\sum_{q\in Q(p)}\dist^z(p,q)
~\le~ O(\epsilon^{-z})\cdot
\sum_{q\in Q^{\star}(p)}\dist^z(p,q).
\end{equation}
By a union bound over all $p\in P$, with probability $1-\frac{1}{\poly(n)}$, $Q(p)$ is well-defined and \eqref{eq:Q-p} holds simultaneously for all $p\in P$. Combining this with \eqref{eq:costxprime} yields $\asncost(\Bx')\le O(\epsilon^{-z})\cdot \cost(\By')$.
For the remainder of the proof, we fix an arbitrary point $p\in P$.
We first establish the following claim.

\begin{claim}\label{claim:Qpartition}
    With probability $1-\frac{1}{\poly(n)}$, there exists $\ell\in[L]$ such that all points in $Q^{\star}(p)$ are separated by the partition $\{Q_i^{(\ell)}\}_{i\in[B]}$.
\end{claim}

\begin{proof}
    Fix an $\ell\in[L]$. For any two distinct points
    $q,q'\in Q^{\star}(p)$, the values $v_q^{(\ell)}$ and
    $v_{q'}^{(\ell)}$ are sampled independently and uniformly
    from $[B]$.
    Hence, $\Pr[v_q^{(\ell)}=v_{q'}^{(\ell)}]=1/B$.
    Since $|Q^{\star}(p)|=n^{\epsilon}$ and $B=2n^{2\epsilon}$,
    a union bound over all pairs of distinct points gives
    \begin{equation*}
        \Pr\big[\exists q\neq q'\in Q^{\star}(p),
        v_q^{(\ell)}=v_{q'}^{(\ell)}\big]
        ~\le~ \binom{n^{\epsilon}}{2}\cdot\frac{1}{B}
        ~\le~ \frac14.
    \end{equation*}
    Thus, with probability at least $3/4$, all points in $Q^{\star}(p)$ lie in distinct parts of the partition $\{Q_i^{(\ell)}\}_{i\in[B]}$.
    Since the random partitions are independent across $\ell\in[L]$ and $L=O(\log n)$ is sufficiently large, with probability $1-\frac{1}{\poly(n)}$, there exists $\ell\in[L]$ such that all points in $Q^{\star}(p)$ lie in distinct parts of the partition $\{Q_i^{(\ell)}\}_{i\in[B]}$.
\end{proof}

Condition on the high-probability event in \Cref{claim:Qpartition} and 
fix any such $\ell\in[L]$. For every $q\in Q^{\star}(p)$, let $i_q\in[B]$ be the unique index such that $q\in Q_{i_q}^{(\ell)}$.
The corresponding sets $Q_{i_q}^{(\ell)}$ are pairwise disjoint, so the ANN points $\hat{q}_{i_q}^{(\ell)}(p)$ computed in Line~\ref{alg:xprime:ann} are distinct, and hence $Q(p)$ is well-defined.
For every $q\in Q^{\star}(p)$, by Line~\ref{alg:xprime:ann}, we have
\begin{equation*}
    \dist^z\big(p,\hat{q}_{i_q}^{(\ell)}(p)\big)
    ~\le~
    O(\epsilon^{-z})\cdot
    \dist^z\big(p,Q_{i_q}^{(\ell)}\big)
    ~\le~ O(\epsilon^{-z})\cdot \dist^z(p,q),
\end{equation*}
where the second inequality follows from $q\in Q_{i_q}^{(\ell)}$. Summing over all $q\in Q^{\star}(p)$ gives
\begin{equation*}
    \sum_{q\in Q^{\star}(p)}
    \dist^z\big(p,\hat{q}_{i_q}^{(\ell)}(p)\big)
    ~\le~
    O(\epsilon^{-z})\cdot
    \sum_{q\in Q^{\star}(p)}\dist^z(p,q).
\end{equation*}
Finally, by the definition of $Q(p)$ in Line~\ref{alg:xprime:Qp}, it follows that
\begin{equation*}
    \sum_{q\in Q(p)}\dist^z(p,q)
    ~\le~
    \sum_{q\in Q^{\star}(p)}
    \dist^z\big(p,\hat{q}_{i_q}^{(\ell)}(p)\big).
\end{equation*}
Combining the above two inequalities proves \eqref{eq:Q-p}, thereby completing the correctness proof of \Cref{alg:xprime}.
This finishes the proof of \Cref{claim:MPC-implementation-xprime}.
\end{proof}

\paragraph{MPC Implementation of \Cref{alg:strengthen-size}.}
We then discuss the MPC implementation of \Cref{alg:strengthen-size} from the second step in \Cref{sec:rounding:algorithm}. 
For this, we need the following claim concerning the MPC implementation of the bi-point rounding procedure used in Line~\ref{alg:strengthen-size:bi-point-rounding}. The claim follows almost directly from~\cite{Cohen-addadKP26}, with the only modification being a generalization of their $z=2$ case to arbitrary $z\ge 1$;
for completeness, we provide a proof in Appendix~\ref{sec:proof-bi-point-rounding}.

\begin{claim}[Bi-point Rounding~\cite{Cohen-addadKP26}]
    \label{lem:bi-point-rounding}
    There exists an MPC algorithm that, given a bi-point solution $(C_1,C_2,a)$, computes a random solution $C\subseteq C_1 \cup C_2 \subseteq F$ such that
    \begin{equation*}
        |C|\le  a|C_1|+(1-a)|C_2| + 1 \quad \text{and} \quad \E[\cost(C)]\le \epsilon^{-O(z)}\cdot\bigl(a\cost(C_1)+(1-a)\cost(C_2)\bigr).
    \end{equation*}
    If the local memory per machine is $s\ge \poly(d\log(n\Delta))$,
    the algorithm runs in $O(\log_s n)$ rounds and uses total memory $2^{O(\epsilon d)}\cdot n\cdot \poly(d\log(n\Delta))$.
\end{claim}
\begin{proof}
    The proof can be found in Appendix~\ref{sec:proof-bi-point-rounding}.
\end{proof}

\begin{claim}[MPC Implementation of \Cref{alg:strengthen-size}]
    \label{lem:MPC-Implementation-strengthen-size}
    There exists an MPC algorithm that successfully implements \Cref{alg:strengthen-size} with probability at least $1-\frac{1}{\poly(n)}$; specifically, conditioned on successful implementation, the guarantee in \Cref{lem:correctness-strengthen-size} holds.

    With local memory $s\ge \poly(d\log(n\Delta))$ per machine, the algorithm runs in $O(\log_s n)$ rounds and uses $2^{O(\epsilon d)}\cdot n^{1+O(\epsilon)}\cdot \poly(d\log(n\Delta))$ total memory.
\end{claim}
\begin{proof}
    Line~\ref{alg:strengthen-size:run} can be implemented straightforwardly by running $h=n^{2\epsilon}$ independent instances of the MPC implementation of \Cref{alg:coarse-rounding} in \Cref{lem:MPC-coarse-rounding} in parallel. By \Cref{lem:MPC-coarse-rounding}, this takes $O(\log_s n)$ rounds and uses $2^{O(\epsilon d)}\cdot n^{1+O(\epsilon)}\cdot\poly(d\log(n\Delta))$ total memory.

    To implement Line~\ref{alg:strengthen-size:cost}, we note that, for a fixed $i\in[h]$, it suffices to compute, for every point $p\in P$, an approximate-nearest center $c_i(p)\in C_i$ satisfying $\dist(p,c_i(p))\le O(\epsilon^{-1})\cdot \dist(p,C_i)$, and then aggregate the power-$z$ distances $\dist^z(p,c_i(p))$. 
The computation of $c_i(p)$ for all $p\in P$ is the main bottleneck. This can be achieved using the MPC approximate nearest-neighbor algorithm from~\cite[Section~3.1]{CzumajGJK024}, which runs in $O(\log_s n)$ rounds and uses $2^{O(\epsilon d)}\cdot n\cdot \poly(d\log(n\Delta))$ total memory.
    Since we can run this computation for all $i\in [h]$ in parallel, Line~\ref{alg:strengthen-size:cost} can be implemented in $O(\log_s n)$ rounds using $2^{O(\epsilon d)}\cdot n^{1+O(\epsilon)}\cdot \poly(d\log(n\Delta))$ total memory.

    To implement Lines~\ref{alg:strengthen-size:bi-point-solution} to~\ref{alg:strengthen-size:if1}, 
    we enumerate all pairs $(i,j)\in[h]^2$. For each pair, we compute an optimal value $a\in[0,1]$ minimizing $a\Hat{\cost}(C_i) + (1-a)\Hat{\cost}(C_j)$ subject to $a|C_i| + (1-a)|C_j| \le k + 2 \cdot |F|/n^{\epsilon}$, whenever such a feasible value exists. 
    We then select, among all resulting bi-point solutions, the one with the minimum estimated cost $a\Hat{\cost}(C_i) + (1-a)\Hat{\cost}(C_j)$; if no such solution exists, we simply pick an arbitrary bi-point solution. 
    In the MPC model, this can be implemented by distributing the records $(|C_i|,\Hat{\cost}(C_i),|C_j|,\Hat{\cost}(C_j))$ so that the value $a$ for each pair $(i,j)$ can be computed locally, and then selecting the resulting bi-point solution $(C_i,C_j,a)$ with the minimum value of $a\Hat{\cost}(C_i) + (1-a)\Hat{\cost}(C_j)$.
    This requires only standard MPC routing and aggregation, and takes $O(\log_s n)$ rounds using $O(h^2)$ additional total memory.

    Finally, we run the bi-point rounding algorithm from \Cref{lem:bi-point-rounding} on the selected $(C_i,C_j,a)$.

    The above implementation runs in $O(\log_s n)$ rounds and uses $2^{O(\epsilon d)} n^{1+O(\epsilon)}\poly(d\log(n\Delta))$ total memory.
    Moreover, since all operations, including the MPC implementation of \Cref{alg:coarse-rounding}, sorting, and aggregation, are either deterministic or succeed with probability at least $1-1/\poly(n)$, a union bound implies that the entire implementation successfully implements \Cref{alg:strengthen-size} with probability at least $1-1/\poly(n)$.
\end{proof}

\paragraph{MPC Implementation of \Cref{alg:strengthen-cost}.}

Finally, we discuss the MPC implementation of \Cref{alg:strengthen-cost} from the third step in \Cref{sec:rounding:algorithm}. Since this is the final algorithm used to establish \Cref{lem:rounding}, its MPC implementation, together with the correctness guarantee in \Cref{claim:correctness-strengthen-cost}, completes the proof of \Cref{lem:rounding}.

\begin{claim}[MPC Implementation of \Cref{alg:strengthen-cost}]
    \label{claim:implementation-strengthen-cost}
    There exists an MPC algorithm that successfully implements \Cref{alg:strengthen-cost} with probability at least $1-\frac{1}{\poly(n)}$; specifically, conditioned on successful implementation, the guarantee in \Cref{claim:correctness-strengthen-cost} holds.

    With local memory $s\ge \poly(d\log(n\Delta))$ per machine, the algorithm runs in $O(\log_s n)$ rounds and uses $2^{O(\epsilon d)}\cdot n^{1+O(\epsilon)}\cdot \poly(d\log(n\Delta))$ total memory.
\end{claim}
\begin{proof}
The computation of $\Hat{\cost}(C)$ for all solutions $C$ returned by the instances of \Cref{alg:strengthen-size} can be implemented in the same way as Line~\ref{alg:strengthen-size:cost} of \Cref{alg:strengthen-size}. 
    Then, we compute every $|C|$ by aggregation, discard all solutions with size larger than $k + 1 + 2\cdot |F|/n^{\epsilon}$, and select the remaining solution with minimum $\Hat{\cost}(C)$ by aggregating over the $O(\log n)$ runs; if no solution remains, we return an arbitrary subset of $F$.
    Therefore, the main bottleneck is running $O(\log n)$ independent instances of \Cref{alg:strengthen-size} in parallel, which, by \Cref{lem:MPC-Implementation-strengthen-size}, takes $O(\log_s n)$ rounds and uses $2^{O(\epsilon d)}\cdot n^{1+O(\epsilon)}\cdot\poly(d\log(n\Delta))$ total memory.
    The subroutine calls above are either deterministic or succeed with probability at least $1-\frac{1}{\poly(n)}$. By a union bound, the entire implementation successfully implements \Cref{alg:strengthen-cost} with probability at least $1-1/\poly(n)$. This finishes the proof.
\end{proof}
     \section{$O(1)$-Approximation with Exactly $k$ Centers}
\label{sec:clustering}

In this section, we design an MPC algorithm for \textsc{Euclidean} \kzC that achieves a true $O(1)$-approximation using \emph{exactly} $k$ centers; this is the main result of this paper, and is presented in \Cref{thm:clustering} below.

\begin{theorem}[True Approximation]
    \label{thm:clustering}
    For any constant $\epsilon\in(0,1)$ and integer $z\ge 1$, there is a randomized MPC algorithm for \textsc{Euclidean} \kzC that,
    \begin{itemize}
        \item 
        takes as input a dataset $P\subseteq \R^d$, a candidate center set $F\subseteq \R^d$, with $n:=|P|+|F|$ denoting the total number of points, an integer $1\le k\le |F|$, and an upper bound $\Delta\ge \Delta(P\cup F)$ on the aspect ratio;
        \item computes, with high probability, a center set $C\subseteq F$ satisfying 
        \begin{equation*}
            |C|~=~ k,\quad \cost(P,C) ~\le~ 2^{O(z^2/\epsilon)}\epsilon^{-O(z/\epsilon)} \cdot \OPTclF(P, F, k).
        \end{equation*}
\end{itemize} 
    With local memory $s\ge 2^{\Omega(1/\epsilon)}\cdot \poly(d\log(n\Delta))$ per machine, 
    the algorithm runs in $O(\epsilon^{-1}\log_s n)$ rounds and uses $2^{O(\epsilon d)}\cdot n^{1+O(\epsilon)}\cdot \poly(\epsilon^{-1}d\log(n\Delta))$ total memory.
\end{theorem}

Before we prove \Cref{thm:clustering}, we discuss how it implies \Cref{thm:clustering-solution}: 
For an input $n$-point dataset $P\subseteq \R^d$ with $d=O(\log n)$, we set $F=P$ and take $z=2$ for \kMeans and $z=1$ for \kMedian.
Here, we note that \kMedian and \kMeans allow centers to be chosen from the entire $\R^d$, but restricting the candidate center set to $F=P$ incurs only an additional $2^z$ factor in the approximation ratio, which we can afford.
As for the aspect ratio, we provide a preprocessing algorithm that reduces it to $\poly(n)$ (without knowing the original aspect ratio) while incurring only an additional $2^{O(z)}$ factor in the approximation ratio; see Appendix~\ref{sec:aspect-ratio}.
With these settings, for both \kMeans and \kMedian, the approximation ratio is $\epsilon^{-O(1/\epsilon)}$, and the total memory is $\tilde{O}(n^{1+O(\epsilon)})$, which becomes $\tilde{O}(n^{1+\epsilon})$ after appropriately rescaling $\epsilon$ in \Cref{thm:clustering}.  
Moreover, since $d=O(\log n)$ and the reduced aspect ratio is $\poly(n)$, the local memory $s=n^{\epsilon}$ satisfies $s\ge 2^{\Omega(1/\epsilon)}\cdot \poly(d\log(n\Delta))$, as required by \Cref{thm:clustering}.
The round complexity is thus $O(\epsilon^{-1}\log_s n)=O(\epsilon^{-2})$. \Cref{thm:clustering-solution} then follows.

The remainder of this section is devoted to proving \Cref{thm:clustering}. We first describe a kernelization subroutine from~\cite{abs-2507-11256} designed for a special \emph{restricted clustering} problem, and then present our algorithm, which follows a novel \emph{iterative kernelization} scheme.

\paragraph{One Kernelization Step.}
We use the technique from~\cite{abs-2507-11256},
which intuitively yields a kernelization for the
\ProblemName{$(r,z)$-Restricted Clustering} problem.
In this problem, the input consists of a dataset $P\subseteq \R^d$, a candidate center set $F\subseteq \R^d$, and an integer $1\le r\le |F| - 1$ specifying the number of centers to remove from $F$ (rather than the number $k$ of centers to select as in \kzC). The goal is to find a set $R\subseteq F$ with $|R|=r$ that minimizes $\cost(P,F\setminus R)$. 
Observe that this problem is equivalent to \kzC when $k+r=|F|$. Therefore, the optimal objective value for \ProblemName{$(r,z)$-Restricted Clustering} is exactly $\OPTclF(P,F,|F|-r)$.

For \ProblemName{$(r,z)$-Restricted Clustering},
\cite{abs-2507-11256} identifies a size-$O(r)$ subset of $F$
and shows that it suffices to find the set $R$ within this subset,
thereby yielding a kernelization step
that reduces the size of the candidate center set.
This also yields a kernelization step for \kzC
that reduces the size of the candidate center set to $O(|F|-k)$
and forms a core step of our algorithm for \Cref{thm:clustering},
which we discuss later.

We state the kernelization subroutine of~\cite{abs-2507-11256}
in \Cref{alg:new-instance} using our notation,
where $(\Pold, \Fold)$ and $(\Pnew, \Fnew)$
denote the pairs of dataset and candidate center set
before and after the kernelization step, respectively,
and summarize its guarantee in \Cref{claim:guarantee-new-instance}.
The proof of \Cref{claim:guarantee-new-instance}
follows almost directly from the analysis in~\cite[Section~11]{abs-2507-11256},
but generalizes their $z=2$ case to arbitrary $z\ge 1$;
for completeness, we provide the proof in Appendix~\ref{sec:guarantee-new-instance}.
Moreover, we show that their algorithm can be efficiently implemented
in the MPC model, as stated in \Cref{claim:MPC-new-instance}.

\begin{algorithm}[h]
\DontPrintSemicolon
\caption{Kernelization for \ProblemName{$(r,z)$-Restricted Clustering} on $(\Pold, \Fold)$}
\label{alg:new-instance}

    compute, for every $c\in \Fold$, a value $\tau(c)$ such that
    $
        \dist(c,\Fold\setminus\{c\})
        \le \tau(c)
        \le O(\epsilon^{-1})\cdot\dist(c, \Fold\setminus\{c\})
    $\;
\label{alg:new-instance:tau}

    compute, for every $p\in \Pold$, an $O(\epsilon^{-1})$-ANN
    $\pi(p)\in \Fold$\;
    \label{alg:new-instance:pi}

    for every $c\in \Fold$,
    $w(c)\gets |\{p\in \Pold:\pi(p)=c\}|$\;
    \label{alg:new-instance:w}

    sort the centers in $\Fold$ as $c_1,\dots,c_{|\Fold|}$ such that
    $
        w(c_1)\cdot\tau^z(c_1)
        \le w(c_2)\cdot\tau^z(c_2)
        \le \cdots
        \le w(c_{|\Fold|})\cdot\tau^z(c_{|\Fold|})
    $\;
\label{alg:new-instance:sort}

    $T_1\gets$ the set $\{c_i\}_{i\in[\min\{6r, |\Fold|\}]}$\;
    \label{alg:new-instance:T1}

    $T_2\gets$ the set $\{s_i\}_{i\in [\min\{6r, |\Fold|\}]}$, where each $s_i$ is an $O(\epsilon^{-1})$-ANN of $c_i$ in $\Fold\setminus T_1$
    
\tcp*{$T_2\gets \emptyset$ if $T_1 = \Fold$}
    \label{alg:new-instance:T2}

    \Return $(\Pnew, \Fnew)$, where $\Fnew\gets T_1\cup T_2$ and  $\Pnew\gets \{p\in \Pold: \pi(p)\in \Fnew\}$\;
    \label{alg:new-instance:return}
\end{algorithm}

\begin{claim}[Guarantees of \Cref{alg:new-instance}]
    \label{claim:guarantee-new-instance}
    The output $(\Pnew, \Fnew)$ of \Cref{alg:new-instance} satisfies the following properties.
    \begin{itemize}
        \item $\Fnew\subseteq \Fold$ and $r+1\le |\Fnew|\le 12r$.
        \item for every $R\subseteq \Fnew\subseteq \Fold$ with $|R|=r$ that is an $\alpha$-approximate solution to \ProblemName{$(r,z)$-Restricted Clustering} on $(\Pnew, \Fnew)$ for some $\alpha\ge 1$, i.e.,
        \begin{equation*}
        \cost(\Pnew,\Fnew\setminus R)
        \le
        \alpha\cdot \OPTclF(\Pnew,\Fnew,|\Fnew|-r),
        \end{equation*}
        $R$ is an $(\epsilon^{-O(z)}\cdot \alpha)$-approximate solution to \ProblemName{$(r,z)$-Restricted Clustering} on $(\Pold, \Fold)$, i.e.,
        \begin{equation*}
        \cost(\Pold,\Fold\setminus R)
        \le
        \epsilon^{-O(z)}\cdot \alpha\cdot
        \OPTclF(\Pold,\Fold,|\Fold|-r).
        \end{equation*}
    \end{itemize}
\end{claim}
\begin{proof}
    The proof can be found in Appendix~\ref{sec:guarantee-new-instance}.
\end{proof}

\begin{claim}[MPC Implementation of \Cref{alg:new-instance}]
    \label{claim:MPC-new-instance}
   There is an MPC implementation of \Cref{alg:new-instance} that,  
   with local memory $s\ge \poly(d\log(n\Delta))$ per machine, runs in $O(\log_s n)$ rounds and uses $2^{O(\epsilon d)}\cdot n\cdot\poly(d\log(n\Delta))$ total memory.
   Here, $n := |\Pold|+|\Fold|$ denotes the total number of points and $\Delta\ge\Delta(\Pold\cup\Fold)$ is an upper bound on the aspect ratio.
\end{claim}
\begin{proof}
    We describe how to implement each step of \Cref{alg:new-instance} in the MPC model:
    Line~\ref{alg:new-instance:tau} can be implemented using the MPC approximate nearest-neighbor algorithm from~\cite[Section~3.1]{CzumajGJK024}.
    Here, since we need an $O(\epsilon^{-1})$-ANN of $c$ in $\Fold\setminus\{c\}$, excluding $c$ itself, we slightly modify the composable function $f(Y)$ used in~\cite[Section~3.1]{CzumajGJK024} to return the $\min\{2,|Y|\}$ points in $Y$ with the smallest IDs, or $\emptyset$ if $Y=\emptyset$; then the same algorithm and argument in~\cite[Section~3.1]{CzumajGJK024} can be applied to implement Line~\ref{alg:new-instance:tau}.
    Similarly, Line~\ref{alg:new-instance:pi} can be implemented by applying the same MPC approximate nearest-neighbor algorithm with $\Fold$ as the terminal set and $\Pold$ as the query set. 
    Given $\{\pi(p)\}_{p\in \Pold}$, Line~\ref{alg:new-instance:w} can be implemented by sorting the pairs $\{(\pi(p),p)\}_{p\in \Pold}$ according to their first coordinates and, in parallel, aggregating the number of points assigned to each center using a standard MPC converge-cast operation (see, e.g.,~\cite[Lemma 5.10]{Cohen-addadKP26}).
    Line~\ref{alg:new-instance:sort} can then be implemented via standard MPC sorting. The set $T_1$ in Line~\ref{alg:new-instance:T1} is obtained by retaining the first $\min\{6r,|\Fold|\}$ centers in this ordering.
    To implement Line~\ref{alg:new-instance:T2}, if $T_1\neq \Fold$, we again invoke the MPC approximate nearest-neighbor algorithm from~\cite[Section~3.1]{CzumajGJK024}, this time taking $\Fold\setminus T_1$ as the terminal set and $T_1$ as the query set. This computes all centers $s_i$ simultaneously in $O(\log_s n)$ rounds.
    Finally, $\Fnew=T_1\cup T_2$ can be constructed via sorting and duplicate removal. Given $\Fnew$, the new dataset $\Pnew=\{p\in \Pold:\pi(p)\in \Fnew\}$ can be computed straightforwardly.

    The bottleneck consists of a constant number of calls to the MPC approximate nearest-neighbor algorithm from~\cite[Section~3.1]{CzumajGJK024}; all other steps use standard MPC sorting, routing, and converge-cast primitives. Therefore, the entire implementation runs in $O(\log_s n)$ rounds and uses
    $2^{O(\epsilon d)}\cdot n\cdot \poly(d\log(n\Delta))$
    total memory.
\end{proof}

\paragraph{Algorithm for \Cref{thm:clustering}: Iterative Kernelization.}
We now present our iterative kernelization for \Cref{thm:clustering} in \Cref{alg:recursive-reduction}.
Specifically, starting from the tuple $(P_{0},F_{0},k_{0})=(P,F,k)$ given as input to \Cref{thm:clustering}, we iteratively construct a sequence of tuples ${(P_{\ell},F_{\ell},k_{\ell})}_{\ell\in[0:L]}$, where the size of the candidate center set decreases rapidly with $\ell$.
In each iteration $\ell\in[0:L-1]$, we construct $(P_{\ell+1},F_{\ell+1},k_{\ell+1})$ from the current tuple $(P_{\ell},F_{\ell},k_{\ell})$ via the following two steps:
\begin{enumerate}
    \item We first compute a preliminary bicriteria solution $\Cpre_{\ell}\subseteq F_{\ell}$ of size at most $k_{\ell}+1+ |F_{\ell}|/n^{\epsilon}$; in the MPC model, this is provided by \Cref{thm:bicriteria}.
    \item Then, we apply \Cref{alg:new-instance} with $r=|\Cpre_{\ell}| - k_{\ell}$ and $(\Pold,\Fold)=(P_{\ell},\Cpre_{\ell})$ to obtain $(\Pnew,\Fnew)$, and set $(P_{\ell+1},F_{\ell+1},k_{\ell+1})=(\Pnew,\Fnew,|\Fnew|-r)$.
    By \Cref{claim:guarantee-new-instance}, the new candidate center set satisfies
    \begin{equation*}
    |F_{\ell+1}|~=~|\Fnew|
    ~\le~ 12(|\Cpre_{\ell}|-k_{\ell})
    ~\le~ 12(1+|F_{\ell}|/n^{\epsilon}).
    \end{equation*}
\end{enumerate}
We note that we do not directly apply \Cref{alg:new-instance} to reduce $(P_{\ell},F_{\ell},k_{\ell})$ (i.e., with $r=|F_{\ell}|-k_{\ell}$ and $(\Pold,\Fold)=(P_{\ell},F_{\ell})$), since this would only guarantee a candidate center set of size at most $12(|F_{\ell}|-k_{\ell})$, which may not yield any meaningful reduction when $k_{\ell}$ is small.

Finally, we compute an optimal solution $C_L$ to \kzC on the last tuple $(P_L,F_L,k_L)$, and then iteratively lift the solution back through the sequence of tuples in the way described in~\Cref{claim:guarantee-new-instance}, eventually obtaining a solution $C_0$ to the original input $(P_0,F_0,k_0)=(P,F,k)$.

\begin{algorithm}[h]
\DontPrintSemicolon
\caption{Clustering via Iterative Kernelization on Instance $(P, F, k)$}
\label{alg:recursive-reduction}

\textbf{if} $|F| = k$ \textbf{then} \Return $F$\;
\label{alg:recursive-reduction:trivial}

let $(P_0, F_0, k_0) \gets (P, F, k)$, and $L\gets \lceil4/\epsilon\rceil$\;
\label{alg:recursive-reduction:zero}

\For{$\ell\gets 0$ \KwTo $L - 1$}{\label{alg:recursive-reduction:first-for}
    compute a preliminary bicriteria solution $\Cpre_{\ell}\subseteq F_{\ell}$ such that (i)~$k_{\ell} + 1\le |\Cpre_{\ell}|\le k_{\ell} + 1 + |F_{\ell}|/n^{\epsilon}$ and (ii)~$\cost(P_{\ell}, \Cpre_{\ell}) \le 2^{O(z^2)}\cdot \epsilon^{-O(z)}\cdot\OPTclF(P_{\ell}, F_{\ell}, k_{\ell})$\;
    \label{alg:recursive-reduction:C}

    $r_{\ell}\gets |\Cpre_{\ell}| - k_{\ell}$
    \tcp*{$1\le r_{\ell}\le |\Cpre_{\ell}|-1$}
    \label{alg:recursive-reduction:rell}

    run \Cref{alg:new-instance} with $r = r_{\ell}$ and $(\Pold, \Fold) = (P_{\ell}, \Cpre_{\ell})$ to obtain $(\Pnew, \Fnew)$\;
    \label{alg:recursive-reduction:new-instance}

    $(P_{\ell + 1}, F_{\ell + 1}, k_{\ell + 1}) \gets (\Pnew, \Fnew, |\Fnew| - r_{\ell} )$\;
    \label{alg:recursive-reduction:ellplus1}
    
    \tcp*{\Cref{alg:new-instance} ensures that $F_{\ell + 1} \subseteq \Cpre_{\ell} \subseteq F_{\ell}$}
}

compute a solution $C_{L}\subseteq F_{L}$ of size $k_{L}$ such that $\cost(P_{L}, C_{L}) = \OPTclF(P_L, F_L, k_L)$\;
\label{alg:recursive-reduction:CL}

\For{$\ell \gets L-1$ \KwDownto  $0$}{\label{alg:recursive-reduction:second-for}
    $C_{\ell}\gets \Cpre_{\ell}\setminus (F_{\ell + 1} \setminus C_{\ell + 1}) \subseteq F_{\ell}$\;  
\label{alg:recursive-reduction:Cell}
}

\Return $C\gets C_0$\;

\end{algorithm}

We are now ready to prove \Cref{thm:clustering} by analyzing the algorithm in \Cref{alg:recursive-reduction}.

\begin{proof}[Proof of \Cref{thm:clustering}]
We analyze the correctness of \Cref{alg:recursive-reduction} and its MPC implementation separately below.

\paragraph{Correctness of \Cref{alg:recursive-reduction}.}
The solution returned in Line~\ref{alg:recursive-reduction:trivial} is optimal, so assume this case does not apply.
We prove the following stronger guarantee by backward induction on $\ell$:
\begin{equation}\label{eq:costCL}
    \forall \ell \in [0: L], \qquad |C_{\ell}| ~=~ k_{\ell}, \quad \cost(P_{\ell}, C_{\ell}) ~\le~ \big(2^{O(z^2)}\epsilon^{-O(z)}\big)^{L-\ell + 1}\cdot \OPTclF(P_{\ell}, F_{\ell}, k_{\ell}).
\end{equation}
The base case is $\ell=L$. By Line~\ref{alg:recursive-reduction:CL}, we directly have
\begin{equation*}
    |C_{\ell}| ~=~ k_{\ell}, \quad \cost(P_{\ell},C_{\ell})
    ~=~
    \OPTclF(P_{\ell},F_{\ell},k_{\ell})
    ~\le~
    \big(2^{O(z^2)}\epsilon^{-O(z)}\big)^{L-\ell+1}
    \cdot \OPTclF(P_{\ell},F_{\ell},k_{\ell}).
\end{equation*}
This proves \eqref{eq:costCL} for the base case $\ell=L$.

For the inductive step, suppose that \eqref{eq:costCL} holds for $\ell+1\in[L]$. We show below that \eqref{eq:costCL} also holds for $\ell$.
We first have the following fact.
\begin{fact}\label{fact:inclusion}
    $C_{\ell + 1} \subseteq F_{\ell + 1} \subseteq \Cpre_{\ell} \subseteq F_{\ell}$, and $r_{\ell} = |\Cpre_{\ell}| - k_{\ell} = |F_{\ell + 1}| - k_{\ell + 1}$.
\end{fact}
\begin{proof}
    The inclusion $C_{\ell+1}\subseteq F_{\ell+1}$ follows by considering two cases. If $\ell+1=L$, it follows directly from Line~\ref{alg:recursive-reduction:CL}. Otherwise, it follows from $C_{\ell+1}\subseteq \Cpre_{\ell+1}$ (Line~\ref{alg:recursive-reduction:Cell}) and $\Cpre_{\ell+1}\subseteq F_{\ell+1}$ (Line~\ref{alg:recursive-reduction:C}).
    The inclusion $F_{\ell+1}\subseteq \Cpre_{\ell}$ follows from the fact that $(P_{\ell+1},F_{\ell+1})=(\Pnew,\Fnew)$ is the output of \Cref{alg:new-instance} with $r=r_{\ell}$ and $(\Pold,\Fold)=(P_{\ell},\Cpre_{\ell})$ (Line~\ref{alg:recursive-reduction:new-instance}), together with the guarantee of \Cref{alg:new-instance} stated in \Cref{claim:guarantee-new-instance}.
    The inclusion $\Cpre_{\ell} \subseteq F_{\ell}$ follows directly from Line~\ref{alg:recursive-reduction:C}.
    Finally, the equation $r_{\ell} = |\Cpre_{\ell}| - k_{\ell} = |F_{\ell + 1}| - k_{\ell + 1}$ follows directly from Lines~\ref{alg:recursive-reduction:rell} and~\ref{alg:recursive-reduction:ellplus1}.
\end{proof}

The size guarantee $|C_{\ell}|=k_{\ell}$ follows directly from \Cref{fact:inclusion}, the definition $C_{\ell}=\Cpre_{\ell}\setminus(F_{\ell + 1}\setminus C_{\ell + 1})$ in Line~\ref{alg:recursive-reduction:Cell}, and \eqref{eq:costCL} for $\ell + 1$, as follows:
\begin{equation*}
    |C_{\ell}| ~=~ |\Cpre_{\ell}| - (|F_{\ell + 1}| - |C_{\ell + 1}|) ~=~ |\Cpre_{\ell}| - (|F_{\ell + 1}| - k_{\ell + 1}) ~=~ k_{\ell}.
\end{equation*}
We then prove the cost guarantee of $C_{\ell}$, i.e., 
$\cost(P_{\ell}, C_{\ell}) \le \big(2^{O(z^2)}\epsilon^{-O(z)}\big)^{L- \ell + 1}\cdot \OPTclF(P_{\ell}, F_{\ell}, k_{\ell})$.
Let $R := F_{\ell + 1} \setminus C_{\ell + 1}$.
We have the following facts:
\begin{itemize}
    \item 
    $|R| =|F_{\ell + 1}| - |C_{\ell + 1}| = |F_{\ell + 1}| - k_{\ell + 1} = |\Cpre_{\ell}| - k_{\ell} = r_{\ell}$
    by \Cref{fact:inclusion} and \eqref{eq:costCL} for $\ell + 1$. 
    \item $\cost(P_{\ell + 1}, F_{\ell + 1} \setminus R) \le \big(2^{O(z^2)}\epsilon^{-O(z)}\big)^{L-\ell}
    \cdot \OPTclF(P_{\ell + 1},F_{\ell + 1},k_{\ell + 1})$ by \eqref{eq:costCL} for $\ell+1$.

    \item $(P_{\ell+1},F_{\ell+1})$ is obtained by applying \Cref{alg:new-instance} with $r = r_{\ell}$ and $(\Pold, \Fold) = (P_{\ell},\Cpre_{\ell})$, as specified in Lines~\ref{alg:recursive-reduction:new-instance} and~\ref{alg:recursive-reduction:ellplus1}.
\end{itemize}
The first two facts together imply that $R$ is a $\big(2^{O(z^2)}\epsilon^{-O(z)}\big)^{L-\ell}$-approximate solution to \ProblemName{$(r_{\ell},z)$-Restricted Clustering} on $(P_{\ell+1},F_{\ell+1})$,
and the third fact allows us to apply \Cref{claim:guarantee-new-instance} and deduce that $R$ is a $\big(\epsilon^{-O(z)}\cdot\big(2^{O(z^2)}\epsilon^{-O(z)}\big)^{L-\ell}\big)$-approximate solution to \ProblemName{$(r_{\ell},z)$-Restricted Clustering} on $(P_{\ell},\Cpre_{\ell})$, i.e.,
\begin{align}
    \cost(P_{\ell}, C_{\ell}) ~&\le~ \epsilon^{-O(z)}\cdot \big(2^{O(z^2)}\epsilon^{-O(z)}\big)^{L-\ell}
    \cdot \OPTclF(P_{\ell}, \Cpre_{\ell}, |\Cpre_{\ell}| - r_{\ell}) 
    \nonumber
    \\
    &=~ \epsilon^{-O(z)}\cdot \big(2^{O(z^2)}\epsilon^{-O(z)}\big)^{L-\ell}
    \cdot \OPTclF(P_{\ell}, \Cpre_{\ell}, k_{\ell})\label{eq:costCL-1}
\end{align}
Finally, we upper-bound $\OPTclF(P_{\ell},\Cpre_{\ell},k_{\ell})$ in terms of $\OPTclF(P_{\ell},F_{\ell},k_{\ell})$ as follows:
Let $C^{\star}_{\ell}\subseteq F_{\ell}$ with size $k_{\ell}$ be such that
$\cost(P_{\ell},C^{\star}_{\ell})=\OPTclF(P_{\ell},F_{\ell},k_{\ell})$.
Let
$C'_{\ell}:=\{\mathrm{NN}(c,\Cpre_{\ell}):c\in C^{\star}_{\ell}\} \subseteq \Cpre_{\ell}$.
Then, $|C'_{\ell}|\le |C^{\star}_{\ell}|=k_{\ell}$ and therefore $\OPTclF(P_{\ell},\Cpre_{\ell},k_{\ell})\le \cost(P_{\ell},C'_{\ell})$. 
For every $p\in P_{\ell}$, let $c\in C^{\star}_{\ell}$ be a closest center to $p$. By the definition of $C'_{\ell}$,
\begin{align*}
    \dist(p,C'_{\ell}) ~&\le~
    \dist\big(p,\mathrm{NN}(c,\Cpre_{\ell})\big) ~\le~
    \dist(p,c) + \dist(c, \Cpre_{\ell}) ~\le~
    2\dist(p,c) + \dist(p,\Cpre_{\ell})\\
    &=~ 2\dist(p, C^{\star}_{\ell}) + \dist(p, \Cpre_{\ell})
\end{align*}
Hence, by the generalized triangle inequality (\Cref{lem:triangle-inequality}),
\begin{align*}
    \OPTclF(P_{\ell},\Cpre_{\ell},k_{\ell})
    ~&\le~
    \cost(P_{\ell},C'_{\ell}) ~\le~
    2^{O(z)}\cdot
    \left(
        \cost(P_{\ell},C^{\star}_{\ell})
        +
        \cost(P_{\ell},\Cpre_{\ell})
    \right)\\
    ~&\le~
    2^{O(z^2)}\cdot\epsilon^{-O(z)}\cdot
    \OPTclF(P_{\ell},F_{\ell},k_{\ell}),
\end{align*}
where the last inequality follows from
$\cost(P_{\ell},C^{\star}_{\ell})
=\OPTclF(P_{\ell},F_{\ell},k_{\ell})$
and
$\cost(P_{\ell},\Cpre_{\ell})
\le
2^{O(z^2)}\cdot\epsilon^{-O(z)}\cdot
\OPTclF(P_{\ell},F_{\ell},k_{\ell})$
by Line~\ref{alg:recursive-reduction:C}.
Plugging this bound into \eqref{eq:costCL-1}, we obtain
\begin{equation*}
    \cost(P_{\ell},C_{\ell})
    ~\le~
    \big(2^{O(z^2)}\epsilon^{-O(z)}\big)^{L-\ell+1}
    \cdot \OPTclF(P_{\ell},F_{\ell},k_{\ell}),
\end{equation*}
which proves \eqref{eq:costCL} for $\ell$.
This completes the proof of \eqref{eq:costCL} by the principle of mathematical induction.

Recall that $L=O(\epsilon^{-1})$ by Line~\ref{alg:recursive-reduction:zero}. Applying \eqref{eq:costCL} with $\ell=0$, the returned solution $C\gets C_0$ satisfies
\begin{equation*}
    \cost(P,C)
    ~\le~
    2^{O(z^2/\epsilon)}\cdot \epsilon^{-O(z/\epsilon)}
    \cdot \OPTclF(P,F,k).
\end{equation*}
This completes the analysis of the approximation ratio of \Cref{alg:recursive-reduction}.

\paragraph{MPC Implementation of \Cref{alg:recursive-reduction}.} 
Throughout the discussion of the MPC implementation of \Cref{alg:recursive-reduction}, we assume that $n^{\epsilon} \ge 100$. Otherwise, both the dataset $P$ and the candidate center set $F$ have size at most $100^{1/\epsilon}=2^{O(1/\epsilon)}$. Hence, the entire instance fits on a single machine with local memory $s\ge 2^{\Omega(1/\epsilon)}\cdot \poly(d\log(n\Delta))$ and can therefore be solved directly.

We first establish the following claim, which upper-bounds the size of the candidate center set $F_L$ in the final instance $(P_L,F_L,k_L)$.

\begin{claim}\label{claim:size-bound-FL}
    For $L=\lceil4/\epsilon\rceil$ and $n^{\epsilon}\ge 100$, it holds that $|F_L|\le O(1)$.\footnote{
        We note that here only the final candidate center set $F_L$ is guaranteed
        to have size $O(1)$, while the dataset $P_L$ may still be large.
        However, this case can already be handled efficiently in MPC.
        Moreover,
        if desired, we can also further reduce the size of $P_L$ to $O(1)$
        using a coreset~\cite{Har-PeledM04,Huang0024}.
    }
\end{claim}

\begin{proof}
    Consider a fixed $\ell\in[0:L-1]$. By Line~\ref{alg:recursive-reduction:C}, we have $|\Cpre_{\ell}| \le k_{\ell}+1+|F_{\ell}| / n^{\epsilon}$.
    Since the tuple $(P_{\ell+1},F_{\ell+1},k_{\ell+1})$ is constructed by running \Cref{alg:new-instance} with $r=r_{\ell}=|\Cpre_{\ell}|-k_{\ell}$ and $(\Pold,\Fold)=(P_{\ell},\Cpre_{\ell})$, it follows from \Cref{claim:guarantee-new-instance} that
    \begin{equation*}
        |F_{\ell+1}|
        ~\le~
        12(|\Cpre_{\ell}|-k_{\ell})
        ~\le~
        12\left(1+ |F_{\ell}| / n^{\epsilon}\right).
    \end{equation*}
    Define $\Phi(F):= |F|-\frac{12}{1-{12}/n^{\epsilon}}$.
    The above inequality implies that
    \begin{align*}
        \Phi(F_{\ell+1})
        ~=~
        |F_{\ell+1}|-\frac{12}{1-{12}/n^{\epsilon}}
        ~&\le~
        12+\frac{{12}}{n^{\epsilon}}|F_{\ell}|
        -\frac{12}{1-{12}/n^{\epsilon}}
        ~=~
        \frac{{12}}{n^{\epsilon}}
        \left(
            |F_{\ell}|-\frac{12}{1-{12}/n^{\epsilon}}
        \right)\\
        &=~
        \frac{{12}}{n^{\epsilon}}\cdot \Phi(F_{\ell}).
    \end{align*}
    Therefore,
    \begin{equation*}
        \Phi(F_L)
        ~\le~
        \left(\frac{{12}}{n^{\epsilon}}\right)^L
        \Phi(F_0)
        ~\le~
        \left(\frac{{12}}{n^{\epsilon}}\right)^L n ~\le~ 1,
    \end{equation*}
    where the second inequality follows from $\Phi(F_0)\le |F_0|\le n$, and the third follows from $12\le100^{3/4}\le n^{3\epsilon/4}$ and $L=\lceil4/\epsilon\rceil$, which give $(12/n^{\epsilon})^L n\le n^{-\epsilon L/4}n\le1$.
    Finally, $|F_{L}| = \Phi(F_{L}) + \frac{12}{1-{12}/n^{\epsilon}} \le O(1)$, which  completes the proof of \Cref{claim:size-bound-FL}.
\end{proof}

We now discuss the MPC implementation of \Cref{alg:recursive-reduction}.
Lines~\ref{alg:recursive-reduction:trivial} and~\ref{alg:recursive-reduction:zero} are straightforward to implement.
For the first \textbf{for} loop (Lines~\ref{alg:recursive-reduction:first-for} to~\ref{alg:recursive-reduction:ellplus1}),
Line~\ref{alg:recursive-reduction:C} is implemented by applying \Cref{thm:bicriteria}, using the original input size $n$ as the size upper bound.
If necessary, we add arbitrary centers from $F_{\ell}$ to ensure $|\Cpre_{\ell}|\ge k_{\ell}+1$ without increasing the cost.
Line~\ref{alg:recursive-reduction:new-instance} is implemented using \Cref{claim:MPC-new-instance}.
Hence, each iteration runs in $O(\log_s n)$ rounds and uses $2^{O(\epsilon d)}\cdot n^{1+O(\epsilon)}\cdot\poly(d\log(n\Delta))$ total memory.
Executing the $L=O(\epsilon^{-1})$ iterations sequentially therefore takes $O(\epsilon^{-1}\log_s n)$ rounds and increases the total memory by at most a factor of $O(\epsilon^{-1})$, yielding a total memory bound of $2^{O(\epsilon d)}\cdot n^{1+O(\epsilon)}\cdot\poly(\epsilon^{-1}d\log(n\Delta))$.
To implement Line~\ref{alg:recursive-reduction:CL}, we use the fact $|F_L|=O(1)$ from \Cref{claim:size-bound-FL} and enumerate all $O(1)$ subsets $C\subseteq F_L$ of size $k_L$.
For each such $C$, we broadcast $C$ to the machines storing $P_L$ and compute $\cost(P_L,C)$ by locally evaluating $\dist^z(p,C)$ and aggregating these values using a standard MPC converge-cast.
We then select a candidate of minimum cost.
All candidate costs can be computed in parallel, so this step takes $O(\log_s n)$ rounds and uses $n\cdot\poly(d\log(n\Delta))$ total memory.
Finally, for the second \textbf{for} loop (Lines~\ref{alg:recursive-reduction:second-for} to~\ref{alg:recursive-reduction:Cell}), each update $C_{\ell}\gets\Cpre_{\ell}\setminus(F_{\ell+1}\setminus C_{\ell+1})$ can be implemented using standard MPC sorting and set operations (see, e.g.,~\cite{AndoniSSWZ18}).
Each iteration takes $O(\log_s n)$ rounds and uses $n\cdot\poly(d\log(n\Delta))$ total memory.
Thus, executing the $L=O(\epsilon^{-1})$ iterations sequentially takes $O(\epsilon^{-1}\log_s n)$ rounds and uses $n\cdot\poly(\epsilon^{-1}d\log(n\Delta))$ total memory.

Overall, the above MPC implementation for \Cref{alg:recursive-reduction} runs in $O(\epsilon^{-1}\log_s n)$ rounds and uses $2^{O(\epsilon d)}\cdot n^{1+O(\epsilon)}\cdot\poly(\epsilon^{-1}d\log(n\Delta))$ total memory.
It remains to analyze the success probability, for which it suffices to consider the $L=O(\epsilon^{-1})$ invocations of the algorithm in \Cref{thm:bicriteria} in Line~\ref{alg:recursive-reduction:C}.
Each invocation succeeds with high probability. Therefore, by a union bound over the $L=O(\epsilon^{-1})$ invocations, which is a constant, all of them succeed with high probability as well.
This finishes the proof of \Cref{thm:clustering}.
\end{proof}

\begin{appendices}
\section{Preprocessing}
\label{sec:aspect-ratio}

In this section, we discuss the MPC preprocessing for dimension reduction and aspect ratio reduction separately; both are fully scalable and round-efficient.

\paragraph{Dimension Reduction.}
To reduce the dimension to $O(\log n)$, we apply an MPC implementation of the JL transform due to~\cite{AhanchiAHKZ23}, which is fully scalable and runs in $O(1)$ rounds.

\paragraph{Aspect Ratio Reduction.}
In the rest of this section, we present an MPC algorithm for reducing the aspect ratio of the input dataset $P$ for \kzC when the candidate center set is $F=P$, i.e., when centers are chosen from $P$.
For brevity, we write $\OPTclF = \OPTclF(P,P,k)$ and $\cost(C) = \cost(P,C)$.
Specifically, given a parameter $1\le \alpha \le \poly(n)$ and an $n$-point dataset $P \subseteq \R^d$, the algorithm computes a new dataset of the form $f(P) = \{ f(p) : p \in P \}$ for some injective mapping $f : P \to \R^d$, such that
\begin{enumerate}[label = (P\arabic*), leftmargin = *]
    \item\label{aspect-ratio-1} the aspect ratio of $f(P)$ is at most $\poly(n)$, and 
    \item\label{aspect-ratio-2} for any $\alpha$-approximate solution $f(C)\subseteq f(P)$ to \kzC on $f(P)$, the set $C \subseteq P$ is a $2^{O(z)}\alpha$-approximate solution to \kzC on $P$.
\end{enumerate}
The algorithm requires local memory $s\ge \poly(d\log n)$, runs in $O(\log_s n)$ rounds, and uses total memory $O(n\cdot \poly(d\log n))$.
Moreover, the algorithm stores the mapping $f$ as a set of pairs $\{(p, f(p)) : p \in P\}$, so that recovering $C$ from a given $f(C) \subseteq f(P)$ can be done in $O(\log_s n)$ rounds (see e.g.~\cite{AndoniSSWZ18}).

An offline version of our algorithm is presented in \Cref{alg:reduce-aspect-ratio}.
In the following, we first analyze the correctness of \Cref{alg:reduce-aspect-ratio}, and then discuss its MPC implementation.
Throughout, we consider the case $k<n$, for which $\OPTclF>0$; the case $k=n$ is trivial, since we can simply return $P$ as the solution.
Moreover, by the above dimension reduction, we assume that $d=O(\log n)$.

\begin{algorithm}[t]
    \caption{Reducing the aspect ratio}
    \DontPrintSemicolon
    \label{alg:reduce-aspect-ratio}
    Compute a value $\eta$ such that $\OPTclF \le \eta \le \beta\cdot \OPTclF$, where $\beta = n^{O(z)}$
    \label{alg:reduce-aspect-ratio:eta}\;
     
    Let $\alpha' = 2^{O(z)}\cdot \alpha$ be sufficiently large\;

    Let $p_1, \dots, p_n$ be an arbitrary ordering of the points in $P$\;    

    Let $T\gets \{p_{i,j} \in \R: p_i = (p_{i,1},\dots,p_{i,d})\in P, j\in [d] \}$, i.e. the set of all coordinate values
    \label{alg:reduce-aspect-ratio:T}

    \tcp*{remove duplicates}

    Sort $T$ in non-decreasing order, denoting the resulting sequence by $t_1, \dots, t_{m}$ ($m = |T|\le nd$)
    \label{alg:reduce-aspect-ratio:compute-sort}\;

    For every $i\in [m]$, compute $h(t_i) \gets t_1 + \sum_{j = 2}^{i} \min\{t_j - t_{j-1},\, (\alpha'\cdot \eta)^{1/z} \}$
    \label{alg:reduce-aspect-ratio:compute-h}\;

    For every $p_i = (p_{i,1},\dots, p_{i,d})\in P$, compute $g(p_i) \gets (h(p_{i,1}),\dots, h(p_{i,d}))$
    \label{alg:reduce-aspect-ratio:map}\;

    For every $i\in [n]$, 
    let $f(p_i) \gets$ a nearest point to $g(p_i)$ in $\mathrm{Net}_i := \{u\cdot \frac{\eta^{1/z}}{n \sqrt{d} \cdot \beta^{1/z}} + i\cdot \frac{\eta^{1/z}}{n^{100} \sqrt{d} \cdot \beta^{1/z}} : u\in \Z \}^d$
    \label{alg:reduce-aspect-ratio:round}\;

    \Return $f(P)$
\end{algorithm}

\paragraph{Proof of Property~\ref{aspect-ratio-1}.}
Fix $p_i,p_j\in P$ with $i\neq j$. We first show that $\dist(f(p_i), f(p_j)) \ge \frac{(\OPTclF)^{1/z}}{\poly(n)}$.
According to Line~\ref{alg:reduce-aspect-ratio:round}, it is easy to verify that $\min_{p\in \mathrm{Net}_i, q\in \mathrm{Net}_j} \dist(p,q) \ge \frac{\eta^{1/z}}{n^{100}\sqrt{d}\cdot \beta^{1/z}} \ge \frac{(\OPTclF)^{1/z}}{\poly(n)}$.
Therefore, $\dist(f(p_i), f(p_j)) \ge \frac{(\OPTclF)^{1/z}}{\poly(n)}$ because $f(p_i)\in \mathrm{Net}_i$ and $f(p_j)\in \mathrm{Net}_j$.

We then show that $\dist(f(p_i), f(p_j)) \le \poly(n) \cdot (\alpha \cdot \OPTclF)^{1/z}$.
We need the following fact.

\begin{fact}\label{fact:reduce-aspect-ratio-move}
    For every $p,q\in P$, $|\dist(g(p), g(q)) - \dist(f(p), f(q))| \le \frac{2(\OPTclF)^{1/z}}{n}$.
\end{fact}
\begin{proof}
    By Line~\ref{alg:reduce-aspect-ratio:round}, rounding moves each point by at most $\frac{\eta^{1/z}}{n \cdot \beta^{1/z}} \le \frac{(\OPTclF)^{1/z}}{n}$. The claimed bound follows from the triangle inequality.
\end{proof}

By \Cref{fact:reduce-aspect-ratio-move},
it suffices to show that $\dist(g(p_i), g(p_j)) \le \poly(n) \cdot (\alpha \cdot \OPTclF)^{1/z}$.
Notice that the construction of $h$ in Line~\ref{alg:reduce-aspect-ratio:compute-h} ensures that for every two points, their corresponding coordinates after applying the mapping $h$ differ by at most $m \cdot (\alpha' \cdot \eta)^{1/z}\le \poly(n) \cdot (\alpha \cdot \OPTclF)^{1/z}$.
Therefore, $\dist(g(p_i), g(p_j)) \le \poly(n) \cdot (\alpha \cdot \OPTclF)^{1/z}$ follows directly, which, as stated, implies that $\dist(f(p_i), f(p_j)) \le \poly(n) \cdot (\alpha \cdot \OPTclF)^{1/z}$.

Combining the above results, we confirm that the aspect ratio of $f(P)$ is at most $\poly(n)$.

\paragraph{Proof of Property~\ref{aspect-ratio-2}.}
For every solution $C \subseteq P$, define $\cost^g(C) := \sum_{p \in P} \dist^z(g(p), g(C))$ as the clustering cost after applying the map $g$, and define $\cost^f(C) := \sum_{p \in P} \dist^z(f(p), f(C))$ as the clustering cost after applying the map $f$.

We first have the following fact.

\begin{fact}\label{fact:reduce-aspect-ratio-h}
    For every $p_i, p_j\in P$, the following guarantees hold.
    \begin{itemize}
        \item If $\dist(p_i, p_j)\le  (\alpha'\cdot \OPTclF)^{1/z}$, then $\dist(g(p_i), g(p_j)) = \dist(p_i,p_j)$.
        \item If $\dist(g(p_i), g(p_j)) <  (\alpha'\cdot \OPTclF)^{1/z}$, then $\dist(g(p_i), g(p_j)) = \dist(p_i,p_j)$.
    \end{itemize}
\end{fact}
\begin{proof}
    We first prove the first guarantee.
    By Line~\ref{alg:reduce-aspect-ratio:map}, $g(p_i) = (h(p_{i,1}), \dots, h(p_{i,d}))$ and $g(p_j) = (h(p_{j,1}), \dots, h(p_{j,d}))$. For every $u \in [d]$, since $|p_{i,u} - p_{j,u}| \le \dist(p_i,p_j) \le (\alpha' \cdot \OPTclF)^{1/z}\le (\alpha' \cdot \eta)^{1/z}$, we have $h(p_{i,u}) - h(p_{j,u}) = p_{i,u} - p_{j,u}$ by Line~\ref{alg:reduce-aspect-ratio:compute-h}. Therefore, we have $\dist(g(p_i), g(p_j)) = \dist(p_i, p_j)$.

    We then prove the second guarantee by considering its contrapositive. 
    If $\dist(g(p_i), g(p_j)) \neq \dist(p_i,p_j)$, then there exists $u \in [d]$ such that $h(p_{i,u}) - h(p_{j,u}) \neq p_{i,u} - p_{j,u}$. 
    By Line~\ref{alg:reduce-aspect-ratio:compute-h}, this implies
    $|h(p_{i,u}) - h(p_{j,u})| \ge (\alpha' \cdot \eta)^{1/z} \ge (\alpha' \cdot \OPTclF)^{1/z}$,
    and therefore $
    \dist(g(p_i), g(p_j)) \ge (\alpha' \cdot \OPTclF)^{1/z}$.
\end{proof}

Let $C^{*} \subseteq P$ be an optimal solution to \kzC on $P$.
For every $p \in P$, let $c^* \in C^*$ denote the closest center in $C^*$ to $p$ (with ties broken arbitrarily); we have $\dist(p, c^{*}) \le (\OPTclF)^{1/z}$, which implies that $\dist(g(p), g(c^*)) = \dist(p, c^{*})$ by \Cref{fact:reduce-aspect-ratio-h}. Therefore, $\cost^g(C^*)\le \cost(C^*)$. 
Then, by \Cref{fact:reduce-aspect-ratio-move} and the generalized triangle inequality (\Cref{lem:triangle-inequality}), we have $\cost^f(C^*) \le 2^{z-1}\cdot \cost^g(C^*) + 2^{z-1}\cdot n\cdot (\frac{2(\OPTclF)^{1/z}}{n})^z \le 2^{O(z)}\cdot \cost(C^*)$.
This means that the optimal \kzC objective on the final dataset $f(P)$ is at most $2^{O(z)}$ times that of the input dataset $P$.

Then, consider an arbitrary solution $f(C) \subseteq f(P)$ with $|f(C)| = k$ (and thus $|C| = k$), and suppose that $f(C)$ is an $\alpha$-approximate solution to \kzC on $f(P)$. Then we have $\cost^f(C) \le 2^{O(z)} \cdot \alpha \cdot \OPTclF$. By \Cref{fact:reduce-aspect-ratio-move} and the generalized triangle inequality (\Cref{lem:triangle-inequality}), it follows that $\cost^g(C) \le 2^{O(z)} \alpha \cdot \OPTclF$. Finally, recall that $\alpha' = 2^{O(z)} \alpha$ is sufficiently large; for every $p\in P$, 
since $\dist^z(g(p), g(C)) \le \cost^g(C) < \alpha'\cdot \OPTclF$, by 
\Cref{fact:reduce-aspect-ratio-h} we have that $\dist^z(p, C) \le \dist^z(g(p), g(C))$. 
Therefore, $\cost(C) \le \cost^g(C) \le 2^{O(z)}\alpha\cdot \OPTclF$.
This confirms that $C$ is a $2^{O(z)} \alpha$-approximate solution to \kzC on the input dataset $P$.

\paragraph{MPC Implementation of \Cref{alg:reduce-aspect-ratio}.}
Line~\ref{alg:reduce-aspect-ratio:eta} requires computing a coarse estimate of the optimal \kzC objective on $P$.
To this end, we invoke the MPC \textsc{$k$-Center} algorithm from~\cite[Lemma 7.2]{CzumajG0J25} to compute an upper estimate $E > 0$ of the optimal \textsc{$k$-Center} objective within a factor of $\poly(n)$.
It is easy to verify that the value $\eta := n \cdot E^z$ approximates the optimal \kzC objective within a factor of $\beta = n^{O(z)}$, i.e., $\OPTclF \le \eta \le \beta \cdot \OPTclF$.

Then, Line~\ref{alg:reduce-aspect-ratio:T} collects the distinct coordinate values, Line~\ref{alg:reduce-aspect-ratio:compute-sort} uses MPC sorting, and Line~\ref{alg:reduce-aspect-ratio:compute-h} can be achieved via an MPC prefix-sum computation (see~\cite{DBLP:conf/isaac/GoodrichSZ11} for both MPC sorting and MPC prefix-sum computation).
To implement Line~\ref{alg:reduce-aspect-ratio:map}, we use sorting to look up the value $h(x)$ for each coordinate $x$ of every point, and then gather the mapped coordinates to form the resulting point. This takes $O(\log_s n)$ rounds.
Finally, the rounding step in Line~\ref{alg:reduce-aspect-ratio:round} can be easily implemented in a single MPC round.

Overall, the complexity of the above implementation is dominated by the MPC \textsc{$k$-Center} algorithm from~\cite[Lemma 7.2]{CzumajG0J25}, which requires local memory $s \ge \poly(d \log n)$, runs in $O(\log_s n)$ rounds, and uses total memory $O(n \cdot \poly(d \log n))$.

\section{Proof of \Cref{lem:correctness-kzC}}
\label{sec:missing-proof-value}

\begin{fact}
\label{obs:lambda0}
    $\OPTfl[\lambda_{0}] = \cost(\By^{(0)}) + \lambda_{0} \cdot \|\By^{(0)}\|_{1}$.
\end{fact}

\begin{proof}
Denote by $\By^{*} \in \R_{\ge 0}^{P}$ an optimal fractional solution to \textsc{Power-$z$ Facility Location} with $\lambda = \lambda_0$; we have 
$\OPTfl[\lambda_{0}]
= \cost(\By^{*}) + \lambda_{0} \cdot \|\By^{*}\|_{1}
\ge (|P| - \|\By^{*}\|_{1}) + \lambda_{0} \cdot \|\By^{*}\|_{1}
= |P|
= \cost(\By^{(0)}) + \lambda_{0} \cdot \|\By^{(0)}\|_{1}$,
where the second step holds since $\min_{p\neq q\in P} \dist(p,q) \ge 1$, and the last step holds since $\By^{(0)}$ is the all-ones vector (see Line~\ref{alg:kzC:y0}).
We also have
$\OPTfl[\lambda_{0}]
\le \cost(\By^{(0)}) + \lambda_{0} \cdot \|\By^{(0)}\|_{1}$ given its optimality.
Combining both equations finishes the proof.
\end{proof}

\begin{fact}\label{obs:lambdaL}
    $\|\By^{(L)}\|_{1} \le 2$.
\end{fact}

\begin{proof}
First, observe that $\OPTfl[\lambda_{L}] \le \lambda_{L} + n\cdot \Delta^{O(z)}$ by considering any solution $\By'$ with $\|\By'\|_{1} = 1$. Then, by \Cref{thm:lmp-fl}, the solution $\By^{(L)}$ satisfies that
\begin{align*}
    \lambda_{L}\cdot \|\By^{(L)}\|_{1}
    ~\le~ \OPTfl[\lambda_{L}]
    ~\le~ \lambda_{L} +  n\Delta^{O(z)}
~\le~ 2\lambda_{L},
\end{align*}
where the last step holds since $\lambda_{L} = 2^{L z}$ and $L = \Theta(\log(n\Delta))$ is sufficiently large. Hence, $\|\By^{(L)}\|_{1} \le 2$, which finishes the proof.
\end{proof}

We are now ready to prove \Cref{lem:correctness-kzC}.

\begin{proof}[Proof of \Cref{lem:correctness-kzC}]
    By \Cref{obs:lambda0,obs:lambdaL}, there exists some $\ell \in [L]$ such that $\|\By^{(\ell-1)}\|_{1} \ge k$ and $\|\By^{(\ell)}\|_{1} \le k$. Consequently, Line~\ref{alg:kzC:istar} selects an arbitrary such index $\ell^{*} \in [L]$, and the output of \Cref{alg:kzC} is $\By := \alpha\cdot \By^{(\ell^{*}-1)} + (1 - \alpha)\cdot \By^{(\ell^{*})}$, where $\alpha \in [0, 1]$ is chosen (in Line~\ref{alg:kzC:alpha}) such that $\|\By\|_{1} = k$.
    Let $\Bx^{(\ell^{*}-1)} \sim \By^{(\ell^{*}-1)}$ and $\Bx^{(\ell^{*})} \sim \By^{(\ell^{*})}$ be the optimal feasible assignments with respect to $\By^{(\ell^{*}-1)}$ and $\By^{(\ell^{*})}$, respectively. It is easy to verify that $\Bx := \alpha\cdot \Bx^{(\ell^{*}-1)} + (1- \alpha)\cdot \Bx^{(\ell^{*})}$ is feasible with respect to $\By = \alpha\cdot \By^{(\ell^{*}-1)} + (1 - \alpha)\cdot \By^{(\ell^{*})}$. As a result,
    \begin{align*}
        \cost(\By) ~&\le~ \asncost(\Bx) ~=~ \sum_{p\in P}\sum_{q\in P} \left(\alpha\cdot \x^{(\ell^{*}-1)}_{p,q} + (1- \alpha)\cdot \x^{(\ell^{*})}_{p,q} \right)\cdot \dist^z(p,q) \\
        &=~ \alpha\cdot \asncost(\Bx^{(\ell^{*}-1)}) + (1-\alpha)\cdot \asncost(\Bx^{(\ell^{*})}) ~=~ \alpha\cdot \cost(\By^{(\ell^{*}-1)}) + (1-\alpha)\cdot \cost(\By^{(\ell^{*})}).
    \end{align*}
    By \Cref{thm:lmp-fl} and \Cref{obs:lambda0}, we have 
    \begin{align*}
        \cost(\By^{(\ell^{*}-1)}) ~&\le~ 2^{O(z^2)} \cdot \epsilon^{-2z} \cdot \left(\OPTfl[\lambda_{\ell^{*}-1}] - \lambda_{\ell^{*}-1} \cdot \|\By^{(\ell^{*}-1)}\|_{1}\right)\\
        &\le~ 2^{O(z^2)} \cdot \epsilon^{-2z} \cdot \left(2^z\cdot \OPTfl[\lambda_{\ell^{*}-1}] - 2^z\cdot \lambda_{\ell^{*}-1} \cdot \|\By^{(\ell^{*}-1)}\|_{1}\right)\\
        &=~ 2^{O(z^2)} \cdot \epsilon^{-2z} \cdot \left(2^z\cdot \OPTfl[\lambda_{\ell^{*}-1}] - \lambda_{\ell^{*}} \cdot \|\By^{(\ell^{*}-1)}\|_{1}\right),
    \end{align*}
    where the last step follows from the fact that $\lambda_{\ell^{*}} = 2^z\lambda_{\ell^{*}-1}$; and 
    \begin{equation*}
        \cost(\By^{(\ell^{*})}) ~\le~ 2^{O(z^2)} \cdot \epsilon^{-2z} \cdot \left(\OPTfl[\lambda_{\ell^{*}}] - \lambda_{\ell^{*}} \cdot \|\By^{(\ell^{*})}\|_{1}\right).
    \end{equation*}
    Therefore,
    \begin{align*}
        \cost(\By) ~&\le~ 2^{O(z^2)} \cdot \epsilon^{-2z} \cdot \left(2^z\alpha\cdot  \OPTfl[\lambda_{\ell^{*}-1}] + (1-\alpha)\cdot \OPTfl[\lambda_{\ell^{*}}]- \lambda_{\ell^{*}} \cdot \|\By\|_{1} \right).
    \end{align*}
    Then, let $\By^{*}\in \R_{\ge 0}^{P}$ with $\|\By^{*}\|_{1} = k$ denote an optimal fractional solution to the \textsc{Euclidean} \kzC problem on input $(P, P, k)$. We have $2^z\alpha\cdot \OPTfl[\lambda_{\ell^{*}-1}] \le 2^z\alpha\cdot ( \lambda_{\ell^{*}-1} \cdot k + \cost(\By^{*})) = \alpha\cdot \lambda_{\ell^{*}} \cdot k + 2^z\alpha\cdot \cost(\By^{*})$ and $(1-\alpha)\cdot \OPTfl[\lambda_{\ell^{*}}] \le (1-\alpha)\cdot \lambda_{\ell^{*}} \cdot k + (1-\alpha)\cdot \cost(\By^{*})$. Hence, 
    \begin{align*}
        \cost(\By) ~&\le~ 2^{O(z^2)} \cdot \epsilon^{-2z} \cdot \left(2^z\alpha\cdot  \OPTfl[\lambda_{\ell^{*}-1}] + (1-\alpha)\cdot \OPTfl[\lambda_{\ell^{*}}]- \lambda_{\ell^{*}} \cdot \|\By\|_{1} \right)\\
        &\le~ 2^{O(z^2)} \cdot \epsilon^{-2z} \cdot \Big( \lambda_{\ell^{*}} \cdot k + (2^z\alpha + 1 - \alpha)\cdot \cost(\By^{*}) - \lambda_{\ell^{*}} \cdot \|\By\|_{1}\Big)\\
        \mr{$\|\By\|_{1} = k$}
        &\le~ 2^{O(z^2)} \cdot \epsilon^{-2z} \cdot \cost(\By^{*}) ~\le~ 2^{O(z^2)} \cdot \epsilon^{-2z} \cdot \OPTclF(P,P,k).
    \end{align*}
    This finishes the proof.
\end{proof}
         \section{Missing Proofs in \Cref{sec:bicriteria}}
\label{sec:missing-proof-bicriteria}

\subsection{Proof of \Cref{claim:generalize-ANS17}}
\label{sec:proof-correctness-coarse-rounding}

\begin{proof}
Throughout this proof, we consider the case where $Q$ is a set rather than a multiset for simplicity; the same argument applies directly to the multiset case, since colocated points can be distinguished by their IDs.
Recall that $\Bx'\sim\By'$ is feasible, $\sum_{q\in Q}x'_{p,q}=1$ for every $p\in P$ and $x'_{p,q}\in\{0,y'_q\}$ for every $p\in P$ and $q\in Q$ by Line~\ref{alg:corase-rounding:xprime}. In particular,
\begin{equation*}
    \forall p\in P,\qquad
    \sum_{q\in Q:x'_{p,q}>0}y'_q
    ~=~\sum_{q\in Q}x'_{p,q}
    ~=~1.
\end{equation*}
Let $Q_+:=\{q\in Q:\exists p\in P,\ x'_{p,q}>0\}$ and, for every $q\in Q_+$, define
\begin{equation*}
    g(q):=\argmin_{p\in P:x'_{p,q}>0}\beta_p.
\end{equation*}
We note that points in $Q\setminus Q_+$ cannot be selected by any $f(p)$ and therefore are never included in $C$.
By Line~\ref{alg:coarse-rounding:return}, a point $q\in Q_+$ is selected if and only if $f(g(q))=q$.

\paragraph{Expected Size.}
For every $p\in P$ and $q\in Q$ with $x'_{p,q}>0$, the exponential clocks have rates $y'_q=1/n^{\epsilon}$, and hence
\begin{equation*}
    \Pr[f(p)=q]
    ~=~\frac{y'_q}{\sum_{q'\in Q:x'_{p,q'}>0}y'_{q'}}
    ~=~y'_q.
\end{equation*}
For every $q\in Q_+$, the value $g(q)$ depends only on the clocks $\{\beta_p\}_{p\in P}$ and is therefore independent of $\{\alpha_{q'}\}_{q'\in Q}$. Consequently,
\begin{align*}
    \Pr[f(g(q))=q]
    ~=~\sum_{p\in P:x'_{p,q}>0}
        \Pr[g(q)=p]\cdot\Pr[f(p)=q]
    ~=~y'_q.
\end{align*}
As a result,
\begin{equation*}
    \E[|C|]
    ~\le~\sum_{q\in Q_+}\Pr[f(g(q))=q]
    ~=~\sum_{q\in Q_+}y'_q
    ~\le~\|\By'\|_1.
\end{equation*}

\paragraph{Expected Cost.}
We follow the analysis of connection paths in~\cite[Section~4.2]{AnNS17} and extend it to the power-$z$ objective.
Let
\begin{equation*}
    E:=\big\{\{p,q\}:p\in P,\ q\in Q_+,\ x'_{p,q}>0\big\}
\end{equation*}
be the edge set of the \emph{undirected} support graph on $P\sqcup Q_+$, where the two vertex classes are kept distinct even when their locations coincide.
Construct a random \emph{directed} graph $G$ on the same vertex set by adding an arc $p\to f(p)$ for every $p\in P$ and an arc $q\to g(q)$ for every $q\in Q_+$. Every vertex in $G$ has exactly one outgoing arc, and the selected points are exactly those belonging to directed cycles of length two.
Moreover, with probability one, all clocks within each vertex class are distinct. On this event, $G$ has no directed cycle of length greater than two (see~\cite[Observation 3.1]{AnNS17}). 

For every $p\in P$, define its \emph{connection path} $\Pi(p)=(e_p^{(1)},\ldots,e_p^{(t_p)})$ by starting at $p$ and following the unique outgoing arc until immediately before a vertex would be visited for the second time.
The last vertex from $Q_+$ on this path belongs to a directed $2$-cycle and hence is included in $C$.
Writing $\dist(e)$ for the distance between the endpoints of an edge or arc $e$, the triangle inequality gives
\begin{equation}\label{eq:ANSz1}
    \dist(p,C)~\le~\sum_{i=1}^{t_p}\dist(e_p^{(i)}).
\end{equation}
Let $\gamma:=5/4$. For $z>1$, H\"older's inequality implies
\begin{align*}
    \dist^z(p,C)
    ~&\le~\Big(\sum_{i=1}^{t_p}
        \gamma^{-i/z}\cdot\gamma^{i/z}\dist(e_p^{(i)})\Big)^z\\
    &\le~\Big(\sum_{i=1}^{t_p}\gamma^{-i/(z-1)}\Big)^{z-1}
        \cdot\sum_{i=1}^{t_p}\gamma^i\dist^z(e_p^{(i)})\\
    &\le~\Big(\frac{1}{\gamma^{1/(z-1)}-1}\Big)^{z-1}
        \cdot\sum_{i=1}^{t_p}\gamma^i\dist^z(e_p^{(i)})\\
    &\le~\Big(\frac{z-1}{\ln\gamma}\Big)^{z-1}
        \cdot\sum_{i=1}^{t_p}\gamma^i\dist^z(e_p^{(i)}).
\end{align*}
Thus, for every integer $z\ge1$,
\begin{equation}\label{eq:ANS-weighted-path}
    \dist^z(p,C)
    ~\le~z^{O(z)}\cdot\sum_{i=1}^{t_p}\gamma^i\dist^z(e_p^{(i)}),
\end{equation}
where the $z = 1$ case follows directly from~\eqref{eq:ANSz1}.

For every $e\in E$, $p\in P$, and integer $i\ge1$, let $\mu_{i,p}(e)$ be the indicator of the event that the $i$-th arc of $\Pi(p)$ traverses $e$ in either direction. Let $w_i(e):=\sum_{p\in P}\mu_{i,p}(e)$.
The connection-path bound in~\cite[Section~4.2]{AnNS17}, applied to the feasible pair $(\Bx',\By')$ with $x'_{p,q}\in\{0,y'_q\}$, gives, for every $e=\{p,q\}\in E$,
\begin{align}
    \E[w_{2t+1}(e)]
    ~&\le~\frac{x'_{p,q}}{2^{\max\{0,t-1\}}},
        &&\forall t\ge0,\label{eq:ANS-odd-position}\\
    \E[w_{2t}(e)]
    ~&\le~\frac{x'_{p,q}}{2^{\max\{0,t-2\}}},
        &&\forall t\ge1.\label{eq:ANS-even-position}
\end{align}
Since $\gamma^2/2<1$, it follows that, for every $e=\{p,q\}\in E$,
\begin{align*}
    \sum_{i\ge1}\gamma^i\E[w_i(e)]
    ~&\le~x'_{p,q}\cdot\Big(
        \sum_{t\ge0}\frac{\gamma^{2t+1}}{2^{\max\{0,t-1\}}}
        +\sum_{t\ge1}\frac{\gamma^{2t}}{2^{\max\{0,t-2\}}}
        \Big)\\
    &=~O(x'_{p,q}).
\end{align*}
Combining this with~\eqref{eq:ANS-weighted-path} and linearity of expectation, we obtain
\begin{align*}
    \E[\cost(C)]
    ~&\le~z^{O(z)}\cdot\sum_{p\in P}\sum_{i\ge1}\sum_{e\in E}
        \gamma^i\dist^z(e)\cdot\E[\mu_{i,p}(e)]\\
    &=~z^{O(z)}\cdot\sum_{e\in E}\dist^z(e)
        \cdot\sum_{i\ge1}\gamma^i\E[w_i(e)]\\
    &\le~O(z^{O(z)})\cdot\sum_{p\in P}\sum_{q\in Q}
        x'_{p,q}\dist^z(p,q)\\
    &=~O(z^{O(z)})\cdot\asncost(\Bx')\\
    &\le~(z/\epsilon)^{O(z)}\cdot\cost(\By'),
\end{align*}
where the last step follows from Line~\ref{alg:corase-rounding:xprime} of \Cref{alg:coarse-rounding}.
Together with the expected size bound, this proves \Cref{claim:generalize-ANS17}.
\end{proof}

\subsection{Proof of \Cref{lem:bi-point-rounding}}
\label{sec:proof-bi-point-rounding}

We describe the bi-point rounding algorithm in \Cref{alg:bi-point-rounding}.

\begin{algorithm}[h]
\DontPrintSemicolon
\caption{Bi-point Rounding on $(C_1,C_2,a)$}
\label{alg:bi-point-rounding}

\tcc{Without loss of generality, assume that $|C_1|\le |C_2|$}

$\kappa \gets \lceil a|C_1| + (1-a)|C_2| \rceil$
\label{alg:bi-point-rounding:kappa}

$C_1'\gets \{C_2(c):c\in C_1\}$, where $C_2(c)\in C_2$ satisfies $\dist(c,C_2(c))\le O(\epsilon^{-1})\cdot \dist(c,C_2)$
\label{alg:bi-point-rounding:C1prime}

add $|C_1|-|C_1'|$ arbitrary distinct points from
$C_2\setminus C_1'$ to $C_1'$
\label{alg:bi-point-rounding:padding}
\tcp*{this ensures $|C_1'| = |C_1|$}

$C_2'\gets$ a uniformly random subset of $C_2\setminus C_1'$ of size $\kappa-|C_1|$
\label{alg:bi-point-rounding:randomsubset}

\Return $C\gets C_1\cup C_2'$ with probability $a$, and $C\gets C_1'\cup C_2'$ with probability $1-a$
\label{alg:bi-point-rounding:return}

\end{algorithm}

\paragraph{Correctness of \Cref{alg:bi-point-rounding}.}
Clearly, the returned set $C\subseteq C_1\cup C_2$ has size at most $\kappa = \lceil a|C_1| + (1-a)|C_2| \rceil \le a|C_1| + (1-a)|C_2| + 1$.
It then remains to analyze the expected cost of $C$.

If $|C_1|=|C_2|$, then Line~\ref{alg:bi-point-rounding:padding} ensures that $C_1'=C_2$, and therefore $C_2'=\emptyset$. 
The algorithm returns $C_1$
with probability $a$ and $C_2$ with probability $1-a$.
Consequently, $\E[\cost(C)]=a\cost(C_1)+(1-a)\cost(C_2)$, as desired.

We henceforth assume that $|C_1|<|C_2|$. Fix an arbitrary point $p\in P$. 
For $j\in\{1,2\}$,
choose $c_j\in C_j$ such that
$\dist(p,c_j)=\dist(p,C_j)$.
We distinguish two cases according to whether $c_2$
belongs to $C_1'$.

Case 1: $c_2\in C_1'$. In this case, if the algorithm returns $C\gets C_1\cup C_2'$, which happens with probability $a$, then $\dist(p,C)\le \dist(p,C_1)$. Otherwise, the algorithm returns $C\gets C_1'\cup C_2'$, which happens with probability $1-a$, and since $c_2\in C_1'\subseteq C$, we have $\dist(p,C)\le \dist(p,c_2)=\dist(p,C_2)$. Therefore, $\E[\dist^z(p, C)] \le a\cdot \dist^z(p, C_1) + (1-a)\cdot \dist^z(p, C_2)$.
    
Case 2: $c_2\in C_2\setminus C_1'$. 
    In this case, we first prove the following inequality
    \begin{equation}\label{eq:pC1prime}
        \dist(p, C_1') ~\le~ O(\epsilon^{-1})\cdot \big( \dist(p, C_1) +  \dist(p, C_2)\big).
    \end{equation}
    To see this, recall that $C_2(c_1)\in C_1'$ (Line~\ref{alg:bi-point-rounding:C1prime}); we then have 
    \begin{align*}
        \dist(p, C_1') ~\le~ \dist(p, C_2(c_1)) ~&\le~ \dist(p, c_1) + \dist(c_1, C_2(c_1))\\
        \mr{Line~\ref{alg:bi-point-rounding:C1prime}}
        &\le~ \dist(p, c_1) + O(\epsilon^{-1})\cdot \dist(c_1, C_2) \\
        &\le~ (1+O(\epsilon^{-1}))\cdot \dist(p, c_1) + O(\epsilon^{-1})\cdot \dist(p, C_2)\\
        &=~ O(\epsilon^{-1})\cdot \big( \dist(p, C_1) +  \dist(p, C_2)\big),
    \end{align*}
    which proves \eqref{eq:pC1prime}.
    Then, we further distinguish the following two subcases: 
    
    (i) If $\dist(p, C_1) \le \dist(p, C_2)$, then 
    \begin{align*}
        \E[\dist^z(p, C)] ~&=~ a\cdot \E[\dist^z(p, C_1\cup C_2')] + (1-a)\cdot \E[\dist^z(p, C_1'\cup C_2')]\\
        &\le~ a\cdot \dist^z(p, C_1) + (1-a)\cdot \dist^z(p, C_1')\\
        \mr{\eqref{eq:pC1prime}}
        &\le~ a\cdot \dist^z(p, C_1) + (1-a)\cdot \epsilon^{-O(z)}\cdot \big(\dist(p, C_1) + \dist(p, C_2) \big)^z\\
&\le~ \epsilon^{-O(z)}\cdot \big(a\cdot \dist^z(p, C_1) + (1-a)\cdot \dist^z(p, C_2) \big).
    \end{align*}
    
    (ii) If $\dist(p, C_1) > \dist(p, C_2)$, assume for now that $\Pr[c_2\in C_2']\in(0,1)$; the boundary cases are discussed later. Then, we have 
    \begin{equation*}
        \E[\dist^z(p, C)] ~=~ \E[\dist^z(p, C) ~|~ c_2\in C_2']\Pr[c_2\in C_2'] + \E[\dist^z(p, C) ~|~ c_2\notin C_2']\Pr[c_2\notin C_2'].
    \end{equation*}
    We analyze the two terms separately.

    For the first term, since it must hold that $C_2'\subseteq C$ (Line~\ref{alg:bi-point-rounding:return}), we have 
    \begin{equation*}
        \E[\dist^z(p, C) ~|~ c_2\in C_2'] ~\le~ \dist^z(p, c_2) = \dist^z(p, C_2).
    \end{equation*}
    For the second term, since the choice between $C=C_1\cup C_2'$ and $C=C_1'\cup C_2'$ (Line~\ref{alg:bi-point-rounding:return}) is independent of the event $c_2\notin C_2'$, we have
    \begin{align*}
        \E[\dist^z(p, C) ~|~ c_2\notin C_2'] ~&=~ \E[a\cdot \dist^z(p, C_1\cup C_2') + (1-a)\cdot \dist^z(p, C_1'\cup C_2') ~|~ c_2\notin C_2']\\
        &\le~ a\cdot \dist^z(p, C_1) + (1-a)\cdot \dist^z(p, C_1')\\
        \mr{\eqref{eq:pC1prime}}
        &\le~ a\cdot \dist^z(p, C_1) + (1-a)\cdot \epsilon^{-O(z)}\cdot \dist^z(p, C_1)\\
        &\le~ \epsilon^{-O(z)}\cdot \dist^z(p, C_1).
    \end{align*}
    As a result, 
    \begin{equation*}
        \E[\dist^z(p, C)] ~\le~ \epsilon^{-O(z)}\cdot \big(\dist^z(p, C_2)\cdot \Pr[c_2\in C_2'] + \dist^z(p, C_1)\cdot \Pr[c_2\notin C_2'] \big).
    \end{equation*}
Since $C_2'$ is a uniformly random subset of $C_2\setminus C_1'$ of size $\kappa - |C_1|$, we have 
    \begin{equation*}
        t~:=~ \Pr[c_2\in C_2'] ~=~ \frac{\kappa - |C_1|}{|C_2\setminus C_1'|} ~=~ \frac{\kappa - |C_1|}{|C_2| - |C_1'|} ~\ge~ 1 - a,
    \end{equation*}
    where we are using the fact that $C_1'\subseteq C_2$, $|C_1| = |C_1'|$ and $\kappa = \lceil a|C_1| + (1-a)|C_2| \rceil \ge a|C_1| + (1-a)|C_2|$. Finally, since $\dist(p, C_1) > \dist(p, C_2)$, we have 
    \begin{align*}
        \E[\dist^z(p, C)] ~&\le~ \epsilon^{-O(z)}\cdot \big(\dist^z(p, C_2)\cdot t + \dist^z(p, C_1)\cdot (1-t) \big)\\
        &\le~ \epsilon^{-O(z)}\cdot \big(\dist^z(p, C_2)\cdot (1-a) + \dist^z(p, C_1)\cdot a).
    \end{align*}

    In all above cases, we have $\E[\dist^z(p, C)] ~\le~ \epsilon^{-O(z)}\cdot (a\cdot \dist^z(p, C_1) + (1-a)\cdot \dist^z(p, C_2))$. 
    Summing over all $p\in P$ gives 
    \begin{equation*}
        \E[\cost(C)] ~\le~ \epsilon^{-O(z)}\cdot (a\cdot \cost(C_1) + (1-a)\cdot \cost(C_2)).
    \end{equation*}

    Finally, we discuss the boundary cases where $\Pr[c_2\in C_2']=0$ or $\Pr[c_2\in C_2']=1$.
    If $\Pr[c_2\in C_2']=0$, then $a=1$ since $\Pr[c_2\in C_2']\ge 1-a$, and the desired bound follows from the analysis of the second term above. If $\Pr[c_2\in C_2']=1$, the desired bound follows directly from the analysis of the first term above and $\dist(p,C_2)<\dist(p,C_1)$.
    Therefore, we complete the proof of the correctness of \Cref{alg:bi-point-rounding}.

    \paragraph{MPC Implementation.}
    Line~\ref{alg:bi-point-rounding:kappa} can be implemented in $O(\log_s n)$ rounds by aggregating $|C_1|$ and $|C_2|$ and broadcasting $\kappa$. Line~\ref{alg:bi-point-rounding:C1prime} amounts to an approximate nearest-neighbor search, for which we can directly use~\cite[Section~3.1]{CzumajGJK024}; this takes $O(\log_s n)$ rounds and uses $2^{O(\epsilon d)}\cdot n\cdot \poly(d\log(n\Delta))$ total memory.
    Line~\ref{alg:bi-point-rounding:padding} can be implemented straightforwardly.
    Line~\ref{alg:bi-point-rounding:randomsubset} can be implemented by generating a random permutation of $C_2\setminus C_1'$ and selecting the first $\kappa-|C_1|$ points as $C_2'$. Finally, Line~\ref{alg:bi-point-rounding:return} can be implemented straightforwardly.
    Overall, \Cref{alg:bi-point-rounding} can be implemented in $O(\log_s n)$ rounds using $2^{O(\epsilon d)}\cdot n\cdot \poly(d\log(n\Delta))$ total memory.         

\section{Proof of \Cref{claim:guarantee-new-instance}}
\label{sec:guarantee-new-instance}

The first property in \Cref{claim:guarantee-new-instance} follows directly from the construction in \Cref{alg:new-instance}; specifically,
\begin{equation*}
    \Fnew = T_1\cup T_2 \subseteq \Fold, \quad r+1\le |\Fnew| = |T_1| + |T_2| \le 12r.
\end{equation*}
Here, the lower bound follows from $|T_1|=\min\{6r,|\Fold|\}\ge r+1$, since $1\le r<|\Fold|$.
If $T_1=\Fold$, then $\Fnew=\Fold$ and $\Pnew=\Pold$, so the second property holds immediately. Hence, assume $T_1\neq\Fold$ for the remainder of the proof.
It remains to prove the second property, for which we need the following claim.

\begin{claim}
    \label{claim:OPTF-Fp}
    It holds that
    \begin{equation}\label{eq:OPTF-Fp}
        \OPTclF(\Pnew,\Fnew,|\Fnew|-r)
        ~\le~ \epsilon^{-O(z)}\cdot \Big(\sum_{p\in \Pold}\dist^z(p, \pi(p)) + \OPTclF(\Pold,\Fold,|\Fold|-r)\Big).
    \end{equation}
\end{claim}
\begin{proof}
    The proof can be found in \Cref{sec:OPTF-Fp}.
\end{proof}

    Fix an arbitrary set $R\subseteq \Fnew$ of size $r$, and assume that, for some $\alpha \ge 1$,
    \begin{equation*}
        \cost(\Pnew, \Fnew\setminus R) \le \alpha\cdot \OPTclF(\Pnew,\Fnew,|\Fnew|-r).
    \end{equation*}
    The goal is to prove
    \begin{equation}\label{eq:guarantee-new-instance-goal}
        \cost(\Pold, \Fold\setminus R) ~\le~ \epsilon^{-O(z)}\cdot \alpha\cdot \OPTclF(\Pold,\Fold,|\Fold|-r).
    \end{equation}
    For every $p\in \Pnew$, since $\Fnew\subseteq \Fold$ by \Cref{alg:new-instance}, we have $\Fnew\setminus R\subseteq \Fold\setminus R$, and hence
    \begin{equation*}
        \dist(p,\Fold\setminus R) ~\le~ \dist(p,\Fnew\setminus R).
    \end{equation*}
    On the other hand, for every $p\in \Pold\setminus \Pnew$, Line~\ref{alg:new-instance:return} implies that $\pi(p)\notin \Fnew$. Since $\pi(p)\in \Fold$ and $R\subseteq \Fnew$, we have $\pi(p)\in \Fold\setminus R$, and therefore
    \begin{equation*}
        \dist(p,\Fold\setminus R) ~\le~ \dist(p,\pi(p)).
    \end{equation*}
    As a result, we have
    \begin{align*}
        \cost(\Pold, \Fold\setminus R) ~&\le~ \cost(\Pnew, \Fnew\setminus R) + \sum_{p\in \Pold\setminus \Pnew} \dist^z(p, \pi(p))\\
        &\le~ \alpha\cdot \OPTclF(\Pnew,\Fnew,|\Fnew|-r) + \sum_{p\in \Pold\setminus \Pnew} \dist^z(p, \pi(p))\\
        \mr{\Cref{claim:OPTF-Fp}}
        &\le~ \epsilon^{-O(z)}\cdot \alpha\cdot \Big(\sum_{p\in \Pold}\dist^z(p, \pi(p)) + \OPTclF(\Pold,\Fold,|\Fold|-r)\Big).
    \end{align*}
    Finally, by the construction of $\pi$ in Line~\ref{alg:new-instance:pi} and the definition of $\OPTclF(\Pold,\Fold,|\Fold|-r)$, we have
    \begin{equation*}
        \sum_{p\in \Pold}\dist^z(p,\pi(p))
        ~\le~
        \epsilon^{-O(z)}\cdot \cost(\Pold,\Fold)
        ~\le~
        \epsilon^{-O(z)}\cdot \OPTclF(\Pold,\Fold,|\Fold|-r).
    \end{equation*}
    Combining everything together, we obtain \eqref{eq:guarantee-new-instance-goal}, which finishes the proof of the second property in \Cref{claim:guarantee-new-instance}.

\subsection{Proof of \Cref{claim:OPTF-Fp}.} 
\label{sec:OPTF-Fp}

For a point $p$ and a nonempty set $Q$, we define $\mathrm{NN}(p, Q):= \argmin_{q\in Q} \dist(p, q)$, with ties broken consistently.
To prove \Cref{claim:OPTF-Fp}, we need the following claim.
\begin{claim}\label{claim:HatR}
    There exists a subset $\Hat{R}\subseteq T_1$ of size $r$, such that
    \begin{equation}\label{eq:HatR}
        \sum_{c\in \Fold} w(c)\cdot \dist^z(c, \Fold\setminus \Hat{R}) ~\le~ \epsilon^{-O(z)}\cdot \min_{\substack{R\subseteq \Fold\\ |R|=r}} \sum_{c\in \Fold} w(c)\cdot \dist^z(c, \Fold\setminus R).
    \end{equation}
\end{claim}
\begin{proof}
    The proof can be found in \Cref{sec:HatR}.
\end{proof}

Let $\Hat{R}\subseteq T_1\subseteq \Fnew\subseteq \Fold$ be a subset of size $r$ as guaranteed by \Cref{claim:HatR}, satisfying~\eqref{eq:HatR}.
We first have 
\begin{equation}\label{eq:disFp}
    \forall c\in \Fnew, \quad \dist(c, \Fnew\setminus \Hat{R}) ~\le~ O(\epsilon^{-1})\cdot \dist(c, \Fold\setminus \Hat{R}).
\end{equation}
To see this, consider a fixed $c\in \Fnew$. If $c\notin \Hat{R}$, then \eqref{eq:disFp} holds immediately since $\dist(c, \Fnew\setminus \Hat{R}) = 0$. Otherwise, $c\in \Hat{R} \subseteq T_1$. 
Suppose that $c=c_i$, for some $i\in [6r]$, in the ordering of $\Fold$ computed in Line~\ref{alg:new-instance:sort}. 
Recall that, in Line~\ref{alg:new-instance:T2}, we compute a point $s_i\in \Fold\setminus T_1$ satisfying
$\dist(c_i,s_i)\le O(\epsilon^{-1})\cdot \dist(c_i,\Fold\setminus T_1)$,
and include $s_i$ in $T_2$.
Combining the facts that $\Fnew = T_1\cup T_2$ (Line~\ref{alg:new-instance:return}), $T_1\cap T_2 = \emptyset$ (Line~\ref{alg:new-instance:T2}) and $\Hat{R}\subseteq T_1$ (\Cref{claim:HatR}), we have $s_i\in \Fnew\setminus \Hat{R}$.
Therefore,
\begin{equation*}
    \dist(c_i,\Fnew\setminus\Hat{R})
    ~\le~
    \dist(c_i,s_i)
    ~\le~
    O(\epsilon^{-1})\cdot \dist(c_i,\Fold\setminus T_1).
\end{equation*}
Observe that $\mathrm{NN}(c_i,\Fold\setminus \Hat{R})\in \Fold\setminus \Hat{R}=(\Fold\setminus T_1)\cup(T_1\setminus \Hat{R})$, where $\Fold\setminus T_1$ and $T_1\setminus \Hat{R}$ are disjoint. We distinguish the following two cases.
\begin{itemize}
    \item If $\mathrm{NN}(c_i, \Fold\setminus \Hat{R}) \in \Fold\setminus T_1$, then $\dist(c_i, \Fold\setminus T_1) \le \dist(c_i, \mathrm{NN}(c_i, \Fold\setminus \Hat{R})) = \dist(c_i, \Fold\setminus \Hat{R})$. Therefore
    \begin{equation*}
        \dist(c_i, \Fnew\setminus \Hat{R}) ~\le~ O(\epsilon^{-1})\cdot \dist(c_i, \Fold\setminus \Hat{R}).
    \end{equation*}
    \item If $\mathrm{NN}(c_i, \Fold\setminus \Hat{R})\in T_1\setminus \Hat{R}$, since $(T_1\setminus \Hat{R}) \subseteq (\Fnew\setminus \Hat{R})$, we directly have
    \begin{equation*}
        \dist(c_i, \Fnew\setminus \Hat{R}) ~\le~ \dist(c_i, \mathrm{NN}(c_i, \Fold\setminus \Hat{R})) ~=~ \dist(c_i, \Fold\setminus \Hat{R}).
    \end{equation*}
\end{itemize}
Overall, \eqref{eq:disFp} holds. We use \eqref{eq:disFp} to prove \eqref{eq:OPTF-Fp}, as follows: Since $|\Fnew\setminus \Hat{R}| = |\Fnew| - |\Hat{R}| = |\Fnew| - r$, we have
\begin{align*}
    &\OPTclF(\Pnew,\Fnew,|\Fnew|-r) \\
    \le~ &\cost(\Pnew, \Fnew\setminus \Hat{R})\\
    \le~ &\sum_{p\in \Pnew} \big(\dist(p, \pi(p)) + \dist(\pi(p), \Fnew\setminus \Hat{R}) \big)^z\\
    \le~ &2^{z-1}\cdot \sum_{p\in \Pold} \dist^z(p, \pi(p)) + 2^{z-1}\cdot \sum_{c\in \Fnew} w(c)\cdot \dist^z(c, \Fnew\setminus \Hat{R})\\
    \le~ &2^{z-1}\cdot \sum_{p\in \Pold} \dist^z(p, \pi(p)) + \epsilon^{-O(z)}\cdot \sum_{c\in \Fnew} w(c)\cdot \dist^z(c, \Fold\setminus \Hat{R})\\
    \le~  &2^{z-1}\cdot \sum_{p\in \Pold} \dist^z(p, \pi(p)) + \epsilon^{-O(z)}\cdot \min_{\substack{R\subseteq \Fold\\ |R|=r}} \sum_{c\in \Fold} w(c)\cdot \dist^z(c, \Fold\setminus R),
\end{align*}
where the third step follows from the construction $\Pnew=\{p\in \Pold:\pi(p)\in \Fnew\}$ in Line~\ref{alg:new-instance:return} and the definition of $w$ in Line~\ref{alg:new-instance:w}; the fourth step follows from \eqref{eq:disFp}; and the last step follows from \Cref{claim:HatR}.

It remains to upper-bound $\min_{\substack{R\subseteq \Fold\\ |R|=r}} \sum_{c\in \Fold} w(c)\cdot \dist^z(c, \Fold\setminus R)$ as follows: 
\begin{align*}
    &\min_{\substack{R\subseteq \Fold\\ |R|=r}} \sum_{c\in \Fold} w(c)\cdot \dist^z(c, \Fold\setminus R)\\
    &\le~ \min_{\substack{R\subseteq \Fold\\ |R|=r}} \sum_{p\in \Pold} \dist^z(\pi(p), \Fold\setminus R) \\
    &\le~ \min_{\substack{R\subseteq \Fold\\ |R|=r}} \sum_{p\in \Pold}2^{z-1}\cdot \Big(\dist^z(p, \pi(p)) + \dist^z(p, \Fold\setminus R) \Big)\\
    &=~ 2^{z-1}\cdot \sum_{p\in \Pold} \dist^z(p, \pi(p)) + 2^{z-1}\cdot \min_{\substack{R\subseteq \Fold\\ |R|=r}} \cost(\Pold, \Fold\setminus R) \\
    &=~ 2^{z-1}\cdot \sum_{p\in \Pold} \dist^z(p, \pi(p)) + 2^{z-1}\cdot \OPTclF(\Pold,\Fold,|\Fold|-r),
\end{align*}
where the first step follows from the definition of $w$ (Line~\ref{alg:new-instance:w}), and the last step follows from the definition of $\OPTclF(\Pold,\Fold,|\Fold|-r)$. Combining everything together, we obtain \eqref{eq:OPTF-Fp}.

\subsection{Proof of \Cref{claim:HatR}}
\label{sec:HatR}

Let $R^{\star}\subseteq \Fold$ be a set of size $r$ that minimizes
$\sum_{c\in \Fold} w(c)\cdot \dist^z(c, \Fold\setminus R^{\star})$.

Since $T_1\neq \Fold$ as we assume in \Cref{sec:guarantee-new-instance}, we have $|T_1|=6r$.
Our proof strategy is to construct an auxiliary set $U\subseteq T_1$ of size at least $r$ such that the centers in $R^{\star}\setminus T_1$ can be replaced, one-for-one, by centers in $U$ while increasing the cost by only a bounded factor.

Below, we construct the set $U$, together with a mapping $h:U\to \Fold$, and analyze their properties; we then use these properties to complete the proof of \Cref{eq:HatR}.

\paragraph{Construction of $U$ and $h$.}
We construct the auxiliary set $U\subseteq T_1$ and a mapping $h$ as follows:
\begin{enumerate}
    \item $U_1\gets T_1\setminus R^{\star}$.

    \item $U_2\gets U_1\setminus
    \{\mathrm{NN}(c,\Fold\setminus R^{\star}):c\in R^{\star}\}$.

    \item $U_3\gets U_2\setminus
    \{\mathrm{NN}(c,U_2(c)):c\in R^{\star},\,U_2(c)\neq\emptyset\}$,
    where $U_2(c):=\{u\in U_2:
    \mathrm{NN}(u,\Fold\setminus\{u\})=c\}$.

    \item For every $u\in U_3$, letting
    $c:=\mathrm{NN}(u,\Fold\setminus\{u\})$, define
    $h(u):=c$ if $c\notin R^{\star}$, and
    $h(u):=\mathrm{NN}(c,U_2(c))$ otherwise.
    In the latter case, $U_2(c)\neq\emptyset$ since $u\in U_2(c)$.

    Observe that $h(u)\neq u$. Indeed, if $h(u)=c=\mathrm{NN}(u,\Fold\setminus\{u\})$, then $h(u)\neq u$ follows directly from the definition of $\mathrm{NN}(u,\Fold\setminus\{u\})$.
    Otherwise, $h(u)=\mathrm{NN}(c,U_2(c))$. By the third step, $\mathrm{NN}(c,U_2(c))\notin U_3$, while $u\in U_3$. It therefore follows that $h(u)\neq u$.

    \item Consider the undirected graph $G$ on $U_3$ containing an edge
    $\{u,h(u)\}$ for every $u\in U_3$ with $h(u)\in U_3$.
    Since every vertex contributes at most one edge, every induced subgraph has at most as many edges as vertices and hence contains a vertex of degree at most $2$. Therefore, the graph $G$ is $2$-degenerate and thus admits a proper $3$-coloring.
    
    Finally, take an arbitrary $3$-coloring and let $U$ be the largest color class of $G$.
\end{enumerate}
Clearly, we have $U\subseteq U_3\subseteq U_2 \subseteq U_1 \subseteq T_1$.

\paragraph{Properties of $U$ and $h$.}
We prove the following properties of $U$ and the mapping $h$.
\begin{enumerate}[label=(\alph*)]
    \item \label{prop:U-a} $|U|\ge r$. To see this,
    observe  that at most $3|R^{\star}|=3r$ points are removed from $T_1$
    in the first three steps. Since $|T_1|=6r$, the set $U_3$ contains
    at least $3r$ points before the last step. By the last step, $U$ contains
    at least one third of the points in $U_3$, so the final set satisfies
    $|U|\ge r$.
    
    \item\label{prop:U-b} $U\cap R^{\star} = \emptyset$. This property holds for the set $U_1$ constructed in the first step. Since $U\subseteq U_1$, it also holds for $U$.
    
    \item\label{prop:U-c} $\mathrm{NN}(c,\Fold\setminus R^{\star})\notin U$, $\forall c\in R^{\star}$. 
    This property holds for the set $U_2$ constructed in the second step. Since $U\subseteq U_2$, it also holds for $U$.

    \item\label{prop:U-d}  $h(u)\in \Fold\setminus (R^{\star}\cup U)$, $\forall u\in U$. To see this, let $v:=\mathrm{NN}(u,\Fold\setminus\{u\})$.
    By construction, if $v\notin R^{\star}$, then $h(u)=v\notin R^{\star}$;
    otherwise,
    $h(u)=\mathrm{NN}(v,U_2(v))\in U_2\subseteq T_1\setminus R^{\star}$.
    Thus, in both cases, $h(u)\notin R^{\star}$.
    Moreover, if $h(u)\in U_3$, then $\{u,h(u)\}$ is an edge of the graph $G$. Since $U$ is a color class of a proper coloring of $G$, $u$ and $h(u)$ cannot both belong to $U$.
    Hence, $h(u)\notin U$, as required.

    \item\label{prop:U-e} $\dist(u,h(u))\le2\cdot\dist(u,\Fold\setminus\{u\})$, $\forall u\in U$.  Again, let $v:=\mathrm{NN}(u,\Fold\setminus\{u\})$.
    If $v\notin R^{\star}$, then $h(u)=v$, and hence
    $\dist(u,h(u))=\dist(u,\Fold\setminus\{u\})$.
    Otherwise, since $u\in U_2(v)$ and
    $h(u)=\mathrm{NN}(v,U_2(v))$, we have
    $\dist(v,h(u))\le\dist(v,u)$.
    By the triangle inequality,
    $\dist(u,h(u))
    \le\dist(u,v)+\dist(v,h(u))
    \le2\cdot\dist(u,v)
    =2\cdot\dist(u,\Fold\setminus\{u\})$.
\end{enumerate}

\paragraph{Completing the Proof of \Cref{claim:HatR}.}
We now return to the proof of \Cref{claim:HatR}, using the auxiliary set $U$ and the mapping $h$ constructed above.

Let $R_1:=R^{\star}\cap T_1$ and $R_2:=R^{\star}\setminus T_1$. Choose an arbitrary subset $R_2'\subseteq U\subseteq T_1$ of size $|R_2|$, and let $\Hat{R}:=R_1\cup R_2'$. 
Such a subset $R_2'$ must exist since $|R_2|\le r \le |U|$ by Property~\ref{prop:U-a}.
Clearly, $\Hat{R}\subseteq T_1$. Moreover, since $U\cap R^{\star}=\emptyset$ by Property~\ref{prop:U-b}, $R_1\subseteq R^{\star}$ and $R_2'\subseteq U$, we have $R_1 \cap R_2' = \emptyset$, and therefore
$|\Hat{R}|=|R_1|+|R_2'|=|R_1|+|R_2|=r$.
Observe that
\begin{align*}
    \sum_{c\in \Fold} w(c)\cdot \dist^z(c, \Fold\setminus \Hat{R}) ~&=~ \sum_{c\in R_1} w(c)\cdot \dist^z(c, \Fold\setminus \Hat{R}) + \sum_{c\in R_2'} w(c)\cdot \dist^z(c, \Fold\setminus \Hat{R})\\
    \sum_{c\in \Fold} w(c)\cdot \dist^z(c, \Fold\setminus R^{\star}) ~&=~ \sum_{c\in R_1} w(c)\cdot \dist^z(c, \Fold\setminus R^{\star}) + \sum_{c\in R_2} w(c)\cdot \dist^z(c, \Fold\setminus R^{\star}).
\end{align*}
To prove \Cref{eq:HatR}, it suffices to prove 
\begin{align}
    \sum_{c\in R_1} w(c)\cdot \dist^z(c, \Fold\setminus \Hat{R}) ~&\le~ \sum_{c\in R_1} w(c)\cdot \dist^z(c, \Fold\setminus R^{\star}), \label{eq:R1}\\
    \sum_{c\in R_2'} w(c)\cdot \dist^z(c, \Fold\setminus \Hat{R}) ~&\le~ \epsilon^{-O(z)}\cdot \sum_{c\in R_2} w(c)\cdot \dist^z(c, \Fold\setminus R^{\star}).\label{eq:R2}
\end{align}
We prove these two inequalities separately below.

\paragraph{Proof of \eqref{eq:R1}.} 
For every $c\in R_1$, by definition,
$\mathrm{NN}(c,\Fold\setminus R^{\star})\notin R^{\star}$.
Moreover, by Property~\ref{prop:U-c},
$\mathrm{NN}(c,\Fold\setminus R^{\star})\notin U$.
Since $\Hat{R}\subseteq R^{\star}\cup U$, it follows that
$\mathrm{NN}(c,\Fold\setminus R^{\star})\notin\Hat{R}$, and hence
$\mathrm{NN}(c,\Fold\setminus R^{\star})\in \Fold\setminus\Hat{R}$.
Therefore,
\begin{equation*}
    \dist(c, \Fold\setminus \Hat{R}) ~\le~ \dist(c, \mathrm{NN}(c, \Fold\setminus R^{\star})) ~=~ \dist(c, \Fold\setminus R^{\star}).
\end{equation*}
Taking the $z$-th power, multiplying by $w(c)$, and summing over all $c\in R_1$ proves \eqref{eq:R1}.

\paragraph{Proof of \eqref{eq:R2}.}
Let $\psi: R_2' \to R_2$ be an arbitrary bijection (noting that $|R_2'| = |R_2|$).
For every $c\in R_2' \subseteq U$, we have $h(c) \in \Fold\setminus (R^{\star}\cup U) \subseteq \Fold\setminus \Hat{R}$ (since $\Hat{R} \subseteq R^{\star} \cup U$) by Property~\ref{prop:U-d}, and $\dist(c, h(c)) \le 2\cdot \dist(c, \Fold\setminus \{c\})$ by Property~\ref{prop:U-e}. Therefore, 
\begin{equation*}
    \dist(c, \Fold\setminus \Hat{R}) ~\le~ \dist(c, h(c)) ~\le~ 2\cdot \dist(c, \Fold\setminus \{c\}).
\end{equation*}
Moreover, since $c\in R_2'\subseteq U\subseteq T_1$ and $\psi(c)\in R_2 = R^{\star}\setminus T_1$, by the construction of $T_1$ in Line~\ref{alg:new-instance:T1}, we have 
\begin{equation*}
    w(c)\cdot \tau^z(c) ~\le~ w(\psi(c))\cdot \tau^z(\psi(c)).
\end{equation*} 
By the definition of $\tau$ in Line~\ref{alg:new-instance:tau}, we then have 
\begin{align*}
    w(c)\cdot \dist^z(c, \Fold\setminus \Hat{R}) ~&\le~ w(c)\cdot 2^z\cdot \dist^z(c, \Fold\setminus \{c\})\\
    &\le~ w(c)\cdot 2^z\cdot \tau^z(c)\\
    &\le~ w(\psi(c))\cdot 2^z\cdot \tau^z(\psi(c))\\
    &\le~ w(\psi(c))\cdot \epsilon^{-O(z)}\cdot \dist^z(\psi(c), \Fold\setminus \{\psi(c)\})\\
    &\le~ w(\psi(c))\cdot \epsilon^{-O(z)}\cdot \dist^z(\psi(c), \Fold\setminus R^{\star}),
\end{align*}
where the last step follows from $\psi(c) \in R_2 \subseteq R^{\star}$.
Summing the above inequality over all $c\in R_2'$ and using the fact that $\psi$ is a bijection proves \eqref{eq:R2}.

\end{appendices}

    \section*{AI Methodology}
    \addcontentsline{toc}{section}{AI Methodology}
    We used ChatGPT to assist with algorithm development. In particular, ChatGPT suggested some concrete steps that helped
    implement the iterative kernelization framework: considering a
    rounding algorithm of~\cite{AnNS17} and converting its expected-size
    guarantee into a high-probability bicriteria guarantee (see \Cref{alg:coarse-rounding,alg:strengthen-size}).
    However, all other major ideas, especially the entire algorithmic framework, the variant of the MP algorithm, and the overall iterative kernelization, were developed by the authors.
The authors verified the correctness and originality of all content, including references.

    \addcontentsline{toc}{section}{References}
    \bibliography{ref.bib}
    \bibliographystyle{alphaurl}
\end{document}